\PassOptionsToPackage{bookmarks, colorlinks=true, plainpages=false, citecolor=darkblue,
  linkcolor=chicago-maroon, anchorcolor=red, urlcolor=chicago-maroon}{hyperref}
\documentclass[Times1COL]{WileyNJDv5} 

\articletype{Research Article}%

\received{Date Month Year}
\revised{Date Month Year}
\accepted{Date Month Year}
\journal{Journal}
\volume{00}
\copyyear{2025}
\startpage{1}

\usepackage{bm}
\usepackage{mathtools}
\usepackage{amssymb}
\usepackage{amsmath}
\usepackage{bbm}  
\usepackage{graphicx}
\usepackage{orcidlink}
\usepackage{threeparttable}
\usepackage{subcaption} 
\usepackage{xcolor} 
\usepackage{float}  
\usepackage[figuresright]{rotating}

\usepackage{enumitem}

\newcommand{\Rev}[1]{{\color{red}{#1}}}

\def\adj{\normalfont\textrm{adj}}
\def\adjcal{\normalfont\textrm{adj,2}\mbox{-}\textrm{cal}}
\def\db{\normalfont\textrm{db}}
\def\unadj{\normalfont\textrm{unadj}}

\def\gob{\normalfont\textrm{gOB}}
\def\gbcal{\normalfont\textrm{gOB\mbox{-}cal}}
\def\cal{\mathrm{cal}}
\def\loo{\normalfont\textrm{loo}}
\def\loocal{\normalfont\textrm{loo}\mbox{-}\textrm{cal}}
\def\loocalout{\normalfont\textrm{loo}\mbox{-}\textrm{cal}\mbox{-}\textrm{out}}

\def\argmin{{\arg\hspace*{-0.5mm}\min}}
\def\JASA{{\normalfont\textrm{JASA}}}
\def\JASACal{{\normalfont\textrm{JASA\mbox{-}cal}}}

\def\OLS{\normalfont\textrm{OLS}}
\def\formbeta{\textbf{formulation~($\bbeta$)}}
\def\formmu{\textbf{formulation~($\bm{\mu}$)}}
\def\bdelta{\bm{\delta}}
\def\bx{\mathbf{x}}
\def\bX{\mathbf{X}}
\def\diff{\mathrm{d}}
\def\bdelta{\bm{\delta}}
\def\bbR{\mathbb{R}}
\def\bbZ{\mathbb{Z}}
\def\bbX{\mathbb{X}}
\def\bbH{\mathbb{H}}
\def\bbP{\mathbb{P}}
\def\bbE{\mathbb{E}}
\def\var{\mathrm{var}}

\def\bZ{\mathbf{Z}}
\def\bSigma{\mathbf{\Sigma}}
\def\bT{\mathbf{T}}
\def\bY{\mathbf{Y}}
\def\bO{\mathbf{O}}

\def\bbeta{\bm{\beta}}
\def\bmX{\bm{X}}

\usepackage{xcolor}
\usepackage[numbers,super,sort&compress]{natbib}
\setcitestyle{comma,numbers,super,open={(},close={)}} 
\makeatletter 
\renewcommand\@biblabel[1]{#1.} 
\makeatother

\definecolor{chicago-maroon}{RGB}{128,0,0}
\definecolor{darkblue}{rgb}{0.0, 0.0, 0.55}

\begin{document}

\title{Adjusting for Many Covariates in Randomized Clinical Trials with GLMs: Bias Reduction by Jackknife and Practical Guidance}


\author[1]{{Sihui Zhao\orcidlink{0009-0008-7867-0991}$^{\ast}$}}

\author[2,3]{{Xinbo Wang\orcidlink{0009-0008-2917-1563}$^{\dagger}$}}

\author{Yongyong Ren\orcidlink{0000-0001-9217-3483}$^{\textbf{3,4}}$} 

\author{Hongyu Zhao\orcidlink{0000-0003-1195-9607}$^{\textbf{5}}$} 

\author{Hui Lu\orcidlink{0000-0001-8347-0830}$^{\textbf{2,3,6}}$} 

\author[1,3,7]{Lin Liu\orcidlink{0000-0002-9883-7962}$^{\dagger}$}

\authormark{ZHAO \textsc{et al.}} 
\titlemark{Covariate adjustment with GLM using jackknife}

\address[1]{\orgdiv{School of Mathematical Sciences}, \orgname{Shanghai Jiao Tong University}, \orgaddress{\state{Shanghai}, \country{China}}}

\address[2]{\orgdiv{Department of Bioinformatics and Biostatistics, School of Life Sciences and Biotechnology}, \orgname{Shanghai Jiao Tong University}, \orgaddress{\state{Shanghai}, \country{China}}}

\address[3]{\orgdiv{SJTU-Yale Joint Center for Biostatistics and Data Science, Technical Center for Digital Medicine, National Center for Translational Medicine}, \orgname{Shanghai Jiao Tong University}, \orgaddress{\state{Shanghai}, \country{China}}}


\address[4]{\orgdiv{Institute of Bioinformatics}, \orgname{Shanghai Academy of Experimental Medicine}, \orgaddress{\state{Shanghai}, \country{China}}}

\address[5]{\orgdiv{Department of Biostatistics}, \orgname{Yale University}, \orgaddress{\street{300 George Street}, \city{New Haven}, \postcode{CT 06511}, \country{United States}}}

\address[6]{\orgdiv{School of Medicine}, \orgname{Shanghai Children's Hospital, Shanghai Jiao Tong University}, \orgaddress{\state{Shanghai}, \country{China}}}

\address[7]{\orgdiv{Institute of Natural Sciences, MOE-LSC},
\orgname{Shanghai Jiao Tong University}, \orgaddress{\state{Shanghai}, \country{China}}}

\corres{
$^{\ast}$ The first author \\
$^{\dagger}$ Co-corresponding authors \\
Xinbo Wang \email{\href{cinbo_w@sjtu.edu.cn}{cinbo\_w@sjtu.edu.cn}}; \\
Lin Liu \email{\href{linliu@sjtu.edu.cn}{linliu@sjtu.edu.cn}}}



\abstract[Abstract]{Adjusting for baseline covariates has become standard practice in analyzing randomized clinical trials. In the low-dimensional setting, it is well understood that covariate adjustment through a parametric working model can sometimes be more efficient than the unadjusted difference-in-mean estimator. However, when the number of adjusted covariates is large relative to the sample size $n$, a na\"{i}ve adjustment may introduce excessive bias, leading to invalid statistical inference. The current literature that tries to resolve this issue is either limited to linear working models or relies on sample splitting, which may raise concerns about the replicability of RCT analyses. In this paper, we devise a novel jackknife-based approach to covariate adjustment through generalized linear models (GLMs), which we term as \underline{JA}ckknife \underline{S}core-based \underline{A}djustment (\JASA), together with its calibrated version \JASACal. By employing a nuanced jackknife strategy, \JASA{} and \JASACal{} avoid sample splitting and make full use of the data, while ensuring that the bias of \JASA{} or \JASACal{} is in general negligible even when the number of adjusted covariates is large compared to $n$. \JASA{} also encompasses state-of-the-art adjusted estimators through linear working models as a special case. In this paper, we provide some preliminary theoretical guarantees on \JASA{} and \JASACal{}, but such a theory is far from complete and a more comprehensive theoretical study is still ongoing work. Nonetheless, through extensive simulation experiments and a real data analysis, we demonstrate that \JASA{} or \JASACal{} can adjust for a much greater number of covariates than existing benchmarks. These empirical results also shed some new light on practical guidance for covariate adjustment with GLMs. Both \JASA{} and \JASACal{} have been incorporated into our R package \href{https://cran.r-project.org/web/packages/HOIFCar/index.html}{\texttt{HOIFCar}} available from CRAN. The package \href{https://cran.r-project.org/web/packages/HOIFCar/index.html}{\texttt{HOIFCar}} is developed to serve as a user-friendly option for covariate adjustment in RCTs, in particular when practitioners hope to adjust for a large number of covariates.}

\keywords{Randomized Clinical Trials, Covariate Adjustment, Generalized Linear Models, Jackknife, Calibration}

\maketitle

\renewcommand\thefootnote{}
\footnotetext{Data used in preparation of this article were obtained from the University of California, San Diego Alzheimer's Disease Cooperative Study. Consequently, ADCS Core Directors contributed to the original ADCS studies and/or provided data, but did not participate in the analyses or the writing of this manuscript/report.}

\allowdisplaybreaks

\section{Introduction}
\label{sec:intro}

Covariate adjustment methods for randomized clinical trials (RCTs) have been a very popular research topic in biostatistics and econometrics in the past decade \citep{leon2003semiparametric, freedman2008regressiona, freedman2008regressionb, zhang2008improving, moore2009covariate, lin2013agnostic, wu2018loop, lei2021regression, ma2022regression, guo2023generalized, ye2023toward, su2026decorrelation, cohen2024no, van2024covariate, van2026automated, zhao2024hoif, chang2024exact, lu2025debiased, jiang2025adjustments, bannick2025general, abadie2025unbiased}, due to (1) the unique impact of evidence from experimental data on policy or decision making and (2) a growing number of trials collecting multidimensional (baseline) covariates. As a testament to the progress made to date, the U.S. Food and Drug Administration (FDA) has recently issued practical guidelines on covariate adjustment in RCTs \citep{FDA2023}. Apparently, there seems to be a clear path for practitioners to follow when estimating treatment effects from RCT data in which covariates are  collected, in addition to the treatment assignments and outcomes of interest.

Nevertheless, as already pointed out in the FDA guidelines, it is clear how to adjust for covariates in RCTs only when the number of adjusted covariates, denoted by $p$, is small relative to the sample size $n$. When $p$ is relatively large, the guidelines \citep{FDA2023} advise that
\begin{quote}
\emph{``... sponsors should discuss their proposal with the relevant review division if the number of covariates is large relative to the sample size or if proposing to adjust for a covariate with many levels (e.g., study site in a trial with many sites).''}
\end{quote}

To address this gap, several recent works \citep{chang2024exact, lu2025debiased, zhao2024hoif, gu2025assumption} developed novel adjusted estimators for the average treatment effect (ATE), which are based on $U$-statistics and avoid sample splitting \citep{chernozhukov2018double}. To avoid confusion, throughout this paper, we always refer to ``sample splitting'' as splitting the sample into $M$ folds, with $M$ bounded (e.g. $M = 2, 3, 4$ or $5$). In particular, we do not interpret $U$-statistics or related concepts such as leave-one-out (LOO)/jackknife as sample splitting here. Importantly, these estimators are never less efficient and sometimes more efficient than the unadjusted estimator, without any structural assumptions (such as low complexity) on the outcome regression (OR) model, even when $p$ diverges with $n$ up to $p = o (n)$. More recently, interpreting the above $U$-statistic estimators as a form of LOO or jackknife ordinary least squares (OLS) regression-adjusted estimators \citep{zhao2024hoif, chiang2026regression}, \citet{abadie2025unbiased} proposed to replace the LOO-OLS regression adjustment by LOO-ridge regression, taking advantage of the numerical stability of ridge regression. In the implementation of our estimator in \citet{zhao2024hoif}, we actually use the pseudo-inverse of the sample Gram matrix also for improved numerical stability in finite samples. In particular, the pseudo-inverse of a matrix is equivalent to its ridgeless inverse \citep{hastie2022surprises}. Similar to \citet{chang2024exact}, \citet{abadie2025unbiased} focused on the setting with $p$ fixed; when $p = o (n)$, by choosing the ridge penalty $\lambda$ appropriately, theoretical guarantees similar to the LOO-OLS estimators are also attainable. Although these methods impose no structural assumptions on the OR, a common theme is their use of a linear working model to approximate the OR, and efficiency improvement can be attributed to a linear projection argument by least squares. We mention in passing that, in the econometrics literature, \citet{cattaneo2018inference, cattaneo2019two} have also considered using jackknife procedures to debias $M$/$Z$-estimators in the presence of high-dimensional nuisance parameters, with $p$ up to $O (\sqrt{n})$.

When using a linear working model, the amount of efficiency gain hinges on the strength of the linear association between the outcome $Y$ and the baseline covariates $X$ \citep{zhao2024hoif, lu2025debiased}. When the outcome $Y$ takes values in a constrained subset of $\bbR$, say $\{0, 1\}$ or the natural numbers $\bbZ$, the linear association may be weak, and the linear working model is also not natural for practitioners. When this is the case, practitioners can resort to generalized linear models (GLMs) \citep{nelder1972generalized} to approximate the OR, hoping to capture certain nonlinear associations to improve efficiency. 
When the number of covariates $p$ is small relative to $n$, \citet{guo2023generalized} developed the so-called generalized Oaxaca-Blinder (gOB) estimator for ATE, using canonical GLMs as the working model. \citet{cohen2024no} added an extra calibration step after fitting the GLM for the gOB estimator, again leveraging the linear projection structure of least squares to improve efficiency. When the dimension $p$ is moderately large, for instance $p \gtrsim \sqrt{n}$, estimation of a $p$-dimensional nuisance parameter can induce a bias of order $n^{-1/2}$ or greater in the estimation of the ATE, thereby invalidating standard inferential procedures such as Wald-type confidence intervals. To get a sense of how severe the problem could become for large $p$ when fitting GLMs using existing methods \citep{ye2023toward, guo2023generalized, cohen2024no}, we point to Figure~\ref{fig:hist_binomial} showing the results of a simulation study, in which the outcome $Y$ is binary, treatment effect is homogeneous, $n = 400$ and $p$ increases from $5$ to $60$. Here, the gOB calibration estimator\citep{guo2023generalized, cohen2024no} (denoted as $\hat{\tau}_{\gbcal}$) has a small bias when $p$ is small ($p = 5$, $n = 400$), but it already starts to fail to control the bias as $p$ increases further, even if the working model is correctly specified. However, our new estimator (denoted as $\hat{\tau}_{\JASACal}$) to be defined later consistently exhibits nearly zero bias across all values of $p$ examined in this experiment. Further details of this simulation study can be found in Section~\ref{sec:sim binary}.

A commonly proposed remedy to correct the above bias is sample splitting and cross-fitting \citep{bannick2025general, van2026automated}, which aims to mitigate bias by decoupling nuisance estimation from target parameter estimation. However, sample splitting presents notable challenges in the context of RCTs. First, the sample sizes available in many RCTs are often too limited to support effective sample splitting without a substantial loss of precision. More importantly, in finite samples, inference can be highly sensitive to the particular splitting scheme employed: different random splits may yield qualitatively different estimates and conclusions. This instability raises serious concerns about the replicability and interpretability of RCT analysis. As a result, we opt for developing procedures that (1) completely bypass sample splitting and cross-fitting, and (2) use as much information in the data as possible without introducing a significant amount of bias \citep{huo2026unified}.

\begin{figure}[htbp]
    \centering
    \includegraphics[width=0.7\linewidth]{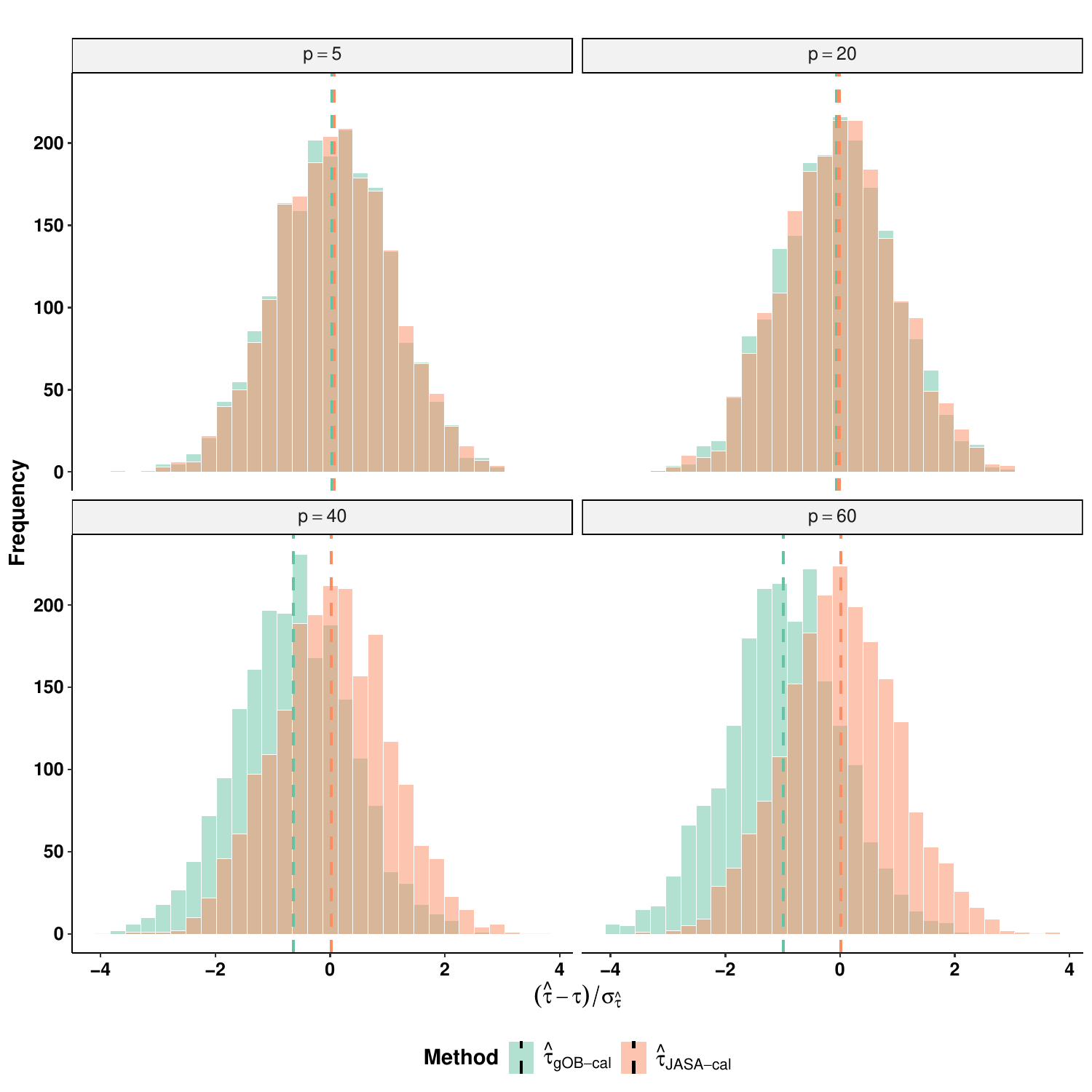}
    \caption{Distributions of scaled bias $(\hat{\tau} - \tau)/\sigma_{\hat{\tau}}$ for binary outcomes generated under a logistic regression model and homogeneous treatment effects, comparing the generalized Oaxaca‑Blinder calibration estimator $\hat{\tau}_{\gbcal}$ with the proposed $\hat{\tau}_{\JASACal}$. Results are shown for sample size $n=400$ and varying numbers of adjustment covariates $p \in \{5, 20, 40, 60\}$, corresponding to dimension‑to‑sample ratios $\alpha \coloneqq p/n \in \{0.0125, 0.05, 0.1, 0.15\}$. Each histogram summarizes 2000 Monte Carlo replications, with vertical dotted lines indicating the mean scaled bias for each estimator. The data‑generating process is detailed in Section~\ref{sec:sim binary}.}
    \label{fig:hist_binomial}
\end{figure}

\subsection*{Main contributions}

In this paper, to address the potentially large bias of the existing adjusted estimators (as shown in Figure~\ref{fig:hist_binomial} for the gOB calibration estimator $\hat{\tau}_{\gbcal}$), we develop a novel methodology to estimate ATE in RCTs. In our new method, 
we adjust for baseline covariates by fitting, within each arm, a working GLM through a jackknife \citep{efron1982jackknife} or leave-one-out (LOO)  bias-reduction scheme of the canonical score equation (see Section~\ref{sec:method} for details).
To our knowledge, the strategy of jackknifing score equations can be traced back to \citet{reeds1978jackknifing}. Importantly, deviating from recent related work \citep{su2026decorrelation, lu2025Cconditional}, our proposed procedure aims to achieve several desiderata:
\begin{itemize}
\item[(1)] The estimator has a bias of order $o (n^{-1 / 2})$.
\item[(2)] The estimator avoids sample splitting and cross-fitting \citep{schick1986asymptotically, ayyagari2010applications, robins2013new, chernozhukov2018double, newey2018cross} for the reproducibility of the results.
\item[(3)] The estimator places data efficiency at the center in the sense that the procedure uses as much information from the data as possible, without incurring a substantial bias.
\end{itemize}
As we will see, our procedure bypasses sample splitting and cross-fitting when fitting GLMs. Furthermore, it differs from standard jackknife-based adjustment methods \citep{wager2016high, cohen2024no} by using more information from the data, in the sense that we perform jackknife/LOO bias reduction only when it is needed; see the comments immediately after \eqref{loo score} later in Section~\ref{sec:ee_jasa}. While achieving all of the above desiderata, our procedure still has negligible bias (of order $o (n^{-1 / 2})$) even when $p$ is large compared to $n$. 

We refer to our new method as \underline{JA}ckknife \underline{S}core-based \underline{A}djustment, or \JASA{} for short. Specifically, due to the similarity between $U$-statistics and jackknife/LOO, \JASA{} in fact includes the aforementioned state-of-the-art adjusted estimators with OR fit by linear working models using $U$-statistics as special cases \citep{chang2024exact, zhao2024hoif, lu2025debiased, abadie2025unbiased}. The adjusted estimators based on \JASA{} are always unbiased for Bernoulli randomized experiments (BRE) and have negligible bias for completely randomized experiments (CRE) compared to the sampling variability when $n$ is sufficiently large.

Furthermore, as in \citet{cohen2024no}, to obtain more efficiency improvement when the working GLM may not be close to the true OR, one can combine \JASA{} with an additional calibration step, by treating the fitted OR as new covariates to be adjusted for in a linear working model. We coin the calibrated version of \JASA{} as \JASACal{}. We establish the $\sqrt{n}$-CAN properties of \JASA{} and \JASACal{} under relatively weak assumptions on the data generating process and the covariate dimension $p$. However, we want to emphasize that the theoretical contribution is not the main focus of this paper and we conjecture that the current theoretical results can be improved. Our main contribution is methodological: The actual implementation of \JASA{} or \JASACal{} takes into account various issues encountered in applications to strive for an accessible and stable method in practice. One potential disadvantage of our method, as detailed in Remark~\ref{rem:computation}, is the relatively long runtime for binary outcomes. It is still an open question how to accelerate the computation, but we argue that runtime should be a secondary concern in high-stakes applications such as RCTs.

Finally, given the now relatively long list of adjusted estimators at our disposal, we carry out extensive simulation experiments to investigate their finite sample performance. These empirical results shed some new light on the practical guidance regarding in what scenarios certain adjusted estimators should be recommended in practice. We also consider various practical issues, including but not limited to: 
\begin{enumerate}[label = \roman*.]
\item In what scenario is it  important to use \JASA{} or \JASACal{} or $U$-statistic-based estimators proposed recently \citep{chang2024exact, zhao2024hoif, lu2025debiased, gu2025assumption} to ensure a negligible bias when adjusting for a large number of covariates?

\item When is fitting the OR with GLMs  better than with linear working models, as the latter is guaranteed to improve efficiency even without calibration?

\item How to solve the score equations in practice and how to perform calibration in our new framework?
\end{enumerate}
In this paper, our overarching goal is to provide a practical method with stable finite-sample performance, for safely adjusting for covariates whose dimension may be large relative to the sample size without incurring substantial bias that could be a major threat to valid inference. As we will see later in the paper, our theoretical guarantees concerning the asymptotic variance of the new adjusted estimators do not entirely cover the actual practical implementation, but we provide comprehensive empirical evidence for its competitive finite-sample performance compared to existing benchmarks. A more complete theoretical justification is an important open problem that we are still working on in ongoing work. Finally, we also mention in passing that a recent independent study further illustrates the intuition and benchmarks the performance of our proposed method with a more comprehensive set of simulation experiments \citep{alene2026use}.

\begin{remark}
\label{rem:car}
One could also choose to balance covariates in the design stage, through stratified or general covariate-adaptive randomization procedures. However, most such procedures in common use incorporate only a limited number of discrete stratification factors and therefore do not balance the full set of collected baseline covariates \citep{ma2022regression, ye2023toward}. By contrast, covariate adjustment in the analysis stage can incorporate continuous covariates together with other prognostic baseline covariates in a more flexible manner; this practice is also encouraged by regulatory guidance \citep{FDA2023}. Before our work, it is unclear what to do if a trial collects many covariates beyond those used in stratification, and we provide a solution that possibly fills a gap in the current guideline. Finally, it is worth clarifying that (1) our method is not intended to select important covariates using data-driven method, and (2) we do not recommend adjusting for covariates that are known \emph{a priori} to be nonprognostic.
\end{remark}

\subsection*{Notation}
Before proceeding, we introduce some common notations used throughout the paper. We reserve $g (\cdot)$ for a monotonically increasing and three-times differentiable (inverse) link function of the working model for OR. We assume that $g( \cdot)$ is known throughout the paper.

We let $\bm{X} = (X_{1}, \ldots, X_{p})^{\top} \in \bbR^{p}$ denote the random vector of baseline covariates, and let $\bm{Z} = (1, \bm{X}^{\top})^{\top}$ denote the augmented covariates vector including an intercept. Let $[n]$ denote the index set $\{1,\ldots, n\}$. When concatenating over all the samples of size $n$, let $\bbX \coloneqq (\bm{X}_{1}, \ldots, \bm{X}_{n})^{\top} \in \bbR^{n \times p}$, where $\bbX$ is assumed to be centered. Let $\bbZ \coloneqq (\mathbf{1}_{n}, \bbX) \in \bbR^{n \times (p + 1)}$  be the design matrix, where $\mathbf{1}_{n} \coloneqq (1, \ldots, 1)^{\top} \in \bbR^{n}$ is the vector of ones. Let $\bSigma = \bbE[\bm{Z}\bm{Z}^{\top}]$ be the population Gram matrix of $\bm{Z}$. Define the sample Gram matrices $\hat{\bSigma}\coloneqq n^{-1}\bbZ^\top \bbZ$, $\hat{\bSigma}^\star \coloneqq n^{-1} \bbX^{\top} \bbX$, and let $\bbH \coloneqq \bbZ \left( \bbZ^\top \bbZ \right)^{-1} \bbZ^\top $ and $\bbH^{\star}\coloneqq\bbX(\bbX^{\top}\bbX)^{-1}\bbX^{\top}$ denote the corresponding hat matrices of $\bbZ$ and $\bbX$, respectively. Let $\bT \coloneqq (T_{1}, \ldots, T_{n})^{\top}$ be the vector of treatment assignments, where $T_{i} \in \{0, 1\}$, and define $n_t=\sum_{i=1}^{n}\mathbbm{1}\{T_i=t\}$ as the number of units assigned to the treatment group $t\in \{0, 1\}$. 
For each $t$, let $\bSigma_{t} \coloneqq \bbE[\bm Z \bm Z^{\top} \mid T = t]$ denote the population Gram matrix of $\bm Z$ within treatment arm $t$, and let $\hat{\bSigma}_{t} \coloneqq n_{t}^{-1} \sum_{i:T_i=t}\bm{Z}_i\bm{Z}_{i}^\top$ denote its sample analogue. We use $\bbH_{t} = \bbZ (\sum_{i:T_{i}=t} \bm Z_{i} \bm Z_{i}^{\top} )^{-1} \bbZ^{\top}$ to denote the hat matrix corresponding to the treatment-arm-specific design matrix. 
Finally, for any square matrix $\mathbf{A}$, we let $\lambda_{\min} (\mathbf{A})$ and $\lambda_{\max} (\mathbf{A})$, respectively, denote the smallest and largest eigenvalues of $\mathbf{A}$. 

\subsection*{Organization}
The remainder of the paper is organized as follows. Section~\ref{sec:setup} sets the stage by introducing the mathematical setup and reviewing existing related methodologies. We then introduce our new estimators in the most general form in Section~\ref{sec:method}, with more specialized forms given in Section~\ref{sec:examples} for linear models with continuous and unconstrained outcomes, logistic linear models with binary outcomes, and Poisson log-linear models with count outcomes. In Section~\ref{sec:sim} we conduct comprehensive synthetic experiments to evaluate the performance of our new estimators, followed by an analysis of a real RCT in Section~\ref{sec:real}. Section~\ref{sec:conclusion} concludes the paper.

We have implemented \JASA{} and \JASACal{} in our R package \href{https://cran.r-project.org/web/packages/HOIFCar/index.html}{\texttt{HOIFCar}}, which is publicly available on CRAN. One of the main goals of our work is to make the software numerically stable and user-friendly to facilitate practical implementation.

\section{Setup and A Review of Some Existing Approaches}
\label{sec:setup}

In this paper, we adopt the super-population framework and consider the simple Bernoulli randomization experiments, thereby implying the following assumptions on the observed sample:
\begin{equation}
\label{iid}
\{\bm{O}_{i} \coloneqq \left( \bm{X}_{i}, T_{i}, Y_{i} \right)\}_{i = 1}^{n} \overset{\rm i.i.d.}{\sim} \bbP,
\end{equation}
and particularly
\begin{equation}
\label{random}
T_{i} \overset{\rm i.i.d.}{\sim} \text{Bernoulli} (\pi_{1}),
\end{equation}
where the probability of assigning to the treatment group, $\pi_{1}$, is strictly bounded away from $0$ and $1$ and we let $\pi_{0} \coloneqq 1 - \pi_{1}$. 
We let $\hat{\pi}_{t} \coloneqq n_t/n$ be the empirical proportion of units assigned to the treatment group $t$ for $t \in \{0, 1\}$. We also impose the following regularity assumptions on the observed data distribution, which are commonly assumed in the literature \citep{lin2013agnostic}.

\begin{assumption}
\label{as:regularity}
\leavevmode
\begin{enumerate}[label = (\roman*)]
\item Each dimension of $\bm{X}$ is strictly bounded away from $-\infty$ and $+\infty$;

\item $Y$ has bounded fourth moment;

\item There exist absolute constants $0<b_{1}\leq b_{2}<\infty$ such that  $b_{1} \leq \lambda_{\min} \{\bSigma_{t}\} \leq \lambda_{\max} \{\bSigma_{t}\} \leq b_{2}$ for $t \in \{0, 1\}$.
\end{enumerate}
\end{assumption}

We are interested in constructing point and interval estimators for the ATE $\tau \coloneqq \bbE [Y (1)] - \bbE [Y (0)]$, where we adopt the potential outcome notation to define ATE and we make the standard consistency assumption to connect the observed outcome and the potential outcomes: $Y = \sum_{t \in \{0, 1\}} \mathbbm{1} \{T = t\} Y (t)$. We denote the true OR as $\mu_{t} (\bm{x}) \coloneqq \bbE [Y (t) \mid \bm{X} = \bm{x}, T = t]$. We are now ready to review several popular estimators of $\tau$. In this review, we only focus on treatment effect estimators without selecting the adjusted covariates from the data \citep{van2026automated, liu2026coadvise}. Below we use the estimated treatment assignment probability $\hat{\pi}_{1}$. In the theoretical analysis, we instead consider the estimators using the true treatment assignment probability $\pi_{1}$ to simplify the proof.

\begin{itemize}
\item The unadjusted estimator:
\begin{align*}
\hat{\tau}_{\unadj} \coloneqq \bar{Y}_{1} - \bar{Y}_{0} = \frac{1}{n} \sum_{i = 1}^{n} \frac{T_{i} Y_{i}}{\hat{\pi}_1} - \frac{1}{n} \sum_{i = 1}^{n} \frac{(1 - T_{i}) Y_{i}}{1-\hat{\pi}_1},
\end{align*}
which is essentially the inverse probability weighting (IPW) estimator of $\tau$. $\hat{\tau}_{\unadj}$ is $\sqrt{n}$-CAN without imposing any assumption on the OR or the dimension $p$, but does not take advantage of potentially rich information in the covariates $\bbX$.

\item The adjusted estimator based on fitting a linear working model by OLS \citep{ma2022regression, ye2023toward}:
\begin{align*}
\hat{\tau}_{\OLS} \coloneqq & \ \Big\{ \frac{1}{n_{1}} \sum_{i = 1}^{n} T_{i} (Y_{i} - \hat{\bbeta}_{1}^{\top} \bm{Z}_{i}) + \frac{1}{n} \sum_{i=1}^{n} \hat{\bbeta}_{1}^{\top} \bm{Z}_{i} \Big\} - \Big\{ \frac{1}{n_{0}} \sum_{i = 1}^{n} ( 1 - T_{i} ) (Y_{i} - \hat{\bbeta}_{0}^{\top} \bm{Z}_{i}) + \frac{1}{n} \sum_{i=1}^{n} \hat{\bbeta}_{0}^{\top} \bm{Z}_{i} \Big\} \\
= & \ \frac{1}{n}\sum_{i=1}^{n}\Big\{\frac{T_iY_i}{\hat{\pi}_1} + \sum_{j=1}^{n}\Big(1 - H_{1,ij}\Big)\frac{T_j Y_j}{\hat{\pi}_1}\Big\} - \frac{1}{n}\sum_{i=1}^{n}\Big\{\frac{(1-T_i)Y_i}{1-\hat{\pi}_1} + \sum_{j=1}^{n}\Big(1 - H_{0,ij}\Big)\frac{(1-T_j) Y_j}{1-\hat{\pi}_1}\Big\},
\end{align*} 
where $\hat{\bbeta}_{t}={\hat{\bSigma}_{t}}^{-1}\frac{1}{n}\sum_{i=1}^{n}\bm{Z}_i\frac{\mathbbm{1} \{T_{i} = t\} Y_{i}}{\hat{\pi}_t}$. $\hat{\tau}_{\OLS}$ can be interpreted as the augmented inverse probability weighting (AIPW) estimator of $\tau$ \citep{robins1994estimation}, with the OR $\mu_{t}(\bm{x})$ estimated by a linear working model within the treatment group $t$ using OLS. \citet{jiang2025adjustments} proved that $\hat{\tau}_{\OLS}$ is $\sqrt{n}$-CAN and is guaranteed to be never less efficient than $\hat{\tau}_{\unadj}$ without any assumption on the true OR $\mu_{t}(\bm{x})$, if $p = o (\sqrt{n})$. When $\sqrt{n} \lesssim p \ll n$, although the asymptotic variance of $\hat{\tau}_{\OLS}$ is still smaller than that of $\hat{\tau}_{\unadj}$, it may have non-negligible bias and fail to be $\sqrt{n}$-CAN. An exception is when the linear working model happens to be the true OR model \citep{jiang2025adjustments}.

\item The bias-corrected estimator based on $U$-statistics \citep{chang2024exact, zhao2024hoif, lu2025debiased, gu2025assumption, liu2017semiparametric}: $\hat{\tau}_{\adj, 2}$ and $\hat{\tau}_{\adj, 2}^{\dag}$, where $\hat{\tau}_{\adj, 2}^{\dag}$ differs from $\hat{\tau}_{\OLS}$ in that $\hat{\bbeta}_{t}$ is replaced by 
\begin{equation}
\label{loo-beta}
\hat{\bbeta}_{t}^{(-i)} \coloneqq \hat{\bSigma}^{-1} \frac{1}{n}\sum_{j \neq i} \bm{Z}_j \frac{\mathbbm{1} \{T_{j} = t\}( Y_{j} - \bar{Y}_{t} )}{\hat{\pi}_{t}}
\end{equation}
and $\hat{\tau}_{\adj, 2}$ simply does not center $Y_{i}$ by $\bar{Y}_{t}$. Both $\hat{\tau}_{\adj, 2}$ and $\hat{\tau}_{\adj, 2}^{\dag}$ are $\sqrt{n}$-CAN and are guaranteed to be never less efficient than $\hat{\tau}_{\unadj}$ without any assumption on the true OR $\mu_{t}(\bm{x})$, if $p = o (n)$.

\item The generalized Oaxaca-Blinder (gOB) estimator \citep{guo2023generalized}: 
\begin{align*}
\hat{\tau}_{\gob} \coloneqq \frac{1}{n} \sum_{i=1}^{n} \left( \hat{Y}_{i} (1) - \hat{Y}_i (0) \right),
\end{align*}
where, for $t \in \{0, 1\}$, $\hat{Y}_{i} (t) \coloneqq \mathbbm{1} \{T_{i} = t\} Y_{i} + \mathbbm{1} \{T_{i} \neq t\} \hat{\mu}_{t}(\bm{Z}_i)$ and $\hat{\mu}_{t}$ is an estimator of $\mu_{t}(\cdot)$ obtained by fitting the MLE of a working canonical GLM in group $t$. Although not immediately obvious, the gOB estimator $\hat{\tau}_{\gob}$ is essentially the AIPW estimator of $\tau$:
\begin{align*}
\hat{\tau}_{\gob} \equiv \frac{1}{n}\sum_{i=1}^{n}\Big\{\frac{T_i\{Y_i - \hat{\mu}_{1}(\bm{Z}_i)\}}{\hat{\pi}_1} + \hat{\mu}_{1}(\bm{Z}_i)\Big\} - \frac{1}{n}\sum_{i=1}^{n}\Big\{\frac{(1-T_i)\{Y_i - \hat{\mu}_{0}(\bm{Z}_i)\}}{1-\hat{\pi}_1} + \hat{\mu}_{0}(\bm{Z}_i)\Big\},
\end{align*}
due to the prediction unbiasedness of the MLE of canonical GLMs \citep{robins2007comment, guo2023generalized}. 

\item The calibrated gOB estimator: \citet{cohen2024no} proposed the calibrated version of the gOB estimator
\begin{align*}
    \hat{\tau}_{\cal}\coloneqq \frac{1}{n}\sum_{i=1}^{n}\left(\tilde{Y}_{i}(1)-\tilde{Y}_i(0)\right),
\end{align*}
where for $t\in \{0,1\}, \tilde{Y}_{i}(t)\coloneqq \mathbbm{1} \{T_i = t\}Y_i + \mathbbm{1}\{T_i \neq t\}\tilde{\mu}_{t}(\bm{Z}_i)$, but $\tilde{\mu}_{t}$ is then obtained by a multivariate ``OLS'' calibration in the treated and control groups: 
\begin{align*}
    \tilde{\mu}_t(\bm{Z}_i) &\coloneqq\hat{\gamma}_{t} + \hat{\mu}_{0}(\bm{Z}_i)\hat{\gamma}_{t,0}  + \hat{\mu}_{1}(\bm{Z}_i)\hat{\gamma}_{t,1},\\        (\hat{\gamma}_{t},\hat{\gamma}_{t,0},\hat{\gamma}_{t,1}) & = \underset{(\gamma_{t},\gamma_{t,0},\gamma_{t,1})}{\argmin} \sum_{i:T_i=t} \left(Y_i - \gamma_{t} - \hat{\mu}_{0}\left(\bm{Z}_i\right)\gamma_{t,0} - \hat{\mu}_{1}\left(\bm{Z}_i\right)\gamma_{t,1}\right)^2.
\end{align*}
By the prediction unbiasedness of OLS, $\hat{\tau}_{\cal}$ can also be represented in the AIPW form.

\item \citet{cohen2024no} further considered the following LOO estimator \cite{wager2016high}:
\begin{align*}
\hat{\tau}_{\loo} & \coloneqq \frac{1}{n}\sum_{i=1}^{n}\Big\{\frac{T_i\{Y_i - \hat{\mu}_{1}^{(-i)}(\bm{Z}_i)\}}{\hat{\pi}_1} + \hat{\mu}_{1}^{(-i)}(\bm{Z}_i)\Big\} - \frac{1}{n}\sum_{i=1}^{n}\Big\{\frac{(1-T_i)\{Y_i - \hat{\mu}_{0}^{(-i)}(\bm{Z}_i)\}}{1-\hat{\pi}_1} + \hat{\mu}_{0}^{(-i)}(\bm{Z}_i)\Big\},
\end{align*}
where $\hat{\mu}_{t}^{(-i)}(\cdot)$ is fitted within treatment arm $t$ by leaving out the $i$-th observation $\bO_i$. The calibrated version of $\hat{\tau}_{\loo}$ is then defined as
\begin{align*}
\hat{\tau}_{\loocal} \coloneqq & \ \frac{1}{n}\sum_{i=1}^{n}\Big\{\frac{T_i\{Y_i - \tilde{\mu}_{1}^{(-i)}(\bm{Z}_i)\}}{\hat{\pi}_1} + \tilde{\mu}_{1}^{(-i)}(\bm{Z}_i)\Big\} - \frac{1}{n}\sum_{i=1}^{n}\Big\{\frac{(1-T_i)\{Y_i - \tilde{\mu}_{0}^{(-i)}(\bm{Z}_i)\}}{1-\hat{\pi}_1} + \tilde{\mu}_{0}^{(-i)}(\bm{Z}_i)\Big\},
\end{align*}
where for $t\in \{0,1\}$ and $i\in[n]$, $\tilde{\mu}_{t}^{(-i)}(\cdot)$ is fitted by solving the following least squares problem:
\begin{align}
    \tilde{\mu}_{t}^{(-i)}(\bm{Z}_i) & \coloneqq \hat{\gamma}_t^{(-i)} + \hat{\mu}_{0}^{(-i)}(\bm{Z}_i)\hat{\gamma}_{t,0}^{(-i)} + \hat{\mu}_{1}^{(-i)}(\bm{Z}_i)\hat{\gamma}_{t, 1}^{(-i)}, \\
    (\hat{\gamma}_t^{(-i)},\hat{\gamma}_{t,0}^{(-i)},\hat{\gamma}_{t, 1}^{(-i)}) &= \underset{({\gamma}_t^{(-i)},{\gamma}_{t,0}^{(-i)},{\gamma}_{t, 1}^{(-i)})}{\argmin} \sum_{\substack{j:T_j=t \\ j\neq i}}\left(Y_j - \gamma_{t}^{(-i)} - \hat{\mu}_{0}^{(-i)}\left(\bm{Z}_j\right)\gamma_{t,0}^{(-i)} - \hat{\mu}_{1}^{(-i)}\left(\bm{Z}_j\right)\gamma_{t,1}^{(-i)}\right)^2. \label{equ: loo cal in}
\end{align} 
\end{itemize}

Before moving forward, we note that, for the simple randomization considered here, when the ORs are fitted within the corresponding arm using GLMs with canonical links, the estimator proposed in \citet{bannick2025general} is equivalent to $\hat{\tau}_{\gob}$. This particular estimator, together with its calibrated version, has been implemented in the popular R package \href{https://cran.r-project.org/web/packages/RobinCar/index.html}{\texttt{RobinCar}}~\cite{bannick2026robincarfamilyrtools}. It should be noted that \href{https://cran.r-project.org/web/packages/RobinCar/index.html}{\texttt{RobinCar}}  can additionally handle black-box machine learning methods. When the ORs are fitted using OLS with linear working models, the adjusted estimator is equivalent to the ANHECOVA estimator proposed in \citet{ye2023toward}, which has also been implemented in \href{https://cran.r-project.org/web/packages/RobinCar/index.html}{\texttt{RobinCar}}. This paper, together with our previous work \citep{zhao2024hoif, gu2025assumption, zhang2026bias}, aims to provide complementary approaches for the settings in which the number of adjusted covariates is large relative to $n$ and current methods implemented in \href{https://cran.r-project.org/web/packages/RobinCar/index.html}{\texttt{RobinCar}} may not perform well. 

\section{The New Jackknife-Based Methodology}
\label{sec:method}

In this section, we introduce our new method \JASA{} and its calibrated version \JASACal{}. Our motivation is to develop an estimator that can adjust for a large number of covariates (relative to $n$) using GLM working models without incurring excessive bias.

\subsection{Jackknife score equations of GLMs}
\label{sec:ee_jasa}

As explained in \citet{zhao2024hoif}, for $\hat{\tau}_{\OLS}$, there are two main sources of bias: (i) bias due to the use of the same sample to estimate the parameter of interest $\hat{\tau}_{\OLS}$ and the nuisance estimate $\hat{\bbeta}_{t}$, in the form of $V$-statistics; and (ii) bias due to the use of treatment assignments in estimating $\bSigma$ in the OLS nuisance estimate $\hat{\bbeta}_{t}$. The bias-corrected estimators based on $U$-statistics\citep{chang2024exact, zhao2024hoif, lu2025debiased, gu2025assumption}, such as $\hat{\tau}_{\adj, 2}$ and $\hat{\tau}_{\adj, 2}^{\dag}$, resolve these two issues, essentially by replacing $\hat{\bbeta}_{t}$ with $\hat{\bbeta}_{t}^{(-i)}$ defined in \eqref{loo-beta}, in which $\bSigma$ is estimated by $\hat{\bSigma}$ computed from the entire sample and the inner product between the outcome and the covariates is computed by leaving the $i$-th sample out.

Similar reasons for $\hat{\tau}_{\gob}$ to be biased apply. When we fit the OR by a working GLM instead of a linear model, in general, we cannot find a closed-form solution for the regression coefficients. Nevertheless, just as the OLS estimator solves the normal equation, the MLE of a working GLM solves the corresponding score equation. Specifically, the population score equation of a canonical GLM within treatment arm $t \in \{0, 1\}$ takes the following form:
\begin{equation}
\label{population score}
\bbE \Big[ \bm{Z} \frac{\mathbbm{1} \{T = t\}}{\pi_{t}} \{Y - \mu (\bm{Z}; \bbeta_{t})\} \Big] = \mathbf{0},
\end{equation}
where $\mu (\bm{Z}; \bbeta_{t}) = g (\bm{Z}^{\top} \bbeta_{t})$ for some known (inverse) link function $g(\cdot)$. When $Y$ is binary, we can choose $g (\cdot)$ to be the expit function $g (u) = 1 / \{1 + \exp (-u)\}$, corresponding to logistic regression; when $Y$ is count data, we can choose $g(\cdot)$ to be the exponential function $g (u) = e^{u}$, corresponding to log-linear Poisson regression. We denote $\hat{\bbeta}_{t}$ as the MLE of $\bbeta_{t}$, which solves the empirical score equation, i.e., the empirical version of \eqref{population score}. 

\begin{remark}   
\label{rem:estimate ps}
In this section, we define all treatment effect estimators using the true probabilities of treatment assignments $\pi_{t}$ for $t \in \{0, 1\}$. In the actual implementation, we follow most of the literature by replacing $\pi_{t}$ with $\hat{\pi}_{t}$, the empirical proportion of units in the treatment group $t$, which also corresponds to the true probabilities of treatment assignments in CRE. The bias incurred by using $\hat{\pi}_{t}$ instead of $\pi_{t}$ is $o (n^{-1 / 2})$, which is negligible compared to the sampling variability. In practice, using $\hat{\pi}_{t}$ could result in estimators of $\tau$ with smaller asymptotic variances.
\end{remark}

In the theoretical results, we need to impose some additional regularity conditions on the data generating distribution.
\begin{assumption}
\label{as:link}
Let 
$\bm{\Xi}_{t} (\bbeta_{t}) \coloneqq \bbE [\frac{\mathbbm{1} \{T = t\}}{\pi_t} \bm{Z} \bm{Z}^{\top} g' (\bm{Z}^{\top} \bbeta_{t})]$ 
be the population Jacobian matrix of the score equation \eqref{population score}. Note that when $g(\cdot)$ is the identity link, $\bm{\Xi}_{t} (\bbeta_{t}) \equiv \bSigma_{t}$. We assume the following: for $t \in \{0, 1\}$,
\begin{enumerate}[label = (\roman*)]
\item $g(\cdot)$ is monotonically increasing and is at least three-times continuously differentiable, and $c_{g} \leq g'(\cdot) \leq C_{g} $ for some absolute constants $0< c_{g} \leq C_{g} < \infty$;

\item There exists an absolute constant $B > 0$ such that $\max \{\Vert \bbeta_{0} \Vert_{2}, \Vert \bbeta_{1} \Vert_{2}\} \leq B$;

\item The MLE $\hat{\bbeta}_{t}$ satisfies $\Vert \hat{\bbeta}_{t} - \bbeta_{t} \Vert_{2} = O_{\bbP} \left( \sqrt{p \log n / n} \right)$;

\item There exist absolute constants $0 < c_{1} \leq c_{2} < \infty$ such that $c_{1} \leq \lambda_{\min} \{\bm{\Xi}_{t} (\bbeta_{t})\} \leq \lambda_{\max} \{\bm{\Xi}_{t} (\bbeta_{t})\} \leq c_{2}$.
\end{enumerate}
\end{assumption}

\begin{remark}
\label{rem:link}
We now comment on the conditions imposed in Assumption~\ref{as:link}. The monotonicity and smoothness requirements in Assumption~\ref{as:link}(a) hold for commonly encountered GLMs, including probit regression, logistic regression, Poisson log-linear model, among many others; as for the boundedness condition, the identity link trivially satisfies it, while the logistic and Poisson inverse links satisfy it whenever \(\bm Z^\top\bbeta_t\) is uniformly bounded. Assumption~\ref{as:link}(b) essentially requires that the population version of the working model $\mu_{t} (\bm{Z}; \bbeta_{t})$ is bounded, under Assumption~\ref{as:regularity}(a). Assumption~\ref{as:link}(c) is also as expected based on the standard $M$/$Z$-estimation theory \citep{van2023weak}.  Finally, Assumption~\ref{as:link}(d) is a direct consequence of Assumption~\ref{as:regularity}(c) and Assumption~\ref{as:link}(a). 
\end{remark}


Inspired by the construction of $\hat{\tau}_{\adj, 2}$ and $\hat{\tau}_{\adj, 2}^{\dag}$, we propose the following jackknife empirical version of the population score equation \eqref{population score} (recall that $\mu (\bm{Z}; \bbeta_{t}) = g (\bm{Z}^{\top} \bbeta_{t})$):
\begin{equation}
\label{loo score}
\frac{1}{n - 1} \sum_{j \neq i} \bm{Z}_{j} \frac{\mathbbm{1} \{T_{j} = t \}}{\pi_{t}} Y_{j} - \frac{1}{n} \sum_{l = 1}^{n} \bm{Z}_{l} g (\bm{Z}_{l}^{\top}\hat{\bbeta}_{t}^{(-i)}) = \mathbf{0}, \, \text{ for } i \in [n].
\end{equation}
The above proposal is novel and unique in the following aspects:
\begin{enumerate}[label = \roman*.]
\item The new jackknife score equations do not rely on sample splitting or cross-fitting, which may be undesirable in practice due to information loss and undermined reproducibility.

\item The new jackknife score equation corresponding to the $i$-th unit differs from the traditional jackknife-based adjustment \citep{cohen2024no} in that it only removes the $i$-th treatment $T_{i}$ and the $i$-th outcome $Y_{i}$ but keeps the $i$-th covariates $\bm{Z}_{i}$ through the second term in \eqref{loo score}. It uses more information in the data and can adjust for more covariates than standard jackknife/LOO methods, as demonstrated in Section~\ref{sec:sim}. The data efficiency of \eqref{loo score} is desirable in RCTs, where it is costly to recruit experimental subjects and collect data.
\end{enumerate}

\begin{remark}
\label{rem:loo}
\citet{wager2016high} and \citet{cohen2024no} considered the following standard jackknife version of the population score equation \eqref{population score}:
\begin{align}
\label{loo score standard}
\frac{1}{n - 1} \sum_{j \neq i} \bm{Z}_j \frac{\mathbbm{1} \{T_j = t\}}{\pi_t} \left\{ Y_j - g (\bm{Z}_j^{\top} \bbeta_{t}^{(-i)}) \right\} = \mathbf{0}, \, \text{ for } i \in [n],
\end{align}
where the second term on the left side not involving $Y$ is also estimated by jackknife.
\end{remark}

We denote the solution to \eqref{loo score} by $\hat{\bbeta}_{t}^{(-i)}$ for $i \in [n]$ and $t \in \{0, 1\}$. 
We then define the JAckknife Score-based Adjustment (\JASA) treatment effect estimator of ATE $\tau$ based on $\{\hat{\bbeta}_{t}^{(-i)}\}_{i = 1}^{n}$ as:
\begin{equation}
\label{equ:JASA}
\begin{split}
\hat{\tau}_{\JASA} \coloneqq &  \hat{\tau}_{1, \JASA} - \hat{\tau}_{0, \JASA}, \\
\text{ where} \ \  \hat{\tau}_{t, \JASA} \coloneqq  \frac{1}{n} \sum_{i = 1}^{n} \Big\{ \frac{\mathbbm{1} \{T_{i} = t\}}{\pi_{t}} & Y_{i} +  \Big( 1 - \frac{\mathbbm{1} \{T_{i} = t\}}{\pi_{t}} \Big) g (\bm{Z}_{i}^{\top} \hat{\bbeta}_{t}^{(-i)}) \Big\}, \ \  \text{for} \ \ t \in \{0, 1\}.
\end{split}
\end{equation}
We have the following theoretical results for $\hat{\tau}_{t, \JASA}$ and $\hat{\tau}_{\JASA}$. The following result benefits from recent theoretical advances in \citet{lin2024worthwhile} on the jackknife $M$/$Z$-estimation. The proofs of all the theoretical results in this section are deferred to Appendix~\ref{app:var}. We want to emphasize that our main contribution is on the methodological side. The theoretical results established here mainly serve as proof of concept, as already indicated in the Introduction.

\begin{proposition}
\label{prop:bias}
When \eqref{iid} and \eqref{random} hold: for $t \in \{0, 1\}$,
\begin{enumerate}[label = (\roman*)]
\item $\hat{\tau}_{t, \JASA}$ and $\hat{\tau}_{\JASA}$ are both exactly unbiased regardless of the values of $n$ and $p$: $\bbE [\hat{\tau}_{t, \JASA} - \tau_{t}] = 0$ and $\bbE [\hat{\tau}_{\JASA} - \tau] = 0$.

\item If Assumptions~\ref{as:regularity}--\ref{as:link} also hold and $p = o (n^{2 / 3})$, then $\hat{\tau}_{t, \JASA}$ and $\hat{\tau}_{\JASA}$ are, respectively, $\sqrt{n}$-CAN estimators of $\tau_{t}$ and $\tau$:
\begin{align*}
\sqrt{n} (\hat{\tau}_{t, \JASA} - \tau_{t}) \overset{\rm d}{\to} \mathcal{N} \left( 0, \sigma_{t}^{2} \right) \text{ and } \sqrt{n} (\hat{\tau}_{\JASA} - \tau) \overset{\rm d}{\to} \mathcal{N} \left( 0, \sigma^{2} \right),
\end{align*}
where
\begin{align*}
& \sigma_{t}^{2} = \var \Big\{ \frac{\mathbbm{1} \{T = t\}}{\pi_{t}} Y + \Big( 1 - \frac{\mathbbm{1} \{T = t\}}{\pi_{t}} \Big) g (\bm{Z}^{\top} \bbeta_{t}) \Big\}, \  \text{for} \ t \in \{0, 1\}, \\
& \sigma^{2} = \var \Big\{ \frac{\mathbbm{1} \{T = 1\}}{\pi_{1}} Y + \Big( 1 - \frac{\mathbbm{1} \{T = 1\}}{\pi_{1}} \Big) g (\bm{Z}^{\top} \bbeta_{1}) - \frac{\mathbbm{1} \{T = 0\}}{\pi_{0}} Y - \Big( 1 - \frac{\mathbbm{1} \{T = 0\}}{\pi_{0}} \Big) g (\bm{Z}^{\top} \bbeta_{0}) \Big\}.
\end{align*}
Here, $\bbeta_{t}$ is the solution to \eqref{population score} for $t \in \{0, 1\}$.
\end{enumerate}
\end{proposition}

\begin{remark}
\label{rem:proof}
The proof of Proposition~\ref{prop:bias} is based on recent advances in $M$/$Z$-estimation theory involving nuisance parameters with dimensions diverging with the sample size $n$ \citep{lin2024worthwhile}; also see \citet{cattaneo2019two} for related important results. We essentially follow the proof roadmap of \citet{lin2024worthwhile} by showing that the difference in the asymptotic variances of $\hat{\tau}_{\JASA}$ and $\hat{\tau}_{\gob}$ is negligible, with the latter reaching $\sigma^{2}$ as soon as $p = o (n^{2 / 3})$. Theoretical results such as Proposition~\ref{prop:bias} are not the main contribution of our paper and mainly serve as a proof of concept. We conjecture that the restriction on $p$ is not sharp, but it may require the development of new proof strategies, which is beyond the scope of this paper. 
\end{remark}

It is worth highlighting that $\hat{\tau}_{\JASA}$ is always unbiased regardless of the scaling between $p$ and $n$, but excessively large $p$ can be detrimental to the variance. As suggested in \citet{cohen2024no}, we can further improve $\hat{\tau}_{\JASA}$ by adding a calibration step, which is achieved by fitting the following linear regression by OLS within each treatment group $t \in \{0, 1\}$:
\begin{align}
\label{model:linear_cal}
(\hat{\gamma}_{t}, \hat{\gamma}_{t,0}, \hat{\gamma}_{t,1}) \coloneqq \underset{\gamma_{t}, \gamma_{t,0}, \gamma_{t,1}}{\argmin} \sum_{i: T_i = t} \left( Y_i - \gamma_{t} - \gamma_{t,0} \cdot \hat{\mu}_{0,i}^{(-i)} - \gamma_{t,1} \cdot \hat{\mu}_{1,i}^{(-i)} \right)^2. 
\end{align}
For each observation $i\in[n]$, we compute the calibrated OR estimate $\tilde{\mu}_{t, i} \coloneqq \hat{\gamma}_{t} +  \hat{\gamma}_{t,0} \cdot \hat{\mu}_{0,i}^{(-i)} + \hat{\gamma}_{t,1} \cdot  \hat{\mu}_{1,i}^{(-i)}$, where $\hat{\mu}_{t,i}^{(-i)} \coloneqq g (\bm{Z}_{i}^{\top} \hat{\bbeta}_{t}^{(-i)})$ for $t\in \{0, 1\}$. The calibrated ATE estimator (\JASACal{}) is then defined as:
\begin{align}
\hat{\tau}_{\JASACal} \coloneqq \frac{1}{n} \sum_{i = 1}^{n} \left[ \left\{ \frac{\mathbbm{1} \{T_i = 1\}}{\pi_1} Y_i + \left( 1 - \frac{\mathbbm{1} \{T_i = 1\}}{\pi_1} \right) \tilde{\mu}_{1,i} \right\} - \left\{ \frac{\mathbbm{1} \{T_i = 0\}}{\pi_0} Y_i + \left(1 - \frac{\mathbbm{1} \{T_i = 0\}}{\pi_0} \right) \tilde{\mu}_{0, i} \right\} \right].
\label{equ:JASA-cal}
\end{align}
At the end of the main text, we present the pseudocode of \JASA{} and \JASACal{} in Algorithm~\ref{alg:jasa}. As a consequence of Proposition~\ref{prop:bias} and the results in \citet{cohen2024no}, the following theoretical results hold for $\hat{\tau}_{t, \JASACal}$ and $\hat{\tau}_{\JASACal}$. The proof is a direct consequence of the observation that $\hat{\tau}_{\JASACal}$ is simply the OLS adjusted estimator with the pairs $(\hat{\mu}_{0, i}^{(-i)}, \hat{\mu}_{1, i}^{(-i)})$, $i \in [n]$, as the new two-dimensional covariates. As long as $p = o (n^{1 / 2})$, the OLS adjusted estimator without jackknife is $\sqrt{n}$-CAN \citep{ma2022regression, ye2023toward, jiang2025adjustments}.

\begin{proposition}
\label{prop:cal}
Under the same assumptions as in Proposition~\ref{prop:bias}(b), for $t \in \{0, 1\}$, $\hat{\tau}_{t, \JASACal}$ and $\hat{\tau}_{\JASACal}$ are, respectively, $\sqrt{n}$-CAN estimators of $\tau_{t}$ and $\tau$:
\begin{align*}
\sqrt{n} (\hat{\tau}_{t, \JASACal} - \tau_{t}) \overset{d}{\to} \mathcal{N} (0, \sigma_{t, \cal}^{2}) \text{ and } \sqrt{n} (\hat{\tau}_{\JASACal} - \tau) \overset{d}{\to} \mathcal{N} (0, \sigma_{\cal}^{2}),
\end{align*}
where $\sigma_{t, \cal}^{2}$ and $\sigma_{\cal}^{2}$, respectively, take the same forms as $\sigma_{t}^{2}$ and $\sigma^{2}$, except that $g (\bm{Z}^{\top} \bbeta_{t})$ shall be replaced by the probability limit of $\tilde{\mu}_{t}(\cdot)$, which exists under the stated assumptions.
\end{proposition}

\begin{remark}
\label{rem:variance estimator}
For inference based on $\hat{\tau}_{\JASA}$ and $\hat{\tau}_{\JASACal}$, we directly apply the following variance estimators based on influence functions:
\begin{align}
\hat{\sigma}_{\JASA}^{2} & = \frac{1}{n} \sum_{i=1}^{n} \Bigg[ \left\{ \frac{\mathbbm{1} \{T_{i} = 1\}}{\pi_1} Y_i + \left( 1 - \frac{\mathbbm{1} \{T_{i} = 1\}}{\pi_1} \right) \hat{\mu}_{1,i}^{(-i)} \right\} \notag \\ 
& \quad \quad - \left\{ \frac{\mathbbm{1} \{T_{i} = 0\}}{\pi_0} Y_i + \left( 1 - \frac{\mathbbm{1} \{T_{i} = 0\}}{\pi_0} \right) \hat{\mu}_{0,i}^{(-i)} \right\} - \hat{\tau}_{\JASA} \Bigg]^2, \label{equ: var_jasa} \\
\hat{\sigma}_{\JASACal}^{2} & = \frac{1}{n} \sum_{i=1}^{n} \Bigg[ \left\{ \frac{\mathbbm{1} \{T_{i} = 1\}}{\pi_1} Y_i + \left( 1 - \frac{\mathbbm{1} \{T_{i} = 1\}}{\pi_1} \right) \tilde{\mu}_{1,i} \right\} \notag \\
& \quad \quad - \left\{ \frac{\mathbbm{1} \{T_{i} = 0\}}{\pi_0} Y_i + \left( 1 - \frac{\mathbbm{1} \{T_{i} = 0\}}{\pi_0} \right) \tilde{\mu}_{0,i} \right\} - \hat{\tau}_{\JASACal} \Bigg]^2. \label{equ: var_jasa-cal}
\end{align}
\end{remark}

It is then straightforward to establish the consistency of these variance estimators, as Proposition~\ref{prop:bias} has shown that $\hat{\bbeta}_{t}^{(-i)}$ is a consistent estimator of $\bbeta_{t}$ for $t \in \{0, 1\}$.

\begin{proposition}
\label{prop:var}
Under the same conditions as in Proposition~\ref{prop:bias}(b), $\hat{\sigma}_{\JASA}^{2} \overset{\bbP}{\to} \sigma^{2}$ and $\hat{\sigma}_{\JASACal}^{2} \overset{\bbP}{\to} \sigma_{\cal}^{2}$.
\end{proposition}

Before moving forward, we would like to clarify that the above theoretical results are preliminary support of the actual methods of \JASA{} and \JASACal{} implemented in practice. This is because, as will be clear immediately in the next subsection, we need to make additional modifications to solving the jackknife score equations to make our proposed methods stably implementable in practice.

\subsection{Practical implementation: Solving \JASA{} by constrained optimization}
\label{sec:opt_jasa}

While the \JASA{} estimator $\hat{\tau}_{\JASA}$ introduced in the previous section is nearly unbiased, the jackknife score equations \eqref{loo score} can be non-trivial to solve numerically and may not always admit a solution in finite samples, in particular when $p$ is close to $n$ \citep{albert1984existence, sur2019moderna, sur2019modernb, chen2024method}. When it is numerically difficult to solve the jackknife score equations \eqref{loo score}, we resort to either of the following two constrained optimization problems \citep{wang2024debiased} instead, depending on the type of the outcome: for $i \in [n]$,
\begin{enumerate}[label = (\roman*)]
\item In \formbeta{}: we still parameterize the OR using the working GLM but minimize the following squared $\ell_{2}$ loss of \eqref{loo score} instead of directly setting \eqref{loo score} to zero:
\begin{align}
\label{opt:general beta}
\hat{\bbeta}_{t}^{(-i)}=&\; \underset{\bbeta_{t}^{(-i)} \in \mathcal{B}_{t}}{\argmin} \Big\|\frac{1}{n - 1} \sum_{j \neq i} \bm{Z}_{j} \frac{\mathbbm{1} \{T_{j} = t\}}{\pi_{t}} Y_{j} - \frac{1}{n} \sum_{l = 1}^{n} \bm{Z}_{l} g (\bm{Z}_{l}^\top \bbeta_{t}^{(-i)}) \Big\|_2^{2},
\end{align}
where $\mathcal{B}_{t}$ is the feasible region of the unknown variable $\bbeta_{t}^{(-i)}$, which can be prespecified by users. A potential disadvantage of \formbeta{} is that the resulting optimization could be nonconvex, due to the nonlinearity of the link function $g(\cdot)$. 

\item In \formmu{}: we view the GLM working model only as a way to ensure that the estimated OR respects the corresponding sample space constraint. For example, if $Y$ is binary, then applying the logistic regression mainly serves the purpose of constraining the estimated OR $\hat{\mu}_{t}$ to be within $[0, 1]$. From this perspective, instead of optimizing the potentially nonconvex objective in \eqref{opt:general beta}, we can directly optimize w.r.t. $\bm{\mu}_{t}^{(-i)}= \left(\mu_{t,1}^{(-i)},\ldots, \mu_{t,n}^{(-i)}\right)^\top$ subject to the constraints on the range of $\bm{\mu}_{t}^{(-i)}$:
\begin{align}
\label{opt:general mu}
\hat{\bm{\mu}}_{t}^{(-i)} = \ \underset{\bm{\mu}_{t}^{(-i)}: m \leq \mu_{t, j}^{(-i)} \leq M, j \in [n]}{\argmin} \Big\|\frac{1}{n - 1} \sum_{j \neq i} \bm{Z}_{j} \frac{\mathbbm{1} \{T_{j} = t \}}{\pi_{t}} Y_{j} - \frac{1}{n} \sum_{l = 1}^{n} \bm{Z}_{l} \mu_{t,l}^{(-i)} \Big\|_2^{2},
\end{align}
where $m$ and $M$ are prespecified constants, determined by the range of the true OR $\mu_{t}(\cdot)$. Specifically, if $Y$ is continuous, we set $m=-\infty, M=\infty$; if $Y$ is binary, we set $m=0, M=1$; and if $Y$ is count, we set $m=0, M=\infty$.
In \eqref{opt:general mu}, the objective function is quadratic and the constraint is linear, so the optimization problem is convex for every $i \in [n]$, guaranteeing the existence of solutions and fast convergence when using off-the-shelf numerical algorithms to solve convex programs.
\end{enumerate}

Before commenting on other aspects of the new optimization formulations, we first explain the rationale behind these optimization formulations.
\begin{itemize}
\item The lesson from the previous section tells us that, as long as an estimator of $g (\bm{Z}_{i}^{\top} \bbeta_{t})$ does not depend on the treatment $T_i$ and outcome $Y_i$ of the $i$-th unit, the corresponding augmentation term has mean zero and the treatment effect estimator does not introduce bias. To see this, note that under i.i.d.\ sampling and Bernoulli randomization, $T_i$ is independent of $\{\bm{Z}_{l}\}_{l=1}^{n}$ and of $\{(T_j, Y_j)\}_{j \neq i}$. 
Therefore any fitted prediction $g (\bm{Z}_{i}^{\top} \hat{\bbeta}_{t}^{(-i)})$ computed without using $(T_i, Y_i)$ is independent of $T_i$, and hence
\begin{equation*}
\bbE \Big[ \Big( 1 - \frac{\mathbbm{1}\{T_i=t\}}{\pi_t} \Big) g (\bm{Z}_{i}^{\top} \hat{\bbeta}_{t}^{(-i)}) \Big]
= \bbE \Big[ 1 - \frac{\mathbbm{1}\{T_i=t\}}{\pi_t} \Big] \, \bbE \big[ g (\bm{Z}_{i}^{\top} \hat{\bbeta}_{t}^{(-i)}) \big] = 0.
\end{equation*}
\item To bypass such numerical difficulties, we choose to minimize the empirical squared $\ell_{2}$ norm of the jackknife score equations \eqref{loo score} subject to some constraints on the ranges of either $\bbeta_{t}$ or the working OR model $\bm{\mu}_{t}(\cdot)$. This choice has two features. First, the constraints help stabilize the jackknife estimates of $\bbeta_{t}$. If the values of $\{\hat{\bbeta}_{t}^{(-i)}\}_{i = 1}^{n}$ are too volatile, the variance of the treatment effect estimator can become too large to be useful in terms of covariate adjustment. Second, the objective functions in the above optimization problems drive the estimates of either $\bbeta_{t}$ or the working OR model $\bm{\mu}_{t}$ to better predict the outcomes, which is crucial to improve the estimation efficiency of the treatment effects \citep{zhao2024hoif, lu2025debiased}.

\item In particular, \formmu{} is mainly designed for count outcomes, when log-linear models are used as the working OR model. We will further elaborate this rationale in Section~\ref{sec:count}.
\end{itemize}
Although the resulting optimization problem is slightly more involved than directly solving the jackknife score equations in \eqref{loo score}, we have implemented the proposed procedures in the R package \href{https://cran.r-project.org/web/packages/HOIFCar/index.html}{\texttt{HOIFCar}}
, which is publicly available on CRAN, to facilitate practical implementation. The constrained optimization procedure mainly serves to stabilize the estimator of the working outcome regression model, when $p$ is large compared to $n$ so that the standard MLE of GLMs may not exist. This phenomenon is well documented in the literature \citep{sur2019moderna}. The constrained optimization procedure can also be viewed as a generalization of OLS with the pseudo-inverse of the sample Gram matrix to GLMs. The main missing piece in our theoretical justification is to show that the resulting adjusted estimator when solving the constrained optimization problem is still asymptotic linear. Extensive simulation studies and the real data analysis in Sections~\ref{sec:sim} and \ref{sec:real} demonstrate the numerical stability and practical performance of the proposed methods. However, proving this result remains an important open problem, and may require a generalization of the techniques developed in \citet{lin2024worthwhile} from $Z$-estimation to $M$-estimation to show that the solution to solving the empirical constrained optimization problem converges to some stable limit. There could be other alternative approaches that also stabilize the working model fit, such as adding explicit regularization to the original jackknife score equations. We leave the above important issues and open problems, together with relaxing the requirement on $p$ to $p = o (n)$, to future work.

\begin{remark}
\label{rem:calibration diff}
In contrast to the calibration strategy in \eqref{equ: loo cal in} for $\hat{\tau}_{\loocal}$,  the approach taken in \eqref{model:linear_cal} for $\hat{\tau}_{\JASACal}$ calibrates the OR fits only over $\{\hat{\mu}_{t, i}^{(-i)}\}_{i = 1}^{n}$. Following this strategy, an alternative calibrated estimator can also be constructed for the standard LOO estimator $\hat{\tau}_{\loo}$; we denote it by $\hat{\tau}_{\loocalout}$.
\end{remark}

\begin{remark}
\label{rem:calibration}
When $p$ is small, collinearity between OR fits in different treatment arms may occur. To ensure stability of our proposed procedure, in the actual implementation, we do not always calibrate the fitted OR. In particular, if the prediction values using the fitted OR are either extremely large/small or are almost constants, we bypass calibration and simply use the uncalibrated estimates $\hat{\mu}_{t, i}^{(-i)}$ for all $i \in [n]$ and $t \in \{0, 1\}$. The precise implementation details and further elaborations can be found in Appendix~\ref{sec:cond calib}.
\end{remark}

\begin{remark}
\label{rem:MLE}
A natural alternative to \formbeta{} or \formmu{} is to directly maximize an anti-derivative of the jackknife score equation \eqref{loo score}. However, as demonstrated in Appendix~\ref{app:MLE}, such a strategy does not improve the performance of directly solving the jackknife score equation \eqref{loo score}.
\end{remark}

\begin{remark}
\label{rem:computation}
Before moving on to the next section, we finally comment on the computational issue of \JASA{} or \JASACal{}, as it is well known that jackknife-based approaches are computationally demanding in general. Given a dataset of size $n$ in the treatment/control arm, \JASA{} either solves the score equations \eqref{loo score}, or the optimization formulations [\eqref{opt:general beta} or \eqref{opt:general mu}] for $2 n$ times. In the implementation, we recommend parallelizing these $2 n$ fits, as they constitute the main computational bottleneck. In terms of the two different formulations: \formmu{} \eqref{opt:general mu} is easier to solve due to its convexity nature, whereas \formbeta{} \eqref{opt:general beta} is more difficult to solve in practice; see Figures~\ref{fig:binomial_g0_time} and~\ref{fig:poisson_g0_time} for empirical evidence. It remains an open question how to improve the runtime performance of \JASA{} or \JASACal{} in practice. But we argue that runtime should be a secondary concern in high-stakes applications such as RCTs.
\end{remark}

\section{Illustration for Several Concrete Outcome Types}
\label{sec:examples}

To further illustrate $\hat{\tau}_{\JASA}$ and its calibrated version $\hat{\tau}_{\JASACal}$, we now specialize to three particular outcome types --  continuous (Section~\ref{sec:continuous}), binary (Section~\ref{sec:binary}), and count (Section~\ref{sec:count}) -- and discuss the corresponding jackknife score equations, the reformulated constrained optimization problems and the connections to existing procedures. For ease of exposition, we let $\bY^{\dag}_{t} \coloneqq (Y_{t, 1}^{\dag}, \ldots, Y_{t, n}^{\dag})^{\top}$, where $Y_{t, i}^{\dag} \coloneqq \mathbbm{1} \{T_{i} = t\} Y_{i} / \pi_{t}$ for $i \in [n]$. Throughout this section, $\bbZ_{-i,\cdot}$ and $\bY^{\dag}_{t,-i}$ are obtained from $\bbZ$ and $\bY^{\dag}_{t}$ by deleting the $i$-th row and the $i$-th entry, respectively.

\subsection{Continuous outcomes}
\label{sec:continuous}
When the outcome is continuous and unconstrained (taking values in $\bbR$), it is natural to choose the identity link $g (u) = u$. The jackknife score equations \eqref{loo score} then reduce to the jackknife normal equations:
\begin{align*}
\frac{1}{n - 1} \sum_{j \neq i} \bm{Z}_{j} Y_{t, j}^{\dag} - \frac{1}{n} \sum_{l = 1}^{n} \bm{Z}_{l} \bm{Z}_{l}^{\top} \bbeta_{t}^{(-i)} = \bm{0}, \ i \in [n],
\end{align*}
which can also be compactly represented in matrix form:
\begin{align}
\label{equ:score_gaussian}
\frac{1}{n - 1} \bbZ_{-i, \cdot}^{\top} \bY_{t, -i}^{\dag} - \frac{1}{n} \bbZ^{\top} \bbZ \bbeta_{t}^{(-i)} = \bm{0}, \ i \in [n].
\end{align}

When $\hat{\bSigma} = n^{-1} \bbZ^{\top} \bbZ$ is invertible, for any $i \in [n]$, a unique solution $\hat{\bbeta}_{t}^{(-i)}$ to \eqref{equ:score_gaussian} exists and takes the same form as the LOO coefficient estimator in $\hat{\tau}_{\adj, 2}$ \citep{zhao2024hoif}. Therefore, in this case, $\hat{\tau}_{\JASA}$ reduces to $\hat{\tau}_{\adj, 2}$. 

When $\hat{\bSigma}$ is non-invertible, for each $i \in [n]$, \eqref{equ:score_gaussian} has infinitely many solutions. Among all solutions, a natural recourse is to use the Moore-Penrose pseudoinverse $\hat{\bSigma}^{\dag}$. We first define $\hat{\bbeta}_{t}^{(-i)} (\lambda)$ as the optimizer of the following constrained program: for $i \in [n]$
\begin{align}
\label{equ:opt1_gaussian}
\hat{\bbeta}_{t}^{(-i)} (\lambda) \coloneqq \underset{\bbeta_{t}^{(-i)}: \Vert \bbeta_{t}^{(-i)} \Vert_{2}^{2} \leq \lambda}{\argmin} \left\Vert \frac{1}{n - 1} \bbZ_{-i, \cdot}^{\top} \bY_{t, -i}^{\dag} - \frac{1}{n} \bbZ^{\top} \bbZ \bbeta_{t}^{(-i)} \right\Vert_{2}^{2}.
\end{align}
The optimization problem \eqref{equ:opt1_gaussian} includes the ridge-penalized unbiased adjusted estimator proposed recently in \citet{abadie2025unbiased} as a special case. Then let $\hat{\bbeta}_{t}^{(-i)} \coloneqq \lim_{\lambda \downarrow 0} \hat{\bbeta}_{t}^{(-i)} (\lambda)$. Here, $\hat{\bbeta}_{t}^{(-i)}$ is defined in terms of \formbeta{} discussed in Section~\ref{sec:opt_jasa}, and can be interpreted as the min-norm or ridgeless solution to \eqref{equ:score_gaussian} \citep{hastie2022surprises}. When $\hat{\bSigma}$ is invertible, the Moore-Penrose pseudoinverse $\hat{\bSigma}^{\dag}$ and the actual inverse $\hat{\bSigma}^{-1}$ coincide. We mention in passing that, in our previous work \citep{zhao2024hoif}, we actually implemented the Moore-Penrose pseudoinverse in our accompanying R package \href{https://cran.r-project.org/web/packages/HOIFCar/index.html}{\texttt{HOIFCar}}.

We now exhibit the \JASA{} estimator when $g(\cdot)$ is the identity link
\begin{equation}
\label{JASA continuous}
\begin{split}
\hat{\tau}_{\JASA} = \frac{1}{n} \sum_{i = 1}^{n} \left[ \left\{ Y_{1, i}^{\dag} + \left( 1 - \frac{\mathbbm{1} \{T_i = 1\}}{{\pi}_1} \right) \hat{\mu}_{1}^{(-i)} (\bm{Z}_{i}) \right\} - \left\{ Y_{0, i}^{\dag} + \left( 1 - \frac{\mathbbm{1} \{T_i = 0\}}{{\pi}_0} \right) \hat{\mu}_{0}^{(-i)} (\bm{Z}_{i}) \right\} \right],
\end{split}
\end{equation}
where
\begin{align*}
\hat{\mu}_{t}^{(-i)} (\bm{Z}_{i}) \coloneqq \bm{Z}_{i}^{\top} \hat{\bbeta}_{t}^{(-i)} = \bm{Z}_{i}^{\top} \hat{\bSigma}^{\dag} \frac{1}{n - 1} \sum_{j \neq i} \bm{Z}_{j} Y_{t, j}^{\dag} = \frac{n}{n - 1} \sum_{j \neq i} H_{i j} Y_{t, j}^{\dag}.
\end{align*}
It is straightforward to see that $\hat{\tau}_{\JASA}$ is equivalent to $\hat{\tau}_{\adj, 2}$ considered in \citet{zhao2024hoif} (see Section~\ref{sec:setup}). Moreover, by viewing $\hat{\tau}_{\adj, 2}$ as a special case of the \JASA{} estimator, we can also further calibrate the predictions $\{\hat{\mu}_{1,i}^{(-i)}\}_{i=1}^{n}$ and $\{\hat{\mu}_{0,i}^{(-i)}\}_{i=1}^{n}$ via OLS as specified in (\ref{model:linear_cal}), thereby obtaining the  corresponding calibrated estimator $\hat{\tau}_{\adjcal}$. 

As discussed in \citet{zhao2024hoif}, one can replace all the appearance of $Y_{i}$ in $\hat{\tau}_{\JASA}$ by $Y_{i} - \bar{Y}_{t}$ for $i \in [n]$. Then the resulting \JASA{} estimator reduces to $\hat{\tau}_{\adj, 2}^{\dag}$, also the estimator constructed in \citet{lu2025debiased}. Similarly, one can construct the calibrated estimator $\hat{\tau}_{\adjcal}^{\dag}$. We summarize the above discussions in the following proposition.

\begin{proposition}
\label{prop:continuous}
When the link function $g(\cdot)$ is identity, $\hat{\tau}_{\JASA} \equiv \hat{\tau}_{\adj, 2}$. If $Y_{i}$ is replaced by $Y_{i} - \bar{Y}_{t}$ in $\hat{\tau}_{\JASA}$ \eqref{JASA continuous} for all $i \in [n]$, then $\hat{\tau}_{\JASA} \equiv \hat{\tau}_{\adj, 2}^{\dag}$.
\end{proposition}
In Appendix~\ref{sec:proof propositoin adj2c db}, we also demonstrate that  $\hat{\tau}_{\adj, 2}^{\dag}$ is agnostic to whether we use $\bbZ$ or $\bbX$. The corresponding variance estimators of $\hat{\tau}_{\adj, 2}$ and $\hat{\tau}_{\adj, 2}^{\dag}$ are provided in Appendix~\ref{var: adj2 adj2c}.

\subsection{Binary outcomes}
\label{sec:binary}

When the outcome is binary, we focus on the expit link function $g (u) = 1 / (1 + e^{-u})$ used in the logistic regression. The compact form of the jackknife score equations reads as follows:
\begin{equation}
\label{equ:score_binary}
\frac{1}{n-1} \bbZ_{-i, \cdot}^{\top} \bY_{t, -i}^{\dag} - \frac{1}{n}\bbZ^\top \frac{1}{1+\exp (-\bbZ \bbeta_{t}^{(-i)})} = \bm{0}, \, i\in [n].
\end{equation}
However, these jackknife score equations may not always admit solutions due to the nonlinearity of $g(\cdot)$, in particular when $p$ is large compared to $n$ \citep{sur2019modernb}. Instead, we formulate it through \formbeta{} as in \eqref{opt:general beta} for each $i \in [n]$:
\begin{align*}
\hat{\bbeta}_{t}^{(-i)} = & \ \underset{\bbeta_{t}^{(-i)}: m_t \leq \bbeta_{t}^{(-i)} \leq M_t}{\argmin} \Big\|\frac{1}{n-1} \bbZ_{-i, \cdot}^{\top} \bY_{t, -i}^{\dag} - \frac{1}{n} \bbZ^\top \frac{1}{1+\exp (-\bbZ \bbeta_{t}^{(-i)})} \Big\|_2^{2}.
\end{align*}
Here, $m_t$ and $M_t$ are pre-specified constants. This optimization problem can be solved using the L-BFGS-B algorithm, as implemented in the \texttt{optim} function of the \texttt{stats} package in R. Based on these estimates, we obtain the estimated OR for each arm as follows:
\begin{align*}
\hat{\mu}_{t, i}^{(-i)} \equiv \hat{\mu}_{t}^{(-i)} (\bm{Z}_{i}) = \frac{1}{1 + \exp (-\bm{Z}_i^\top \hat{\bbeta}_t^{(-i)})}, \, i \in [n], t\in \{0,1\}.
\end{align*}
Substituting these into equation \eqref{equ:JASA} yields the \JASA{} estimator $\hat{\tau}_{\JASA}$. A subsequent calibration of $\{\hat{\mu}_{1,i}^{(-i)}\}_{i=1}^{n}$ and $\{\hat{\mu}_{0,i}^{(-i)}\}_{i=1}^{n}$ following \eqref{model:linear_cal} leads to the calibrated \JASA{} estimator $\hat{\tau}_{\JASACal}$.

\subsection{Count outcomes}
\label{sec:count}

When the outcome is count,  we focus on the exponential link function $g (u) = e^{u}$ used for the Poisson log-linear regression. The compact form of the jackknife score equations reads as follows:
\begin{equation}
\label{equ:score_count}
\frac{1}{n-1} \bbZ_{-i, \cdot}^{\top} \bY_{t, -i}^{\dag} - \frac{1}{n} \bbZ^\top \exp \left( \bbZ \bbeta_t^{(-i)} \right) = \bm{0}, \, i \in [n].
\end{equation}
Similarly, we can write \eqref{equ:score_count} in \formbeta{} as the following optimization problem:
\begin{align}
\label{opt:count}
\hat{\bbeta}_{t}^{(-i)} = & \ \underset{\bbeta_{t}^{(-i)}: m_t \leq \bbeta_{t}^{(-i)} \leq M_t}{\argmin} \left\| \frac{1}{n - 1} \bbZ_{-i, \cdot}^{\top} \bY_{t, -i}^{\dag} - \frac{1}{n} \bbZ^\top \exp \left( \bbZ \bbeta_t^{(-i)} \right) \right\|_2^{2}, \ i \in [n],
\end{align}
where the lower and upper bounds of the feasible region, $m_{t}$ and $M_{t}$, can be prespecified. However, due to the exponential transformation in $g(\cdot)$, solving \eqref{opt:count} can be quite unstable numerically, especially when $p$ is relatively large compared to $n$. To mitigate the numerical instability, instead of estimating $\bbeta_{t}^{(-i)}$, we resort to \formmu{} and directly optimize the OR as in \eqref{opt:general mu} with non-negative constraints because the OR has to be non-negative:
\begin{align}
\hat{\bm\mu}_{t}^{(-i)}  \coloneqq & \ \underset{\bm{\mu}_{t}^{(-i)}: \bm{\mu}_{t}^{(-i)} \geq \bm{0}}{\argmin} \left\| \frac{1}{n - 1} \bbZ_{-i, \cdot}^{\top} \bY_{t, -i}^{\dag} - \frac{1}{n}\bbZ^\top \bm{\mu}_{t}^{(-i)} \right\|_2^{2}, \ i \in [n].
\end{align}
Here, for each $i \in [n]$, we solve a non-negative least squares (NNLS) problem, for which numerically stable algorithms exist \citep{kim2010tackling}. 
Substituting $\{\hat{\mu}_{1, i}^{(-i)}\}_{i = 1}^{n}$ and $\{\hat{\mu}_{0, i}^{(-i)}\}_{i = 1}^{n}$ into \eqref{equ:JASA} yields the \JASA{} estimator $\hat{\tau}_{\JASA}$. Again, a subsequent calibration step leads to the calibrated \JASA{} estimator $\hat{\tau}_{\JASACal}$.

\begin{remark}
\label{rem:survival}
As pointed out by a referee, time-to-event outcomes are common in oncology trials and it will be useful to provide practitioners with similar tools in the context of survival analysis. Our paper mostly focuses on continuous, binary, and count outcomes. Nonetheless, it is possible to extend our framework to time-to-event outcomes. Here, we point out one direction. In a seminal work \citep{ye2024covariate}, the authors proposed a methodology to adjust for baseline covariates using a linear working model in the context of log-rank tests. Specifically, they first estimate a linearization of the log-rank test numerator, denoted by $O$ in \citet{ye2024covariate}, and then conduct a linear regression using $O$ as response, and the covariates $X$ and the interaction between treatment and $X$ as predictors (Eq. (2) in \citet{ye2024covariate}). In our framework, setting $g(\cdot)$ as the identity link and replacing $Y$ by $O$ yields the corresponding JASA adjusted estimator without incurring bias due to the potentially large number of covariates compared to the sample size $n$; this estimator will be equivalent to that in our previous work \citep{zhao2024hoif}. We expect most of our theoretical and empirical results in this paper to continue to hold. But a more comprehensive theoretical and empirical investigation of the performance of applying our framework to time-to-event outcomes is beyond the scope of this paper, and will be reported in a separate future work.
\end{remark}

\section{Simulation Studies}
\label{sec:sim}

In this section, we conduct extensive experiments to examine the finite-sample performance of the proposed \JASA{} and \JASACal{} estimators for three different types of outcome: continuous, binary, and count. We focus on simple randomization, $T_{i} \sim \text{Bernoulli} (\pi_{1})$ with $\pi_{1} = 1 / 3$ for $i \in [n]$. Our target parameter is the ATE $\tau$. We compute the true ATE under various data generating mechanisms by averaging over a large sample of size $10^7$. As noted in Remark~\ref{rem:estimate ps}, in the implementation, we use the estimated probability of treatment assignments $\hat{\pi}_{t}$ instead of $\pi_{t}$, for $t \in \{0, 1\}$. The replication code for our simulation studies can be found in the accompanying \href{https://github.com/Cinbo-Wang/Simu_JASA/tree/main}{GitHub page}. We also note that in the numerical experiments below, we consider settings that may slightly violate the regularity conditions (Assumptions~\ref{as:regularity}--\ref{as:link}). Thus, our empirical results also serve as a robustness check of the performance of our proposed method when violating the technical assumptions, which we believe can be relaxed at the expense of a more complicated proof.

We compare \JASA{} and \JASACal{} with two groups of estimators, in addition to the unadjusted estimator $\hat{\tau}_{\unadj}$:
\begin{itemize}
\item Group 1 (working GLMs): $\hat{\tau}_{\gob}$ and its calibrated version $\hat{\tau}_{\gbcal}$ \citep{guo2023generalized, bannick2025general}; standard LOO estimators including $\hat{\tau}_{\loo}$ \citep{wu2018loop, cohen2024no}, and its two calibrated versions $\hat{\tau}_{\loocal}$ \citep{cohen2024no} and $\hat{\tau}_{\loocalout}$.

\item Group 2 (linear working models): $\hat{\tau}_{\OLS}$ \citep{ye2023toward, ma2022regression}, $\hat{\tau}_{\adj, 2}$ \citep{zhao2024hoif}, $\hat{\tau}_{\adj, 2}^{\dag}$ \citep{lu2025debiased}, and their calibrated versions $\hat{\tau}_{\adjcal}$ and $\hat{\tau}_{\adjcal}^{\dag}$.
\end{itemize}

In our implementation, for $\hat{\tau}_{\gob}$, $\hat{\tau}_{\gbcal}$, and $\hat{\tau}_{\OLS}$, we directly apply the state-of-the-art R package \href{https://cran.r-project.org/web/packages/RobinCar/index.html}{\texttt{RobinCar}}. For continuous outcomes, estimators $\hat{\tau}_{\gob}$, $\hat{\tau}_{\gbcal}$, and $\hat{\tau}_{\OLS}$ are equivalent, when an identity link function is specified in the GLM. We generate $K=2000$ Monte Carlo replications to calculate the summary statistics for evaluation. We fix the sample size $n = 400$ and vary the covariate dimension $p$ in the cases of continuous and binary outcomes. 
The evaluation metrics considered include the following:
\begin{itemize}
    \item The absolute mean bias $|\bbE [\hat{\tau} - \tau]|$.
    
    \item The ratio of Monte Carlo standard deviations (SD) between the estimator under comparison and $\hat{\tau}_{\unadj}$ (abbreviated as the SD ratio).
    
    \item The ratio between the mean of the estimated standard error and the Monte Carlo standard deviation (SE/SD).
    
    \item Empirical coverage probabilities (CPs) of the nominal 95\% Wald confidence intervals (CIs) associated with all estimators under comparison, including $\hat{\tau}_{\unadj}$. 

    \item In Appendix~\ref{app:qqplots}, qqplots of $\hat{\tau}_{\JASA}$ and $\hat{\tau}_{\JASACal}$ are shown to demonstrate that their distributions are close to normal.
\end{itemize}

\subsection{Continuous outcomes}
\label{sec:sim continuous}

We first examine the finite-sample performance of $\hat{\tau}_{\JASA}$ and $\hat{\tau}_{\JASACal}$ with various benchmarks just mentioned in the case of continuous outcomes. 

We first generate baseline covariates $\bbX_\star \in \bbR^{n\times p_\star}$, with each row $\bmX_{\star,i} \sim t_3 (0,\bSigma)$, where the true covariance matrix, $\bSigma$, has the Toeplitz structure: $\bSigma_{k,l}=0.1^{|k-l|},k,l \in[p_\star] $ and $p_\star = 60$. We generate the following coefficient vector: $\beta_j = (-1)^j j^{-1/4},j\in[p_\star]$ and rescale $\bbeta$ so that $\|\bbeta\|_2 = 1$. Then we vary the true OR in the control group ($t = 0$), the true conditional ATE (CATE) and the condition number $\alpha \coloneqq p / n$, in the data generating process:
\begin{itemize}
    \item \textbf{The true OR in the control group ($t = 0$)}:
        \begin{itemize}
        \item Linear: $\mu_0 (\bm{x}_{\star}) = \bm{x}_\star^\top \bbeta$,
        
        \item Nonlinear: $\mu_{0} (\bm{x}_{\star}) = \frac{1}{2} \mathrm{sign} (\bm{x}_\star^{\top} \bbeta) \cdot |\bm{x}_\star^{\top} \bbeta|^{\frac{1}{3}} + \cos (\bm{x}_\star^{\top} \bbeta) + \bm{x}_\star^{\top} \bbeta$.
        \end{itemize}
    \item \textbf{The true CATE}: 
        \begin{itemize}
        \item Homogeneous: $\mu_{1} (\bm{x}_{\star}) - \mu_{0} (\bm{x}_{\star}) \equiv 1$,
    
        \item Heterogeneous: $\mu_{1} (\bm{x}_\star) - \mu_{0} (\bm{x}_{\star}) = 1 - \frac{1}{2} \min\{x_{\star, 1}^{2}, 5\} + \frac{1}{4} \sum_{j=1}^{10} x_{\star,j} \beta_j$.
        \end{itemize}
    We then generate the potential outcomes $Y_i (t) \mid \bm{X}_{\star, i} = \bm{x}_{\star} \sim \mathrm{N} (\mu_{t} (\bm{x}_{\star}), 1)$, for $t = 0, 1$.
    \item \textbf{Condition number $\alpha \coloneqq p / n$}: 
    We set the number of adjusted covariates to be $p \coloneqq \lceil n \cdot \alpha \rceil$, where $\alpha$ ranges over $\{0.005, 0.0125, 0.025, 0.05, 0.1,$ $ \mathbf{0.15}, 0.2, 0.25, 0.3\}$ covering low-, moderate-, and relatively high-dimensional regimes. As we observe the performance deterioration of conventional regression adjustment methods such as $\hat{\tau}_{\gob}$ and $\hat{\tau}_{\gbcal}$ when $\alpha$ is as small as 0.025, we include this range to examine how the estimators behave as the covariate dimension becomes increasingly non-negligible relative to the sample size. In our setting, the value $\alpha=0.15$ corresponds to $p=p_\star=60$, at which all outcome-relevant covariates are included. Thus, when $\alpha < 0.15$, we omit certain outcome-relevant covariates; when $\alpha = 0.15$, we adjust for all outcome-relevant covariates; and when $\alpha > 0.15$,  we augment this set with additional non-prognostic covariates. Specifically, we generate additional irrelevant covariates, independently sampled from a $t (3)$ distribution. This design allows us to distinguish the consequences of omitting outcome-relevant covariates from those of adjusting for additional irrelevant covariates.
\end{itemize}

As demonstrated in Figure \ref{fig:gaussian_g0}, when $Y$ is continuous, \formbeta{} and \formmu{} yield nearly identical results when computing $\hat{\tau}_{\JASA}$ and $\hat{\tau}_{\JASACal}$. Therefore, we only report the simulation results based on \formbeta{} (Figures \ref{fig:gaussian_g1}--\ref{fig:gaussian_g2}).
\begin{itemize}
    \item Under model misspecification, $\hat{\tau}_{\gob},\hat{\tau}_{\gbcal}$ and $\hat{\tau}_{\OLS}$ exhibit substantial bias as soon as $\alpha \geq 0.025$ (or equivalently $p = 10$), whereas standard LOO estimators ($\hat{\tau}_{\loo}, \hat{\tau}_{\loocal}$, and $\hat{\tau}_{\loocalout}$), methods based on $U$-statistics ($\hat{\tau}_{\adj,2}, \hat{\tau}_{\adj,2}^{\dag}, \hat{\tau}_{\adjcal}, \hat{\tau}_{\adjcal}^{\dag}$), and our newly proposed estimators $\hat{\tau}_{\JASA},\hat{\tau}_{\JASACal}$ maintain negligible bias, comparable to that of the unadjusted estimator.
  
    \item For any estimator among $\hat{\tau}_{\adj,2}, \hat{\tau}_{\adj,2}^{\dag}, \hat{\tau}_{\adjcal}, \hat{\tau}_{\adjcal}^{\dag}$, $\hat{\tau}_{\JASA}$, and $\hat{\tau}_{\JASACal}$, the SD ratio decreases as $\alpha$ increases to 0.15, then gradually increases, while still retaining efficiency gains over the unadjusted estimator $\hat{\tau}_{\unadj}$. Here, $\alpha=0.15$ corresponds to the point at which all outcome-relevant covariates have been included; for larger values of $\alpha$, the additional covariates are non-prognostic. The mild increase in SD ratios beyond $\alpha=0.15$ reflects the variance cost of adjusting for additional irrelevant covariates. However, standard LOO estimators exhibit faster increases in the SD ratio when $\alpha > 0.15$ relative to $U$-statistic based estimators and our newly proposed estimators $\hat{\tau}_{\JASA}$ and $\hat{\tau}_{\JASACal}$.
  
    \item Calibration often further reduces the asymptotic variance, when comparing between $\hat{\tau}_{\adj,2}$ and $\hat{\tau}_{\adjcal}$, between $\hat{\tau}_{\adj,2}^{\dag}$ and $\hat{\tau}_{\adjcal}^{\dag}$, and between $\hat{\tau}_{\JASA}$ and $\hat{\tau}_{\JASACal}$. 
    We also observe that, for the standard LOO estimator, the calibrated version $\hat{\tau}_{\loocalout}$, which follows the same calibration approach as $\hat{\tau}_{\JASACal}$, outperforms $\hat{\tau}_{\loocal}$ proposed in \citet{cohen2024no} in terms of variance. We recall that Remark~\ref{rem:calibration diff} discussed the difference in these two calibration approaches.
  
    \item The SE/SD ratios of $\hat{\tau}_{\gob},\hat{\tau}_{\gbcal}$ and $\hat{\tau}_{\OLS}$ are in general less than 1 for larger values of $\alpha$, meaning that their variance estimators underestimate their variances, which in turn leads to under-coverage by the resulting Wald CIs. The SE/SD ratios of both the estimators based on $U$-statistics and our newly proposed estimators stay around 1 and attain stable CP close to the nominal $95\%$ level.
\end{itemize}

\begin{figure}[htbp]
    \centering
    \begin{subfigure}[t]{0.95\textwidth}
        \centering
        \includegraphics[width=\linewidth, page=1]{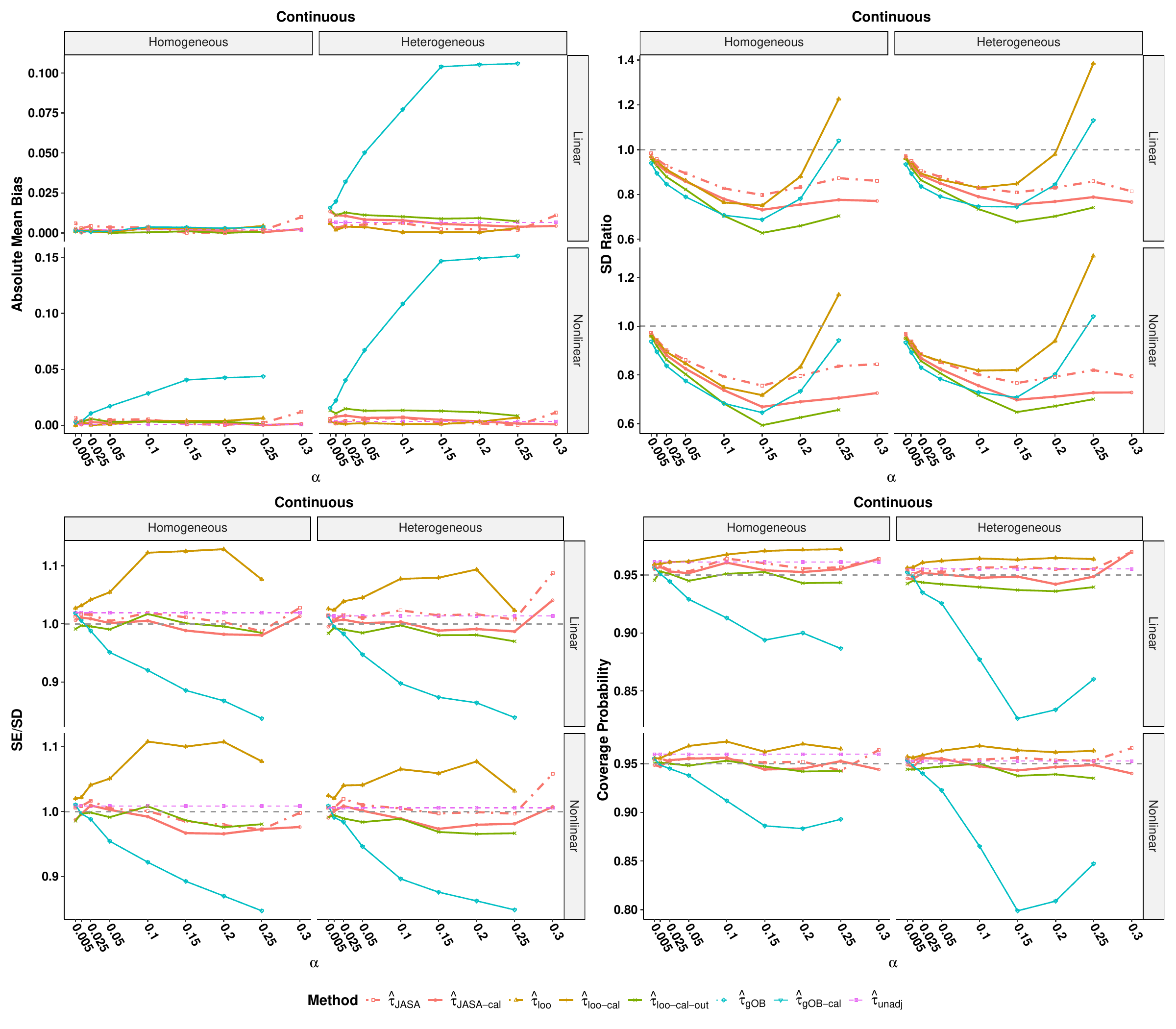}
        \caption{Comparison for continuous outcomes (Group 1)}
        \label{fig:gaussian_g1}
    \end{subfigure}

    \caption{Performance comparison for \JASA{}, \JASACal{} and other estimators with linear or nonlinear adjustment under continuous outcomes. Each panel shows varied OR in the control group, CATE and the condition number $\alpha\coloneqq p/n$. Methods with calibration are shown in solid lines. }
    \label{fig:gaussian}
\end{figure}

\begin{figure}[htbp]
    \ContinuedFloat 
    \centering 
    \begin{subfigure}[t]{0.95\textwidth}
        \centering
        \includegraphics[width=\linewidth, page=1]{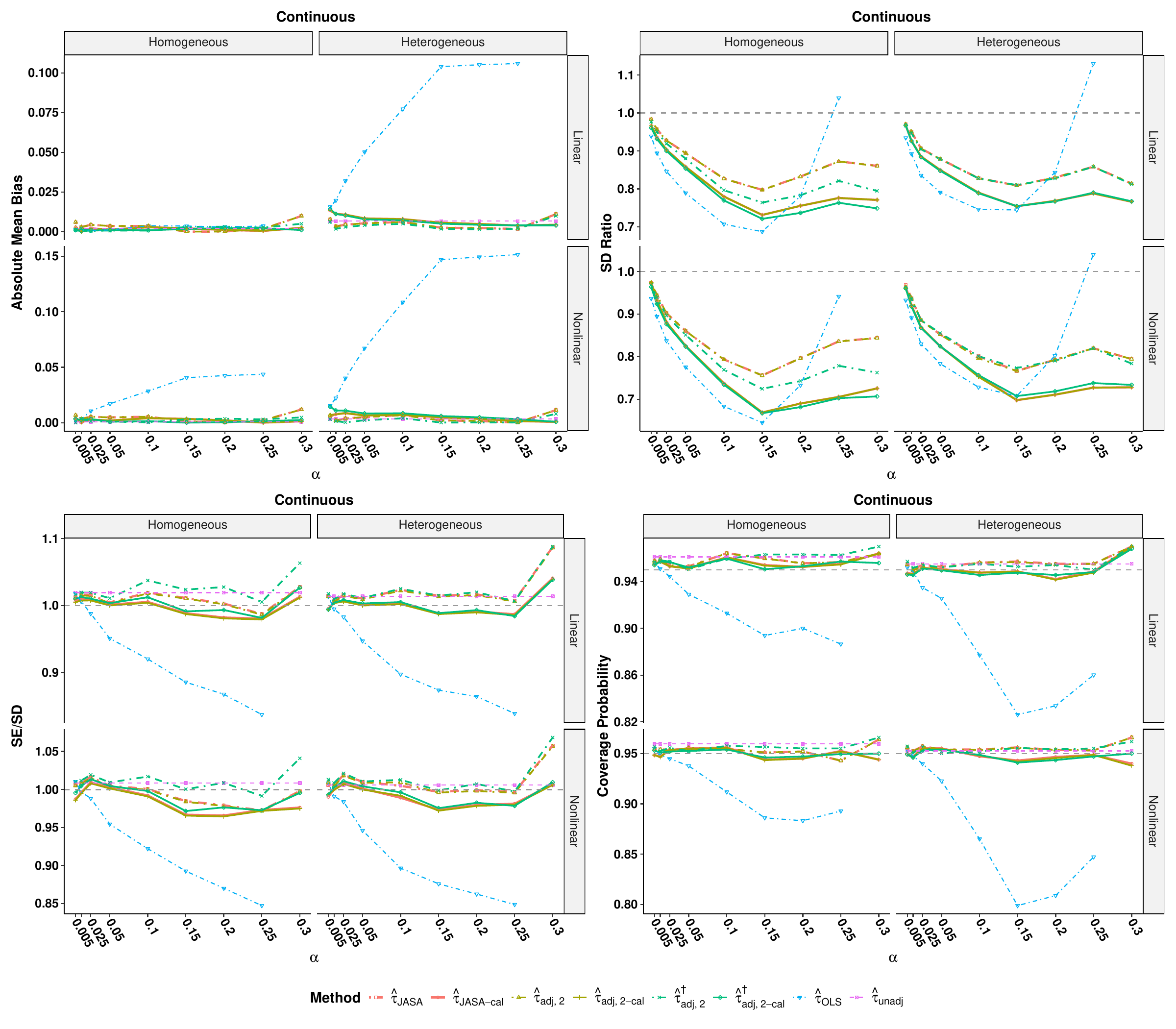}
        \caption{Comparison for continuous outcomes (Group 2)}
        \label{fig:gaussian_g2}
    \end{subfigure}
    
    \caption{Performance comparison for \JASA{}, \JASACal{} and other estimators with linear or nonlinear adjustment under continuous outcomes.  Each panel  shows varied OR in the control group, CATE and the condition number $\alpha\coloneqq p/n$. Methods with calibration are shown in solid lines. (Continued)} 
\end{figure}

\subsection{Binary outcomes}
\label{sec:sim binary}

Next, we examine the finite-sample performance of $\hat{\tau}_{\JASA}$ and $\hat{\tau}_{\JASACal}$ with various benchmarks in the case of binary outcomes. The data-generating process mirrors that of the continuous-outcome case, except that the potential outcomes are drawn from Bernoulli distributions: for $i \in [n]$ and $t\in\{0,1\}$, $Y_i (t) \mid \bm{X}_{\star, i} = \bm{x}_{\star} \sim \mathrm{Bernoulli} (g (2 \mu_{t} (\bm{x}_{\star})))$, where $g(u)=1/(1+e^{-u})$.

As illustrated in Figure~\ref{fig:binomial_g0}, solving the optimization problem \eqref{opt:general beta} (\formbeta{}) yields better performance than solving \eqref{opt:general mu} (\formmu{}) for binary outcomes, so we only compute $\hat{\tau}_{\JASA}$ and $\hat{\tau}_{\JASACal}$ in \formbeta{}. We present the median runtime in Figure~\ref{fig:binomial_g0_time}, which shows that computing \formbeta{} is indeed quite demanding.
We summarize our simulation results below (Figures~\ref{fig:binomial_g1}--\ref{fig:binomial_g2}).
\begin{itemize}
\item Even under correct model specification, $\hat{\tau}_{\gob}$, $\hat{\tau}_{\gbcal}$, and $\hat{\tau}_{\OLS}$ exhibit substantial bias even when $\alpha = 0.05$ and the bias becomes more pronounced as $\alpha$ increases. In contrast, standard LOO estimators, estimators based on $U$-statistics ($\hat{\tau}_{\adj,2}$, $\hat{\tau}_{\adj,2}^{\dag}$, $\hat{\tau}_{\adjcal}$, and $\hat{\tau}_{\adjcal}^{\dag}$), and our newly proposed $\hat{\tau}_{\JASA},\hat{\tau}_{\JASACal}$ maintain negligible bias comparable to $\hat{\tau}_{\unadj}$ across all $\alpha$ values and even under model misspecification.

\item Among estimators exhibiting negligible bias, those without calibration ($\hat{\tau}_{\loo}$, $\hat{\tau}_{\adj,2}$, and $\hat{\tau}_{\JASA}$) often have larger variance than $\hat{\tau}_{\unadj}$ even when $\alpha < 0.15$. Calibration reduces SD ratios below 1 ($\hat{\tau}_{\loocalout}$, $\hat{\tau}_{\adjcal}$, $\hat{\tau}_{\adjcal}^{\dag}$, and $\hat{\tau}_{\JASACal}$). When $\alpha>0.15$, all outcome-relevant covariates have already been included and the additional covariates are non-prognostic. $\hat{\tau}_{\JASA}$ and $\hat{\tau}_{\JASACal}$ remain nearly unbiased in this regime, while their SD ratios increase only mildly; in particular, the SD ratio of $\hat{\tau}_{\JASACal}$ remains below 1.
    
\item Among the standard LOO estimators, calibration following the approach of $\hat{\tau}_{\JASACal}$ ($\hat{\tau}_{\loocalout}$) again outperforms $\hat{\tau}_{\loocal}$ proposed in \citet{cohen2024no} in terms of variance. This further supports our recommendation that $\hat{\tau}_{\loocalout}$ should be used instead of $\hat{\tau}_{\loocal}$.
    
\item Similar to the case of continuous outcomes, for $\hat{\tau}_{\gob},\hat{\tau}_{\gbcal}$ and $\hat{\tau}_{\OLS}$, their SE/SD ratios are much smaller than 1 as $\alpha$ increases, also resulting in a decrease in CPs by their associated Wald CIs.
\end{itemize}

\begin{figure}[htbp]
    \centering    
    \begin{subfigure}[t]{0.95\textwidth}
        \centering
        \includegraphics[width=\linewidth, page=1]{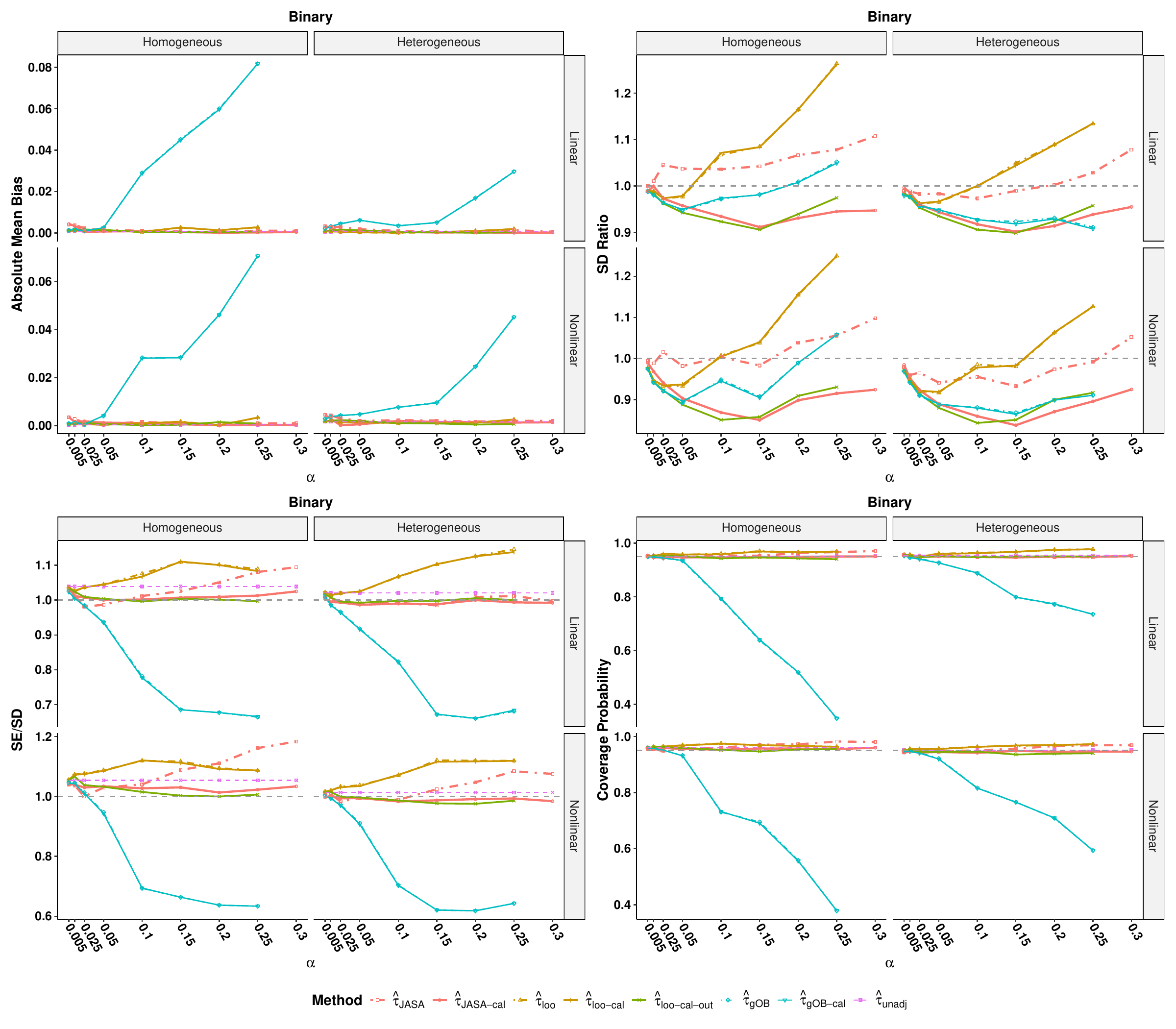}
        \caption{Comparison for binary outcomes (Group 1)}
        \label{fig:binomial_g1}
    \end{subfigure}    
    \caption{Performance comparison for \JASA{}, \JASACal{} and other estimators with linear or nonlinear adjustment under binary outcomes. Each panel  shows varied OR in the control group, CATE and the condition number $\alpha\coloneqq p/n$. Methods with calibration are shown in solid lines.} 
\end{figure}
\begin{figure}[htbp]
    \ContinuedFloat 
    \centering
    \begin{subfigure}[t]{0.95\textwidth}
        \centering
        \includegraphics[width=\linewidth, page=1]{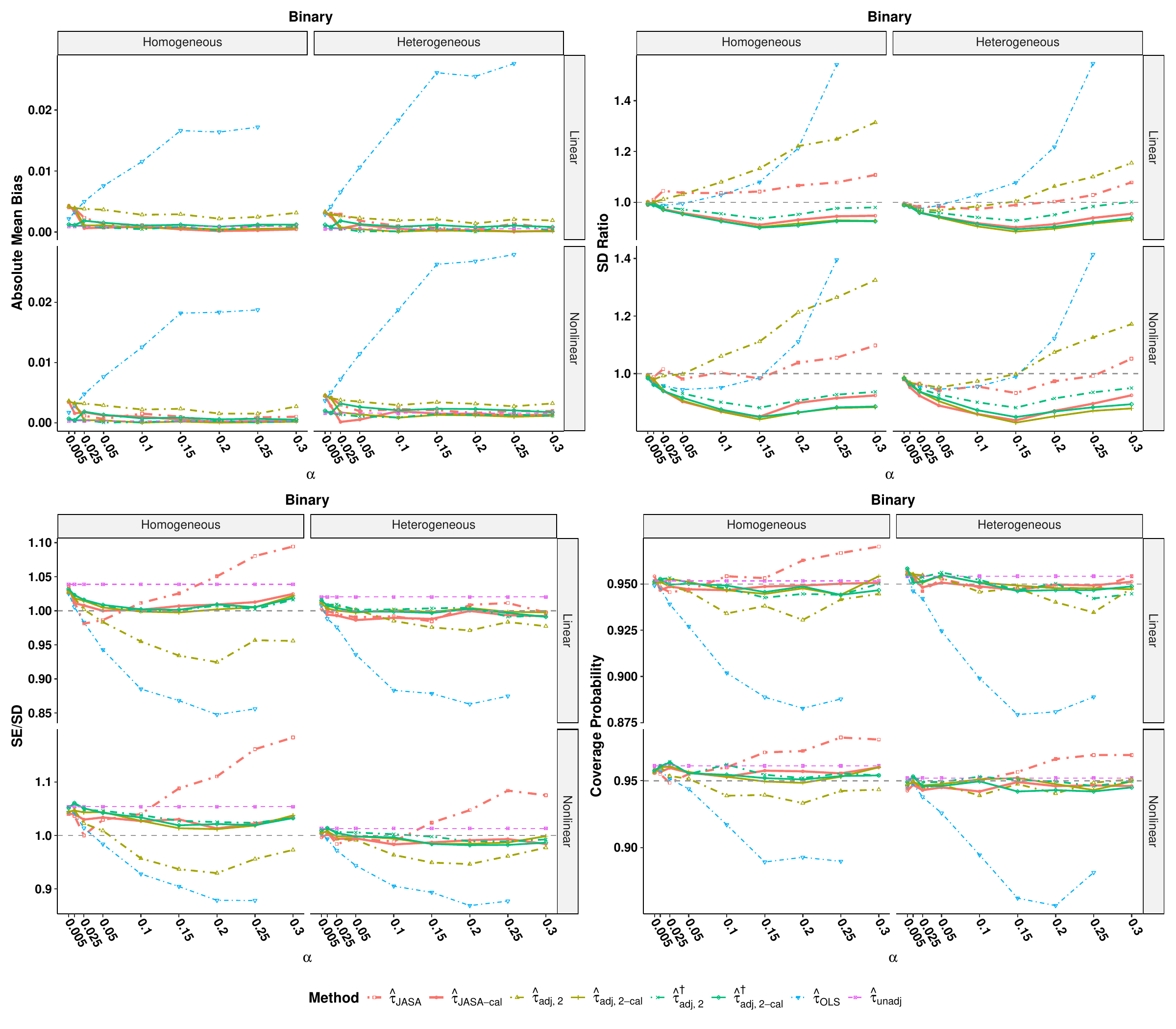}
        \caption{Comparison for binary outcomes (Group 2)}
        \label{fig:binomial_g2}
    \end{subfigure}    
    \caption{Performance comparison for \JASA{}, \JASACal{} and other estimators with linear or nonlinear adjustment under binary outcomes. Each panel  shows varied OR in the control group, CATE and the condition number $\alpha\coloneqq p/n$. Methods with calibration are shown in solid lines. (Continued)} 
\end{figure}


\subsection{Count outcomes}
\label{sec:sim count}

Finally, we examine the finite-sample performance of $\hat{\tau}_{\JASA}$ and $\hat{\tau}_{\JASACal}$ with various benchmarks in the case of count outcomes. The data-generating process follows a  structure similar to the case of continuous outcomes. We first generate the design matrix $\bbX \in \bbR^{n\times p}$. For each row $\bmX_{i}$, we generate $ \tilde{\bmX}_{i} \sim \mathcal{N}(0,\bSigma)$, where $\bSigma_{k,l}=0.1^{|k-l|},k,l \in [p]$, and then the covariates are truncated such that $\bmX_{i}=\max\{\min\{\tilde{\bmX}_{i}, 3\}, -3\}$.  We fix the covariate dimension $p = 20$. We set the true coefficient vector to $\beta_j = (-1)^j/j^{1/4},j\in[p]$, and then rescale it so that $\|\bbeta\|_2 = 1$. Then, we similarly vary the following settings in the simulation:
\begin{itemize}
    \item \textbf{The true OR in the control group} ($t=0$):
         \begin{itemize}
                    \item Linear: $\mu_0(\bm{x}) = \bm{x}^\top \bbeta$,
                    \item Nonlinear: $\mu_0(\bm{x}) = \frac{1}{2} \mathrm{sign} \left(\bm{x}^{\top} \bbeta\right)|\bm{x}^{\top} \bbeta|^{\frac{1}{3}} + \cos(\bm{x}^{\top} \bbeta) + \bm{x}^{\top} \bbeta$.
          \end{itemize}
    \item \textbf{The true CATE}:
        \begin{itemize}
            \item Homogeneous: $\mu_1(\bm{x}) - \mu_0(\bm{x}) \equiv 1$,
            \item Heterogeneous: $\mu_1(\bm{x}) - \mu_0(\bm{x}) =1 - \frac{1}{2}\min\{x_{1}^{2}, 5\} + \frac{1}{4}\bm{x}^{\top} \bbeta$.
        \end{itemize}
       We then generate the potential outcomes $Y_i(t)|\bm{X}_{i}=\bm{x} \sim \mathrm{Poisson}(g({\mu}_{t}(\bm{x})))$, where $g(u)=\exp\left(\min\{u, 4\}\right)$.
     \item \textbf{Sample size} $n$: We vary the sample size $n\in \{200, 252, 317, 399, 502, 631, 795, 1000\}.$ 
\end{itemize}

As illustrated in Figure~\ref{fig:poisson_g0}, solving the optimization problem \eqref{opt:general mu} (\formmu{}) exhibits better performance than solving \eqref{opt:general beta} (\formbeta{}) for count outcomes, so we only present $\hat{\tau}_{\JASA}$ and $\hat{\tau}_{\JASACal}$ in \formmu{} in Figures~\ref{fig:poisson_g1}--\ref{fig:poisson_g2}. We also present the median runtime of both approaches in Figure~\ref{fig:poisson_g0_time}, and for the \formmu{} approach it generally takes less than 1 minute to obtain the results.
We summarize our simulation results below.
\begin{itemize}
\item $\hat{\tau}_{\gob}$ and $\hat{\tau}_{\gbcal}$ again exhibit substantial bias as $n$ decreases, while $\hat{\tau}_{\OLS}$ keeps the bias well controlled across different sample sizes; but in Appendix~\ref{simu: count suppl infla}, we give an example where $\hat{\tau}_{\OLS}$ also exhibits substantial bias when  $p$ is large compared to $n$.

\item For small $n$, $\hat{\tau}_{\loo}$ and $\hat{\tau}_{\loocal}$ are less efficient than $\hat{\tau}_{\unadj}$, while $\hat{\tau}_{\loocalout}$ improves efficiency. Meanwhile, calibration consistently reduces the variance below that of $\hat{\tau}_{\unadj}$. In particular, $\hat{\tau}_{\JASACal}$ exhibits superior efficiency over both $\hat{\tau}_{\adjcal}$ and $\hat{\tau}_{\adjcal}^{\dag}$ in all scenarios.

\item $\hat{\tau}_{\JASACal}$ maintains stable SE/SD ratios near 1 for all $n$, ensuring nominal CPs. In contrast, $\hat{\tau}_{\loocalout}$ shows persistently low SE/SD ratios, resulting in poor coverage. Similarly, $\hat{\tau}_{\OLS}$ only achieves appropriate SE/SD ratios at larger $n$.
\end{itemize}

In summary, in the simulation studies, we demonstrate the competitive finite-sample performance of $\hat{\tau}_{\JASA}$ and $\hat{\tau}_{\JASACal}$, in particular $\hat{\tau}_{\JASACal}$, compared to existing adjusted estimators and standard LOO estimators, and even the estimators based on $U$-statistics. For standard calibrated LOO estimators $\hat{\tau}_{\loocal}$ and $\hat{\tau}_{\loocalout}$, our analysis finds that $\hat{\tau}_{\loocalout}$, using the same calibration strategy as $\hat{\tau}_{\JASACal}$, consistently outperforms $\hat{\tau}_{\loocal}$, which was proposed in \citet{cohen2024no}.

\begin{figure}[htbp]
    \centering
    \begin{subfigure}[t]{0.95\textwidth}
        \centering
        \includegraphics[width=\linewidth, page=1]{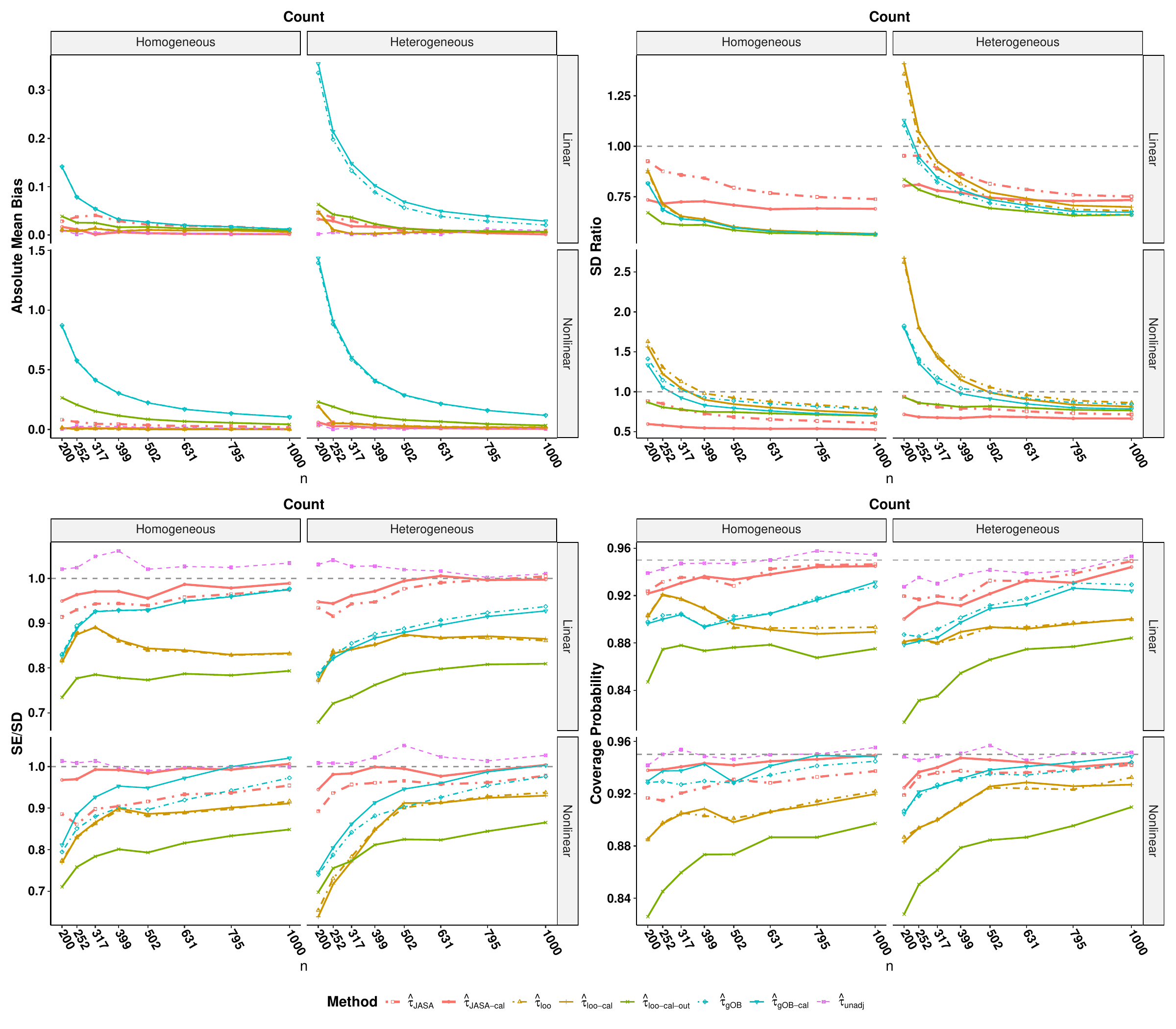}
        \caption{Comparison for count outcomes (Group 1)}
        \label{fig:poisson_g1}
    \end{subfigure}

    \caption{Performance comparison for \JASA{}, \JASACal{} and other estimators with linear or nonlinear adjustment under count outcomes. Each panel shows varied OR in the control group, CATE and the sample size $n$. Methods with calibration are shown in solid lines. }
\end{figure}

\begin{figure}[htbp]
    \ContinuedFloat 
    \centering 
    \begin{subfigure}[t]{0.95\textwidth}
        \centering
        \includegraphics[width=\linewidth, page=1]{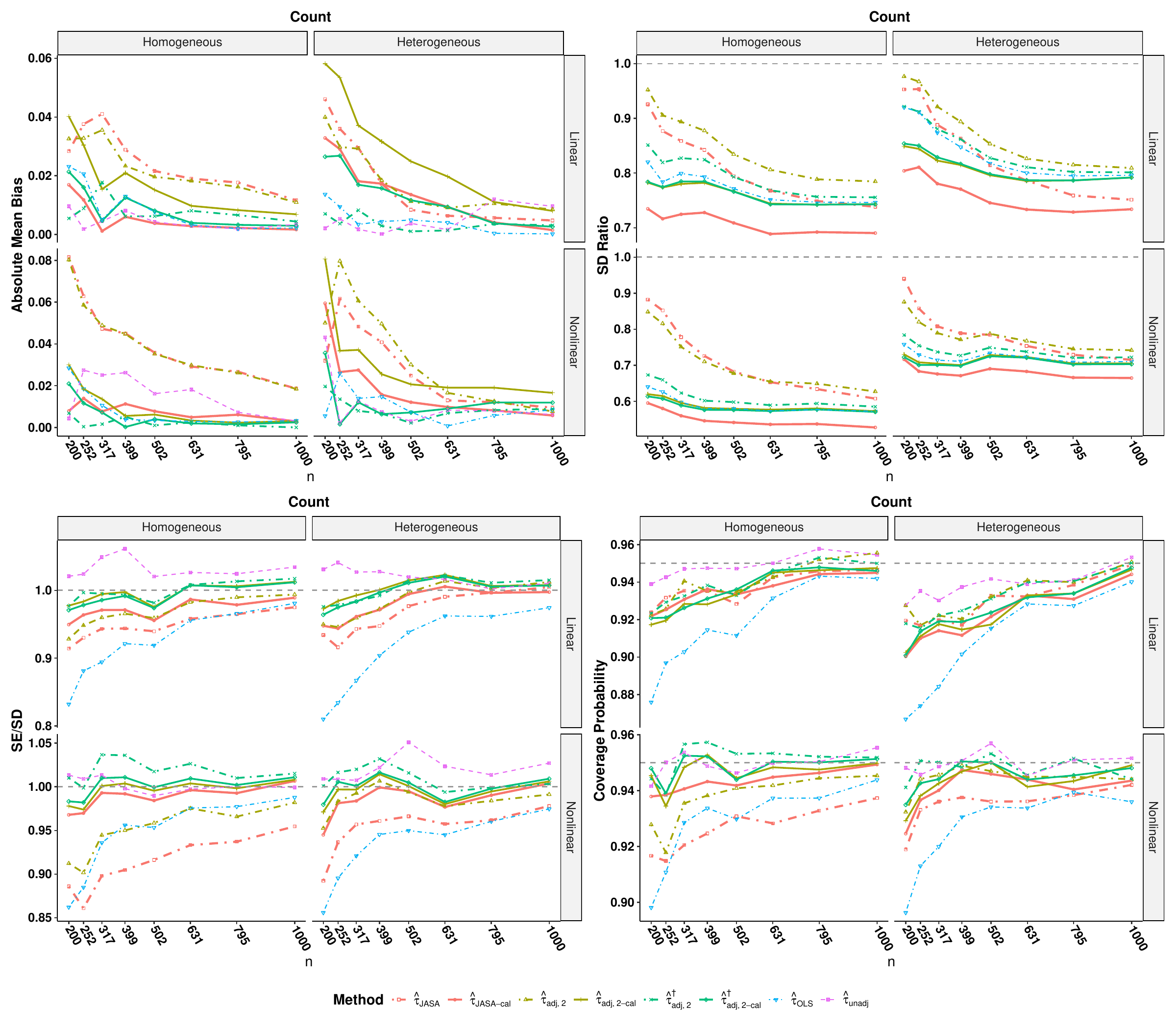}
        \caption{Comparison for count outcomes (Group 2)}
        \label{fig:poisson_g2}
    \end{subfigure}
    
    \caption{Performance comparison for \JASA{}, \JASACal{} and other estimators with linear or nonlinear adjustment under count outcomes. Each panel  shows varied OR in the control group, CATE and the sample size $n$. Methods with calibration are shown in solid lines. (Continued)} 
\end{figure}

\section{Real Data Application}
\label{sec:real}
We demonstrate the benefits of covariate adjustment by comparing our methods with other approaches using an Alzheimer's disease (AD)-related RCT. This multicenter, randomized, double-blind trial (2003--2006) investigated whether 18-month high-dose B-vitamin supplement slows cognitive decline in AD patients \cite{ADC_016} (ClinicalTrials.gov ID: NCT00056225). Individuals were randomly assigned to two groups: 60\% were treated with daily high-dose supplements and 40\% were treated with placebo.  Finally, a total of 340 participants completed the trial, of whom 202 were in the treatment  group and 138 in the placebo group. The primary outcome was the 18-month change rate in Alzheimer's Disease Assessment Scale-cognitive subscale (ADAS-cog) score~\cite{rosen1984new}. This performance-based measure assesses multiple cognitive domains, with scores ranging from 0--70 (higher scores indicate greater impairment). 
Secondary outcomes, including the Mini-Mental State Examination (MMSE) and Alzheimer's Disease Cooperative Study Activities of Daily Living (ADCS-ADL) among others, were assessed at each visit. Baseline data include demographics (e.g., age, sex, education) and laboratory measures (e.g., homocysteine, ApoE4 status, hematological parameters).

In the original analysis\citep{ADC_016}, only the baseline ADAS-cog and age were included as covariates. 
Here, we retain these two prespecified covariates and sequentially incorporate 48 additional covariates that are selected using a stepwise selection procedure based on statistical significance. Specifically, covariates were added sequentially in ascending order based on p-values, which were obtained from multiple linear regression models that always included the baseline ADAS-cog and age. It should be emphasized that this outcome-based covariate selection procedure is not covered by our theoretical framework and is not standard practice for confirmatory RCT analyses. Accordingly, this exercise is intended only as a proof of concept to illustrate the empirical behavior of the adjustment methods as the covariate dimension increases, rather than as a recommended variable-selection-and-adjustment strategy for practice.

\begin{figure}[htbp]
    \centering
    \includegraphics[width=0.95\linewidth]{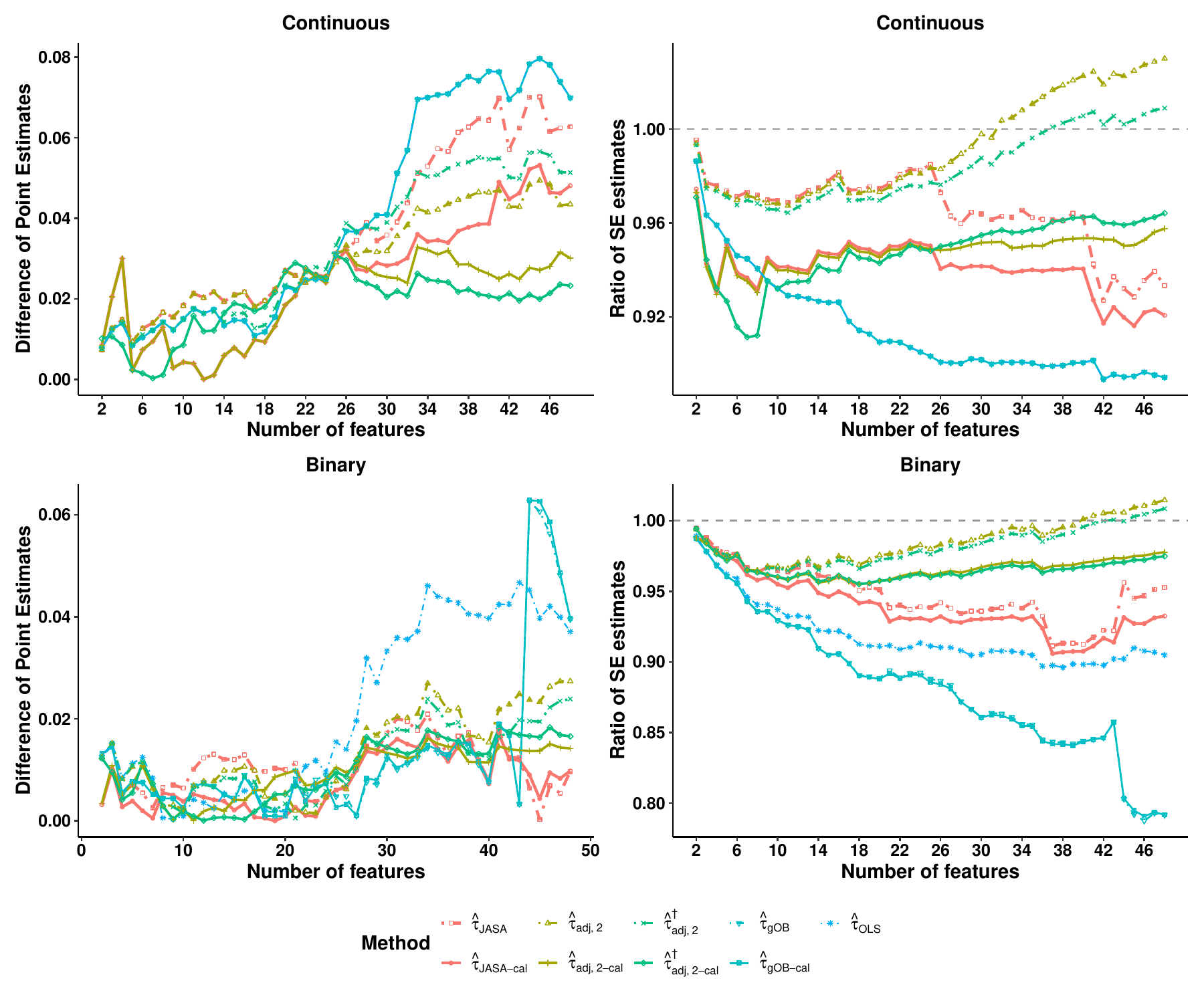}
    \caption{Performance metrics for \JASA{}, \JASACal{} and other adjusted estimators under continuous (top row) and binary (bottom row) outcomes. The first column displays the absolute difference in ATE estimate between each estimator and the unadjusted estimator. The second column shows the ratio of standard error (SE) estimate for each estimator relative to the unadjusted estimator.}
    \label{fig:appl_ADC016}
\end{figure}

In the first row of Figure~\ref{fig:appl_ADC016}, we report the absolute difference between various adjusted estimators and $\hat{\tau}_{\unadj}$, and the ratio of the estimated standard errors of the adjusted estimators to that of $\hat{\tau}_{\unadj}$. It should be noted that it is not standard practice in RCT data analysis to perform variable selection first and then adjust for the selected variables using adjusted estimators. It only serves as a proof-of-concept to illustrate the performance of our newly proposed estimators on real data.

As the number of covariates increases, the differences between the adjusted estimators and $\hat{\tau}_{\unadj}$ become greater. Among them, $\hat{\tau}_{\OLS}$ (which is equivalent to $\hat{\tau}_{\gob}$ and $\hat{\tau}_{\gbcal}$ under identity link) and $\hat{\tau}_{\JASA}$ exhibit the largest difference. 
We believe that the discrepancy between $\hat{\tau}_{\JASA}$ and $\hat{\tau}_{\unadj}$ stems from the variance of $\hat{\tau}_{\JASA}$, whereas  $\hat{\tau}_{\JASACal}$  shows a smaller discrepancy.
Meanwhile, the ratios between the SEs of most estimators and that of $\hat{\tau}_{\unadj}$ remain stable across a wide range of $p$, whereas the corresponding ratio for $\hat{\tau}_{\OLS}$ continues to decline. This reduction in the estimated variance of $\hat{\tau}_{\OLS}$ can be partially attributed to the underestimation of variances, as observed in our simulation study. In general, calibrated adjusted estimators, $\hat{\tau}_{\JASACal}$, $\hat{\tau}_{\adjcal}$ and $\hat{\tau}_{\adjcal}^{\dag}$, have the lowest standard errors.

Furthermore, we also consider the case of binary outcomes, defined as 1 if the ADAS-cog score decreases from the baseline and 0 otherwise. Covariates were chosen in a manner analogous to the case of continuous outcomes, with logistic regression replacing linear regression. The results are shown in the second row of Figure \ref{fig:appl_ADC016}. The message is largely similar, so we do not repeat ourselves. However, it is worth noting that, as more covariates are incorporated, the absolute differences between $\hat{\tau}_{\gob}$ and $\hat{\tau}_{\OLS}$, $\hat{\tau}_{\gbcal}$ and $\hat{\tau}_{\OLS}$ become more pronounced. In contrast, $\hat{\tau}_{\JASACal}$ maintains a stable absolute difference and a smaller variance than $\hat{\tau}_{\unadj}$, $\hat{\tau}_{\adjcal}$, and $\hat{\tau}_{\adjcal}^{\dag}$.

\section{Conclusions}
\label{sec:conclusion}

With rapid advancements in data collection technologies, future RCTs are expected to collect higher dimensional covariates. However, as we mentioned, there is little direction from regulatory agencies on how to adjust for those large-dimensional covariates. When ORs are fit by linear working models, promising progress has been made in the past few years, rendering the estimators based on $U$-statistics a more robust alternative to the OLS adjusted estimator $\hat{\tau}_{\OLS}$ when $p$ is large \citep{lei2021regression, chang2024exact, zhao2024hoif, lu2025debiased, gu2025assumption, abadie2025unbiased}. To the best of our knowledge, parallel results when ORs are fit by GLMs do not exist yet. In this paper, we make progress along this line by proposing novel procedures to solve jackknife score equations or jackknife constrained optimization problems derived from the score equations, inspired by the $U$-statistic construction using linear working models \citep{chang2024exact, zhao2024hoif, lu2025debiased}. Our new jackknife-based procedures, coined as \JASA{} and \JASACal{}, also place the aforementioned $U$-statistic estimators into a more unified framework. Extensive simulation studies and an analysis of a completed RCT demonstrate the competitive finite-sample performance of our new procedures. Based on our simulation studies, we summarize in Table~\ref{tab:my-table} several empirical findings regarding the relative performance of different adjusted estimators. In our simulation settings, 
when the ratio $\alpha = p/n$ is relatively small (say less than 0.025), standard OLS adjusted estimators or calibrated gOB estimators generally perform well; but as soon as $\alpha > 0.025$, $\hat{\tau}_{\JASACal}$ or $U$-statistic-based estimators such as $\hat{\tau}_{\adj, 2\mbox{-}\cal}^{\dag}$ tend to exhibit more robust performance in practice. We emphasize that these numerical cutoffs are specific to our simulation designs and may not be directly interpreted as universal thresholds. In particular, we view \JASA{} and \JASACal{} as a helpful alternative for practitioners if they decide to adjust for a large number of covariates using GLMs in RCT data analysis.

We conclude the paper by listing a few future directions. First, when the working model is a canonical GLM, it is important to characterize the precise conditions under which $\hat{\tau}_{\JASA}$ or $\hat{\tau}_{\JASACal}$ is $\sqrt{n}$-CAN, in particular when $p$ increases with $n$. Second, it is also important for practice to develop parallel approaches for covariate adaptive designs \citep{ma2024new}, time-to-event outcomes \citep{van2021principled, ye2024covariate}, and other natural parameterizations of treatment effects for binary outcomes \citep{vansteelandt2022assumption, richardson2017modeling}. Third, to the best of our knowledge, when general machine learning methods are used to fit the OR, sample splitting is still the only viable approach to building adjusted estimators that avoid excessive bias. In the context of RCT, however, it is more desirable to avoid sample splitting for reproducibility purposes, a philosophy espoused by $\hat{\tau}_{\JASA}$, $\hat{\tau}_{\adj, 2}$, $\hat{\tau}_{\adj, 2}^{\dag}$, and their calibrated versions. 
It is still unclear how to apply the jackknife score equations for GLMs developed here to replace the sample splitting strategy when fitting more general machine learning methods. An exception is when the empirical loss function can be decomposed into two components, one involving only covariates $\bm{X}$, and the other also involving the treatment and outcome variables. It is also important to investigate how to improve the runtime performance of \JASA{} or \JASACal{} especially when \formbeta{} is employed, as emphasized in Remark~\ref{rem:computation}. Finally, and probably more importantly, our theoretical guarantees on \JASA{} and \JASACal{} are still preliminary and are only available for the version that exactly solves the jackknife score equations. In our ongoing work, we are still investigating the theoretical properties of the actually implemented methods in practice.

\begin{table}[htbp]
  \centering  
  \begin{threeparttable}
  \caption{Empirical performance patterns observed in the simulation settings under different outcome types and condition number $\alpha \coloneqq p / n$. }
  \label{tab:my-table}
    \begin{tabular}{@{}c|ccc@{}}
      \toprule
      \textbf{Outcome} & \textbf{Continuous} & \textbf{Binary} & \textbf{Count} \\ \midrule
      $\alpha$ very small \tnote{*}& $\hat{\tau}_{\OLS}$ & $\hat{\tau}_{\gbcal}$ & $\hat{\tau}_{\gbcal}$ \\
      $\alpha$ not very small & $\hat{\tau}_{\JASACal}$\tnote{a} \, or $\hat{\tau}_{\adjcal}^{\dag}$ & $\hat{\tau}_{\JASACal}$\tnote{b} & $\hat{\tau}_{\JASACal}$\tnote{c} \\ \bottomrule
    \end{tabular}
    \begin{tablenotes}
      \item[*] Inflated bias can be observed when $\alpha\geq 0.025$ for $\hat{\tau}_{\OLS}$ under continuous outcome (Figure \ref{fig:gaussian_g2}), $\alpha\geq0.05$ for $\hat{\tau}_{\gbcal}$ under binary/count outcome (Figure \ref{fig:binomial_g1} and \ref{fig:poisson_g1}).
      \item[a] Similar performance between \formbeta{} and \formmu{}.
      \item[b] Prefer \formbeta{}.
      \item[c] Prefer \formmu{}.
    \end{tablenotes}
  \end{threeparttable}
\end{table}


\begin{algorithm}[htbp]
\caption{Pseudocode of \JASA{}}
\label{alg:jasa}
\begin{algorithmic}[0]
\normalfont
\renewcommand{\algorithmicrequire}{\textbf{Input:}}
\renewcommand{\algorithmicensure}{\textbf{Output:}}

\Require Dataset  $\mathcal{D} = \{(\bm{X}_i,Y_i,T_i)\}_{i=1}^n$ where $\bm{X}_i \in \bbR^{p}$ and  $T_{i} \overset{\rm i.i.d.}{\sim} \text{Bernoulli} (\pi_{1})$. Let $\bm{Z}_i=(1,\bm{X}_i^\top)^\top$.

\Ensure ATE estimates $\hat{\tau}_{\JASA}, \hat{\tau}_{\JASACal}$ with variance estimates.

\State
\textbf{Step 1: Initial Working Model Estimation} 
\State Estimate the initial coefficient $\hat{\bbeta}_{t}$ of the outcome regression model for arm $t\in\{0, 1\}$, via GLM using implicit and explicit bias reduction methods \citep{kosmidis2020mean, zhang2026bias}. 

\State
\textbf{Step 2: Jackknife Score-based Adjustment} 
\State For each sample $i\in [n]$, and  $t\in \{0, 1\}$,

\State \hspace{2em} \textbf{2a.} Check the existence of solution: Solve the optimization using \texttt{optim} package:
\State   
    \begin{align*}
        \hat{\bm \mu}_t^{(-i)}  \coloneqq &\; \underset{\bm{\mu}_t^{(-i)}:\bm{\mu}_t^{(-i)}\in [m,M]^{n}}{\argmin} f(\bm{\mu}_t^{(-i)}), 
    \end{align*}
\State \hspace{4em} where $f(\bm{\mu}_t^{(-i)})\coloneqq \|\frac{1}{n-1}\sum_{j\neq i}\frac{\mathbbm{1}\{T_j=t\}}{\hat{\pi}_t}Y_j\bm{Z}_j-\frac{1}{n}\sum_{l=1}^{n} \mu_{t,l} \bm{Z}_{l}\|_2^{2}$ and $m, M$ are determined by the 
\State \hspace{4em} outcome type.
\State \hspace{4em} If $\hat{\tau}_{\JASA}$ is estimated through \formmu{}, steps 2b-2d can be omitted.
    
\State \hspace{2em} \textbf{2b.} If $f(\hat{\bm{\mu}}_{t}^{(-i)})\leq \epsilon$
, then solve the nonlinear score equation via \texttt{BBsolve} package to get $\hat{\bbeta}_{t}^{(-i)}$:
\State
    \begin{align*}
        \frac{1}{n - 1} \sum_{j \neq i} \frac{\mathbbm{1} \{T_{j} = t \}}{\hat{\pi}_{t}} Y_{j} \bm{Z}_{j}- \frac{1}{n} \sum_{l = 1}^{n} g (\bm{Z}_{l}^\top \bbeta_{t}^{(-i)}) \bm{Z}_{l} = \mathbf{0}.
    \end{align*}
    
\State \hspace{2em} \textbf{2c.} If $f(\hat{\bm{\mu}}_{t}^{(-i)}) > \epsilon$ or \texttt{BBsolve} does not converge, then solve the constrained optimization: 
\State
    \begin{align*}
    & \hat{\bbeta}_{t}^{(-i)}  \coloneqq \underset{\bbeta_{t}^{(-i)}: m_t \leq \bbeta_t^{(-i)} \leq  M_t}{\argmin} \Big\| \frac{1}{n - 1} \sum_{j \neq i} \frac{\mathbbm{1} \{T_{j} = t \}}{\hat{\pi}_{t}} Y_{j} \bm{Z}_{j}- \frac{1}{n} \sum_{l = 1}^{n} g (\bm{Z}_{l}^\top \bbeta_{t}^{(-i)}) \bm{Z}_{l} \Big\|_2^{2},
    \end{align*} 
\State \hspace{4em} where $[m_t,M_t] = [\min(\hat{\bbeta}_t),\max(\hat{\bbeta}_t)]$. Bounds $[m_t,M_t]$ may alternatively be initialized using $\hat{\bbeta}_{t,\rm init}^{(-i)}$ estimated without
\State \hspace{4em} $(T_{i},Y_{i})$ for each $i$; initialization with $\hat{\bbeta}_t$ is computationally more efficient, and it does not lead to significant bias 
\State \hspace{4em} inflation empirically based on our simulation results, with a theoretical justification left to future work.

\State \hspace{2em} \textbf{2d.} Compute $\hat{\mu}_{t,i}^{(-i)} = g(\bm{Z}_i^\top \hat{\bbeta}_t^{(-i)})$.

\State \textbf{Step 3: Linear Calibration}
\State \hspace{2em} \textbf{3a.} Solve the linear regression within each arm $t\in\{0,1\}$:
    \begin{align*}
    (\hat{\gamma}_{t}, \hat{\gamma}_{t,0}, \hat{\gamma}_{t,1}) = \underset{\gamma_{t}, \gamma_{t,0}, \gamma_{t,1}}{\argmin} \sum_{i: T_i = t}\left(Y_i - \gamma_{t} - \hat{\mu}_{0,i}^{(-i)}\gamma_{t,0} - \hat{\mu}_{1,i}^{(-i)}\gamma_{t,1}\right)^2. 
    \end{align*}
\State \hspace{4em} Set $\tilde{\mu}_{t,i}=\hat{\gamma}_t + \hat{\mu}_{0,i}^{(-i)}\hat{\gamma}_{t,0} + \hat{\mu}_{1,i}^{(-i)}\hat{\gamma}_{t,1}$. 
    
\State \hspace{2em} \textbf{3b.} If $p \leq c_0$, apply the conditional calibration strategy (Appendix \ref{sec:cond calib}).
 
\State \hspace{2em} \textbf{3c.} Compute $\hat{\tau}_{\JASA},\hat{\tau}_{\JASACal}$ using formulae (\ref{equ:JASA}) and  (\ref{equ:JASA-cal}).
\State \hspace{3.5em} Compute $\hat{\sigma}^2_{\JASA}, \hat{\sigma}^2_{\JASACal}$ using formulae (\ref{equ: var_jasa}) and (\ref{equ: var_jasa-cal}). 
\end{algorithmic}  
\end{algorithm}

\clearpage

\bmsection*{Acknowledgments}
The authors are grateful to the two anonymous reviewers for their highly constructive comments that have helped us significantly improve the paper. The authors also thank \href{https://maweiruc.github.io/}{Wei Ma}, \href{https://math.sjtu.edu.cn/faculty/chengwang/}{Cheng Wang} and Xin Zhang (Pfizer) for helpful discussions and thank \href{https://research.ugent.be/web/person/muluneh-alene-addis-0/en}{Muluneh Alene Addis}, \href{https://research.ugent.be/web/person/kelly-van-lancker-0/en}{Kelly Van Lancker}, and \href{https://users.ugent.be/~svsteela/}{Stijn Vansteelandt} for very encouraging feedback on the practical performance of our proposed method.

The authors acknowledge the Alzheimer's Disease Cooperative Study (ADCS) and the University of California, San Diego (UCSD) as the source of the data used in this study. This work was supported by the National Institute on Aging (NIA) of the National Institutes of Health under award number U19 AG010483. Data used in preparation of this manuscript/publication/article were obtained from the University of California, San Diego Alzheimer's Disease Cooperative Study. Consequently, ADCS Core Directors contributed to the original ADCS studies and/or provided data, but did not participate in the analyses or the writing of this manuscript/report. The authors acknowledge Xiangfu Laboratory and the jCloud platform of Shanghai Jiao Tong University for providing computational resources.

\bmsection*{Financial disclosure}
This work was partially supported by the National Science Foundations of China Grants No.12471274 (XW, SZ, LL), National Key R \& D Program (2025YFA1016700) (XW, SZ, LL), the Neil Shen's SJTU Medical Research Fund (XW, HL, LL), the Science and Technology Commission of Shanghai Municipality (STCSM) (Grant No. 23JS1400700; 24JS2840200; 25JS2850100) (HL), and the Shanghai Municipal Health Commission (2025ZHYL022) (HL).

\bmsection*{Conflict of interest}
The authors declare no potential conflict of interests.

\bmsection*{Data Availability Statement}
The data that support the findings of this study are available from the Alzheimer's Disease Cooperative Study (ADCS). Restrictions apply to the availability of these data, which were used under license for this study. Data are available from \href{https://www.adcs.org/data-sharing/}{https://www.adcs.org/data-sharing/} with the ADCS. The data analyzed in this study corresponds to the protocol number ADC-016.

\bibliography{wileyNJD-AMA}

\bmsection*{Supporting information}
Additional supporting information can be found in the Appendix. The replication code for our simulation studies can be found in the accompanying \href{https://github.com/Cinbo-Wang/Simu_JASA/tree/main}{GitHub page}.

\appendix

\bmsection{The Calibration Procedure in the Actual Implementation} \label{sec:cond calib}
\vspace*{12pt}

To stabilize the calibrated estimators, we employ a conditional calibration approach when $p$ is small. Specifically, if $p \leq c_0$ (in the simulation, we set $c_0=10$), we define the range $w\coloneqq \max\{Y_i, i\in [n]\} - \min\{Y_i, i\in [n]\}$, and set the bounds $l=\min\{Y_i, i\in [n]\} - 0.1w$ and $u=\max\{Y_i, i\in [n]\} + 0.1w$. If there exists an index $i_0\in[n]$, such that $\min\{\tilde{\mu}_{0,i_0}, \tilde{\mu}_{1,i_0}\} < l$ or $\max\{\tilde{\mu}_{0,i_0}, \tilde{\mu}_{1,i_0}\} > u$, we omit the calibration step and set $\tilde{\mu}_{t,i}=\hat{\mu}_{t,i}^{(-i)}$ for all $i\in[n]$ and $t\in\{0,1\}$.
 
To elaborate on the effect of ``conditional'' calibration, we define $\hat{\tau}_{\JASACal}^{\dagger}$ as the \JASACal{} estimator that directly incorporates linear calibration and compare it with $\hat{\tau}_{\JASACal}$ (both based on \formbeta{}). We consider the simulation setup for continuous (Section~\ref{sec:sim continuous}), binary (Section~\ref{sec:sim binary}) and count (Appendix~\ref{simu: count suppl infla}) outcomes.  To simplify the presentation, we consider only a subset of scenarios, specifically those with $\mu_0(\bx_{\star})=\frac{1}{2} \mathrm{sign} (\bm{x}_\star^{\top} \bbeta) |\bm{x}_\star^{\top} \bbeta|^{\frac{1}{3}} + \cos (\bm{x}_\star^{\top} \bbeta) + \bm{x}_\star^{\top} \bbeta$ and $\mu_1(\bx_\star)-\mu_0(\bx_\star)\equiv 1$, and $\alpha \in \{0.005, 0.0125, 0.025, 0.05\}$. As illustrated in Figure \ref{fig:comp_jasacal_direct}, $\hat{\tau}_{\JASACal}^{\dag}$ exhibits a SD ratio higher than 1  for all outcome types when $\alpha=0.005$ (or $ p=2$), along with higher bias relative to $\hat{\tau}_{\JASACal}$. This difference diminishes as $\alpha$ increases.

\begin{figure}[htpb]
    \centering
    \includegraphics[width=0.95\linewidth]{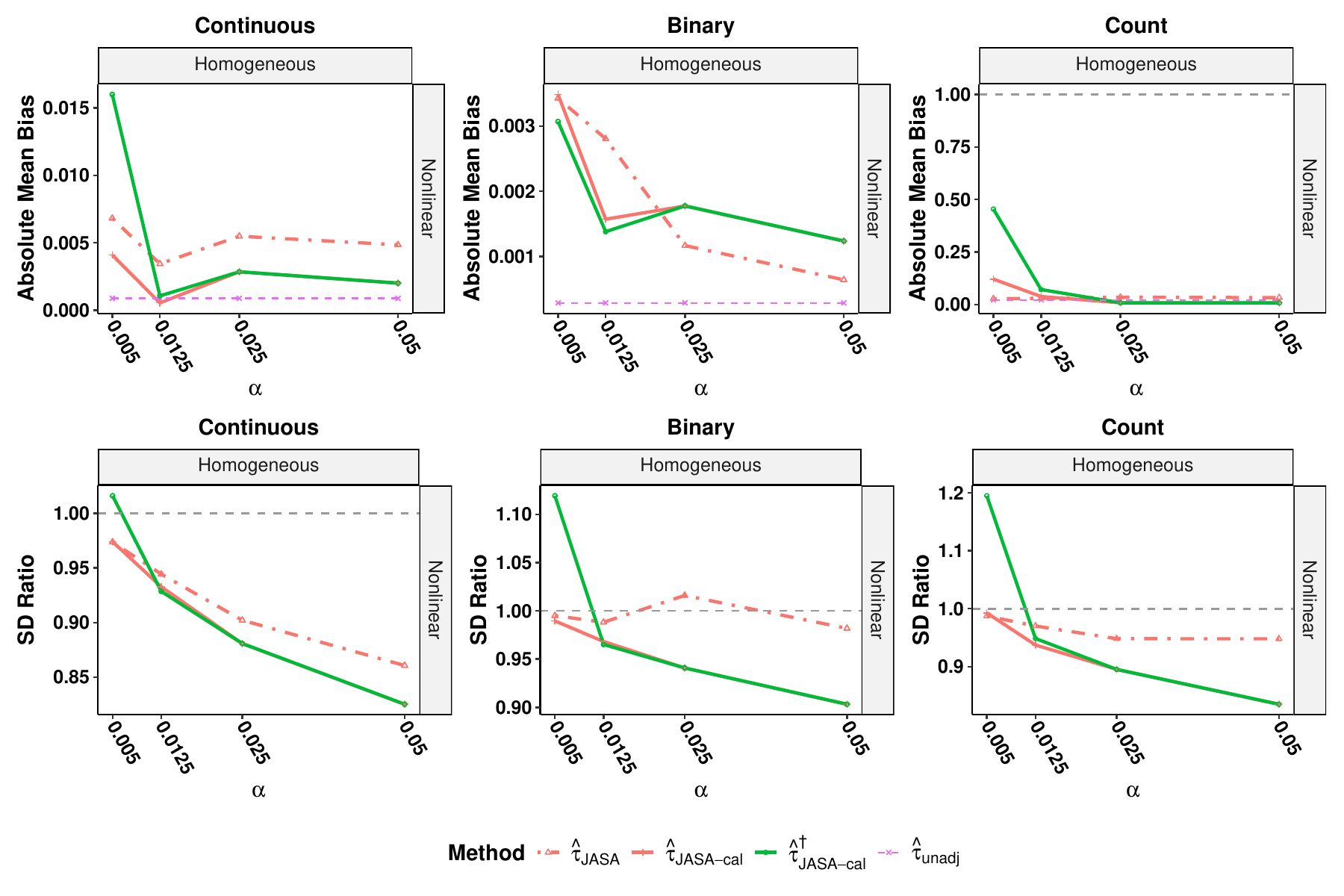}
    \caption{Bias and SD ratio comparison between the $\hat{\tau}_{\JASACal}$ estimator and its direct counterpart $\hat{\tau}_{\JASACal}^{\dagger}$.}
    \label{fig:comp_jasacal_direct}
\end{figure}

To investigate the problem of SD inflation after calibration with $\alpha=0.005$, we present an example from the  simulation data for each outcome type. We denote $\{\hat{\mu}_{t,i}^{(-i)}\}_{i=1}^{n}$ as the OR estimates by solving the jackknife score equation and $\{\tilde{\mu}_{t,i}\}_{i=1}^{n}$ as the  prediction after linear calibration. As shown in Figure \ref{fig:example_extrem_ols}, there are extreme predictions in $\{\tilde{\mu}_{t,i}\}_{i=1}^{n}$ for $t\in\{0,1\}$,  compared to  the observed outcome $Y_i$. These extreme values further lead to $\hat{\tau}_{\JASACal}^{\dag}$ exhibiting increased bias and SD.

\begin{figure}[htpb]
    \centering
    \includegraphics[width=0.95\linewidth]{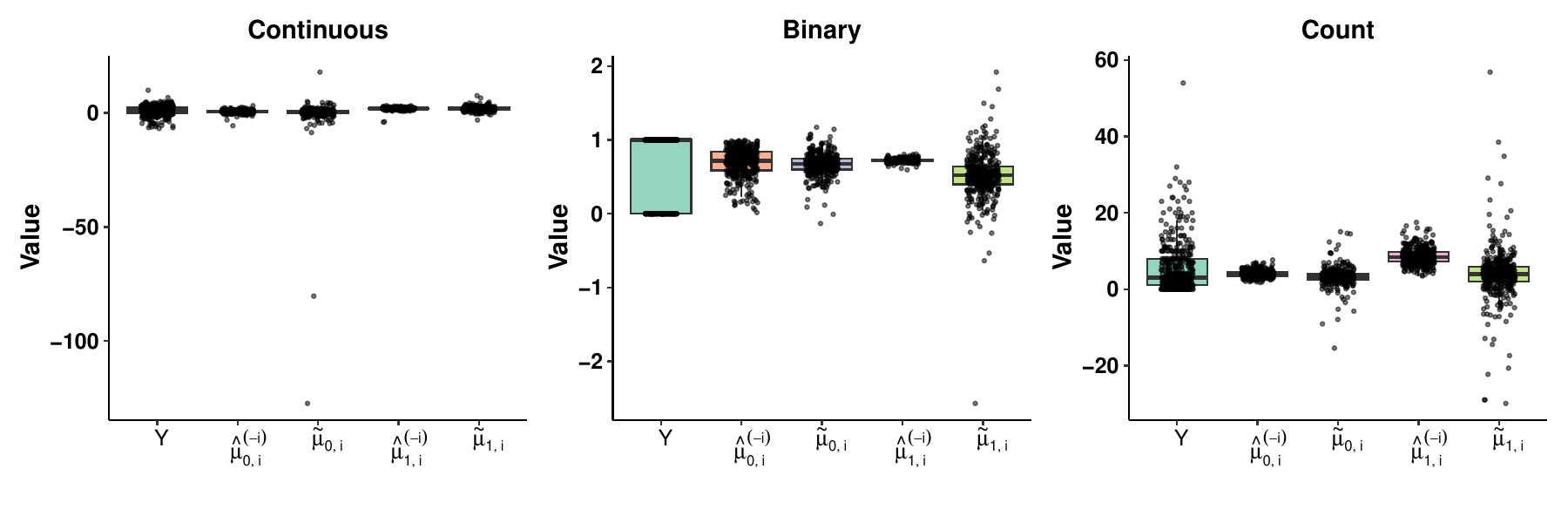}
    \caption{Comparison of predictions across two arms. $Y$ represents the observed outcome. $\{\hat{\mu}_{t,i}^{(-i)}\}_{i=1}^{n}, t\in\{0,1\}$ represents the  OR estimates by solving the jackknife score equation and $\{\tilde{\mu}_{t,i}\}_{i=1}^{n}$ represents the  predictions after linear calibration.} 
    \label{fig:example_extrem_ols}
\end{figure}

\bmsection{Proof of Main Theoretical Results}\label{app:var} 
\vspace*{12pt}

We first prove Proposition~\ref{prop:bias}. 
\begin{proof}[Proof of Proposition~\ref{prop:bias}]
The bias part of Proposition~\ref{prop:bias} is straightforward and is therefore omitted. Theorem~1 of \citet{lin2024worthwhile} can be used to show that when $p = o (n^{2 / 3})$ the asymptotic variance of the gOB estimator $\hat{\tau}_{\gob}$ equals $\sigma^{2}$ after being standardized by $\sqrt{n}$ (see Lemma~\ref{lem:MLE} for a rephrasing of this statement). Therefore, we are left to prove that the asymptotic variance of $\hat{\tau}_{\JASA}$ is the same as that of $\hat{\tau}_{\gob}$.

Without loss of generality, we only prove the results for $t = 1$, so we omit the dependence on $t$ in all relevant notation. Let $\bdelta \coloneqq \hat{\bbeta} - \bbeta$ and $\hat{\bdelta}_{i} \coloneqq \hat{\bbeta}^{(-i)} - \hat{\bbeta}$, where $\hat{\bbeta}$ is the solution to the score equation without the jackknife operation. By construction and Newton-Leibniz formula, the following identity holds for every $i \in [n]$:
\begin{align*}
& \frac{1}{n} \sum_{l = 1}^{n} \bm{Z}_{l} \int_{0}^{1} \{\nabla_{\bbeta} g (\bm{Z}_{l}^{\top} (\hat{\bbeta} + s \cdot \hat{\bdelta}_{i})) \diff s\}^{\top} \hat{\bdelta}_{i} = \frac{1}{n} \bm{Z}_{i} Y_{i}^{\dag} \\
& \Rightarrow \hat{\bdelta}_{i} = \Big\{ \underbrace{\int_{0}^{1} \frac{1}{n} \sum_{l = 1}^{n} \bm{Z}_{l} \{\nabla_{\bbeta} g (\bm{Z}_{l}^{\top} (\hat{\bbeta} + s \cdot \hat{\bdelta}_{i})) \diff s\}^{\top}}_{\eqqcolon \, \hat{\bm{\Xi}}} \Big\}^{-1} \frac{1}{n} \bm{Z}_{i} Y_{i}^{\dag}.
\end{align*}
In particular, we can simplify $\hat{\bm{\Xi}}$ as follows.
\begin{align*}
\hat{\bm{\Xi}} = \int_{0}^{1} \frac{1}{n} \sum_{l = 1}^{n} \bm{Z}_{l} \bm{Z}_{l}^{\top} g' (\bm{Z}_{l}^{\top} (\hat{\bbeta} + s \cdot \hat{\bdelta}_{i})) \diff s.
\end{align*}
By Lemma~\ref{lem:concentration},
\begin{align*}
\Vert \hat{\bdelta}_{i} \Vert_{2} \leq \Vert \hat{\bm{\Xi}}^{-1} \Vert_{\rm op} \frac{1}{n} \Vert \bm{Z}_{i} Y_{i}^{\dag} \Vert_{2} \lesssim \frac{1}{n} \Vert \bm{Z}_{i} Y_{i}^{\dag} \Vert_{2} = O_{\bbP} \left( \frac{\sqrt{p}}{n} \right),
\end{align*}
hence $\Vert \sqrt{n} \hat{\bdelta}_{i} \Vert_{2} = O_{\bbP} (\sqrt{p / n}) = o_{\bbP} (1)$. Then, in view of the decomposition,
\begin{align*}
\sqrt{n} (\hat{\tau}_{\JASA} - \tau) = \sqrt{n} (\hat{\tau}_{\JASA} - \hat{\tau}_{\gob}) + \sqrt{n} (\hat{\tau}_{\gob} - \tau),
\end{align*}
and Lemma~\ref{lem:MLE} on the asymptotic variance of $\sqrt{n} (\hat{\tau}_{\gob} - \tau)$, the desired conclusion holds if we can show that
\begin{align*}
\sqrt{n} (\hat{\tau}_{\JASA} - \hat{\tau}_{\gob}) = o_{\bbP} (1).
\end{align*}
To this end,
\begin{align*}
& \ \sqrt{n} (\hat{\tau}_{\JASA} - \hat{\tau}_{\gob}) \\
= & \ \frac{1}{n} \sum_{i = 1}^{n} \Big( 1 - \frac{T_{i}}{\pi_{1}} \Big) \sqrt{n} \{g (\bm{Z}_{i}^{\top} \hat{\bbeta}) - g (\bm{Z}_{i}^{\top} \hat{\bbeta}^{(-i)})\} \\
= & \ \frac{1}{n} \sum_{i = 1}^{n} \Big( 1 - \frac{T_{i}}{\pi_{1}} \Big) \sqrt{n} g' (\bm{Z}_{i}^{\top} \tilde{\bbeta}) \bm{Z}_{i} (\hat{\bbeta} - \hat{\bbeta}^{(-i)}),
\end{align*}
for some $\tilde{\bbeta}$ between $\hat{\bbeta}$ and $\hat{\bbeta}^{(-i)}$. Finally, by Cauchy-Schwarz inequality, we have $\sqrt{n} (\hat{\tau}_{\JASA} - \hat{\tau}_{\gob}) = o_{\bbP} (1)$ as desired.

\begin{lemma}
\label{lem:concentration}
Under Assumptions~\ref{as:regularity}--\ref{as:link}, the following hold: As $n \rightarrow \infty$ and $p = o (n)$,
\begin{enumerate}[label = (\roman*)]
\item $\Vert \bm{Z} Y^{\dag} \Vert_{2} = O_{\bbP} (\sqrt{p})$;

\item With probability converging to 1, the smallest eigenvalue of $\hat{\bm{\Xi}}$ is bounded below by some fixed constant $c > 0$.
\end{enumerate}
\end{lemma}

\begin{proof}
We first show part (a). By Assumption~\ref{as:regularity}
\begin{align*}
\bbE [\Vert \bm{Z} Y^{\dag} \Vert_{2}^{2}] = \bbE [Y^{\dag 2} \bm{Z}^{\top} \bm{Z}] = O (p).
\end{align*}
Therefore, by Markov's inequality, the desired result holds.

Next, we show part (b). By Assumption~\ref{as:link} on the link function $g(\cdot)$, $\hat{\bm{\Xi}}$ is semidefinite non-negative. Let $\hat{\bm{\Xi}} (\bbeta) =  \frac{1}{n} \sum_{l = 1}^{n} \bm{Z}_{l} \bm{Z}_{l}^{\top} g' (\bm{Z}_{l}^{\top}\bbeta)$ be the sample version of $\bm{\Xi} (\bbeta) = \bbE [\bm{Z}\bm{Z}^{\top} g' (\bm{Z}^{\top}\bbeta)]$. 
We also have by standard matrix concentration bound 
\begin{align*}
\Vert \hat{\bm{\Xi}} - \hat{\bm{\Xi}} (\bbeta) \Vert_{\rm op} & = \Big\Vert \int_{0}^{1} \frac{1}{n} \sum_{l = 1}^{n} \bm{Z}_{l} \bm{Z}_{l}^{\top} \{g' (\bm{Z}_{l}^{\top} (\hat{\bbeta} + s \cdot \hat{\bdelta}_{i})) - g' (\bm{Z}_{l}^{\top}\bbeta)\} \diff s \Big\Vert_{\rm op} \\
& \leq \sup_{\bar{\bbeta}} \left\Vert \frac{1}{n} \sum_{l = 1}^{n} \bm{Z}_{l} \bm{Z}_{l}^{\top} g'' (\bm{Z}_{l}^{\top} \bar{\bbeta}) \bm{Z}_{l}^{\top} \{\Vert \bdelta_{i} \Vert_{2} + \Vert \bdelta \Vert_{2}\} \right\Vert_{\rm op} \\
& = o_{\bbP} (1).
\end{align*}
\end{proof}

The following lemma, a direct corollary of Theorem~1 of \citet{lin2024worthwhile}, characterizes the asymptotic properties of $\hat{\tau}_{\gob}$, which is the same as $\hat{\tau}_{\JASA}$ except that jackknife is not employed.
\begin{lemma}
\label{lem:MLE}
Under the same assumptions as in Proposition~\ref{prop:bias}, when $p = o (n^{2 / 3})$,
\begin{align*}
\var \{\sqrt{n} (\hat{\tau}_{\gob} - \tau)\} = \sigma^{2} + o (1).
\end{align*}
\end{lemma}
\end{proof}

With Proposition~\ref{prop:bias}, the proof of Proposition~\ref{prop:cal} is straightforward because the calibration approach is equivalent to OLS adjustment as in \citet{ye2023toward} and \citet{ma2022regression}, except that the covariates are estimated OR. Since the dimension of the new covariates is at most 3 (including the intercept), the desired claims are direct consequences of well-established results of OLS adjusted estimators.

\bmsection{A Property of \texorpdfstring{$\hat{\tau}_{\adj, 2}^{\dag}$}{}}\label{sec:proof propositoin adj2c db}
\vspace*{12pt}

Note that the debiased estimator $\hat{\tau}_{\db}$ proposed by \citet{lu2025debiased} is designed under CRE and implicitly does not contain a constant term, which can be expressed as:
\begin{equation}
\label{debiased estimator}
\hat{\tau}_{\db} \coloneqq \left(\hat{\tau}_{1, \adj} + \frac{\hat{\pi}_{0}}{\hat{\pi}_{1}} \frac{1}{n} \sum_{i = 1}^{n} \frac{T_{i}}{\hat{\pi}_{1}} H_{i,i}^{\star} \left( Y_{i} - \bar{Y}_{1} \right)\right) - \left(\hat{\tau}_{0, \adj} + \frac{\hat{\pi}_{1}}{\hat{\pi}_{0}} \frac{1}{n} \sum_{i = 1}^{n} \frac{1-T_{i}}{\hat{\pi}_{0}} H_{i,i}^{\star} \left( Y_{i} - \bar{Y}_{0} \right)\right)
\end{equation}
where $\hat{\pi}_{t} = n_t / n$, $\hat{\tau}_{t,\adj} \coloneqq \frac{1}{n}\sum_{i=1}^{n}\frac{\mathbbm{1}\{T_i=t\}}{\hat{\pi}_t}\left(Y_i - \bX_i^\top \tilde{\bm{\beta}}_t \right)$ and $\tilde{\bm \beta}_{t}=\left(\bbX^\top \bbX\right)^{-1}\sum_{j=1}^{n}\bX_{j}\frac{\mathbbm{1}\{T_j=t\}\left(Y_j - \bar{Y}_{t}\right)}{\hat{\pi}_t}$. In the following, we show that $\hat{\tau}_{\db}$ (also $\hat{\tau}_{\adj, 2}^{\dag}$) is the same regardless of whether we use $\bbZ$ or $\bbX$. 

\begin{proposition}
Whether the design matrix  contains the intercept term ($\bbZ$) or not ($\bbX$), we have
\begin{equation}
  \hat{\tau}_{\adj, 2}^{\dag} \equiv \hat{\tau}_{\db}.
\end{equation}
\end{proposition}

\begin{proof}
Here, we focus on the treatment arm for conciseness; the control arm follows analogously. We first consider the design matrix with an intercept term.
Note that $\bbH = \bbH^\star + \frac{1}{n}\mathbf{1}_n\mathbf{1}_{n}^{\top}$, so we have
\begin{gather*}
    \bm{Z}_i^{\top}\hat{\bbeta}_{1} = \bm{Z}_i^{\top}\left(\bbZ^\top \bbZ\right)^{-1}\sum_{j=1}^{n}\bm{Z}_j \frac{T_j\left(Y_j - \bar{Y}_1\right)}{\hat{\pi}_1} = \sum_{j=1}^{n} H_{i,j} \frac{T_j\left(Y_j - \bar{Y}_1\right)}{\hat{\pi}_1} = \sum_{j=1}^{n}H_{i,j}^{\star}\frac{T_j\left(Y_j - \bar{Y}_1\right)}{\hat{\pi}_1} = \bm{X}_{i}^{\top}\tilde{\bbeta}_1, \\
    \sum_{i=1}^{n}H_{i,i}^{\star}\frac{T_i}{\hat{\pi}_1} \left(Y_i-\bar{Y}_{1}\right) = \sum_{i=1}^{n}H_{i,i}\frac{T_i}{\hat{\pi}_1} \left(Y_i-\bar{Y}_{1}\right).
\end{gather*}
Then
\begin{align*}
    \hat{\tau}_{1,\db} \coloneqq &\ \frac{1}{n}\sum_{i=1}^{n}\frac{T_i}{\hat{\pi}_1}\left(Y_i - \bm{X}_i^{\top}\tilde{\bbeta}_1\right) + \frac{\hat{\pi}_{0}}{\hat{\pi}_{1}} \frac{1}{n} \sum_{i = 1}^{n} \frac{T_{i}}{\hat{\pi}_{1}} H_{i,i}^{\star} \left( Y_{i} - \bar{Y}_{1} \right) \\
    = &\ \frac{1}{n}\sum_{i=1}^{n}\frac{T_i}{\hat{\pi}_1}\left(Y_i - \bm{Z}_i^{\top}\hat{\bbeta}_1\right) + \frac{\hat{\pi}_{0}}{\hat{\pi}_{1}} \frac{1}{n} \sum_{i = 1}^{n} \frac{T_{i}}{\hat{\pi}_{1}} H_{i,i} \left( Y_{i} - \bar{Y}_{1} \right). 
\end{align*}

For the first term, we have:
\begin{align*}
    &\ \frac{1}{n}\sum_{i=1}^{n}\frac{T_i}{\hat{\pi}_1}\left(Y_i-\bm{Z}_i^\top \hat{\bbeta}_1\right)  \\
    &\ = \frac{1}{n}\sum_{i=1}^{n}\frac{T_iY_i}{\hat{\pi}_1} + \frac{1}{n}\sum_{i=1}^{n}\left(1 - \frac{T_i}{\hat{\pi}_1}\right)\bm{Z}_i^{\top}\hat{\bbeta}_1 \\
    &\ = \frac{1}{n}\sum_{i=1}^{n}\frac{T_iY_i}{\hat{\pi}_1} + \frac{1}{n}\sum_{i=1}^{n}\left(1-\frac{T_i}{\hat{\pi}_1}\right)\sum_{j=1}H_{i, j}\frac{T_j\left(Y_j-\bar{Y}_1\right)}{\hat{\pi}_1} \\
    &\ = \frac{1}{n}\sum_{i=1}^{n}\frac{T_iY_i}{\hat{\pi}_1} + \frac{1}{n}\sum_{i\neq j}\left(1-\frac{T_i}{\hat{\pi}_1}\right)H_{i, j}\frac{T_j\left(Y_j - \bar{Y}_1\right)}{\hat{\pi}_1} + \frac{1}{n}\sum_{i=1}^{n}\left(1 - \frac{T_i}{\hat{\pi}_1} \right) H_{i, i} \frac{T_i\left(Y_i - \bar{Y}_1\right)}{\hat{\pi}_1} \\
    &\ = \hat{\tau}_{1,\adj, 2}^{\dag} + \frac{1}{n}\sum_{i=1}^{n}\left(\hat{\pi}_1 - 1\right)H_{i,i}\frac{T_i\left(Y_i - \bar{Y}_1\right)}{\hat{\pi}_1^{2}} \\
    &\ = \hat{\tau}_{1,\adj, 2}^{\dag} - \frac{\hat{\pi}_{0}}{\hat{\pi}_{1}} \frac{1}{n} \sum_{i = 1}^{n} \frac{T_{i}}{\hat{\pi}_{1}} H_{i,i} \left(Y_{i} - \bar{Y}_1 \right).
\end{align*}
Here, the second line follows from the property of linear regression that $\frac{1}{n}\sum_{i=1}^{n}\bm{Z}_i^\top \hat{\bbeta}_1 = \frac{1}{n}\sum_{i=1}^{n}\frac{T_i \left(Y_i - \bar{Y}_1\right)}{\hat{\pi}_1} = 0$. So, we have the following equation 
\begin{align*}
   \hat{\tau}_{1,\adj, 2}^{\dag} = \frac{1}{n}\sum_{i=1}^{n}\frac{T_i}{\hat{\pi}_1}\left(Y_i-\bm{Z}_i^\top \hat{\bbeta}_1\right) + \frac{\hat{\pi}_{0}}{\hat{\pi}_{1}} \frac{1}{n} \sum_{i = 1}^{n} \frac{T_{i}}{\hat{\pi}_{1}} H_{i,i} \left(Y_{i} - \bar{Y}_1 \right) = \hat{\tau}_{1,\db}.
\end{align*}
When the design matrix contains no intercept, one can derive the identity in a similar way from the above proof; see \citet{zhao2024hoif}.
\end{proof}

\bmsection{Variance estimators of \texorpdfstring{$\hat{\tau}_{\adj,2}$}{} and \texorpdfstring{$\hat{\tau}_{\adj,2}^{\dag}$}{}}
\label{var: adj2 adj2c}
\vspace*{12pt}

To estimate the variance of $\hat{\tau}_{\adj, 2}$ and $\hat{\tau}_{\adj, 2}^{\dag}$, we derive their
influence functions using Theorem 1 from \citet{bannick2025general}. For each $i\in[n]$ and $t\in\{0,1\}$, we have:
\begin{gather*}
    \bbE\left[\hat{\mu}_{t,i,\adj,2}^{(-i)}\right] 
    = \bbE\left[\sum_{j\neq i}H_{i,j}^{\star}\frac{\mathbbm{1}\{T_j=t\}Y_j}{\hat{\pi}_t}\right] + \bbE\left[\frac{1}{n}\sum_{j\neq i}\frac{\mathbbm{1}\{T_j=t\}Y_j}{\hat{\pi}_t}\right] = \bbE[Y(t)], \\
    \bbE\left[\hat{\mu}_{t,i,\adj,2}^{\dag(-i)}\right] 
    = \bbE\left[\sum_{j\neq i}H_{i,j}^{\star}\frac{\mathbbm{1}\{T_j=t\}\left(Y_j - \bar{Y}_{t}\right)}{\hat{\pi}_t}\right] + \bbE\left[\frac{1}{n}\sum_{j\neq i}\frac{\mathbbm{1}\{T_j=t\}\left(Y_j - \bar{Y}_{t}\right)}{\hat{\pi}_t}\right] = 0.
\end{gather*}
So their influence functions can be formulated as:
\begin{align*}
    \phi_{t,i,\adj, 2} =&\ \frac{\mathbbm{1}\{T_i=t\}}{\pi_t}Y_i(t) + \left(1 - \frac{\mathbbm{1}\{T_i=t\}}{\pi_t}\right){\mu}_{t,i,\adj,2}^{(-i)} - \tau_t, \\
    \phi_{t,i,\adj, 2}^{\dag} =&\ \frac{\mathbbm{1}\{T_i=t\}}{\pi_t}\left(Y_i(t) - \tau_t\right) + \left(1 - \frac{\mathbbm{1}\{T_i=t\}}{\pi_t}\right){\mu}_{t,i,\adj,2}^{\dag(-i)}. 
\end{align*}
Using these, we estimate the variances as follows:
\begin{align*}
    \hat{\sigma}_{\adj,2}^{2} &= \frac{1}{n^2}\sum_{i=1}^{n}\left\{
    \left( \frac{T_i}{\hat{\pi}_1} Y_i + \left( 1 - \frac{T_i}{\hat{\pi}_1} \right) \hat{\mu}_{1,i,\adj,2}^{(-i)}\right) -
    \left( \frac{1 - T_i}{1 - \hat{\pi}_1} Y_i + \left( 1 - \frac{1 - T_i}{1 - \hat{\pi}_1} \right) \hat{\mu}_{0,i,\adj, 2}^{(-i)} \right) - \hat{\tau}_{\adj,2} \right\}^2, \\
     \hat{\sigma}_{\adj,2}^{\dag2} &= \frac{1}{n^2} \sum_{i=1}^{n} \left\{ \left( \frac{T_i}{\hat{\pi}_1} \left(Y_i - \hat{\tau}_{1,\adj, 2}^{\dag}\right) + \left( 1 - \frac{T_i}{\hat{\pi}_1} \right) \hat{\mu}_{1,i,\adj, 2}^{\dag(-i)} \right) \right. \\
      &\quad - \left.
    \left( \frac{1 - T_i}{1 - \hat{\pi}_1} \left(Y_i - \hat{\tau}_{0,\adj,2}^{\dag}\right) + \left( 1 - \frac{1 - T_i}{1 - \hat{\pi}_1} \right) \hat{\mu}_{0,i,\adj,2}^{\dag(-i)} \right) \right\}^2.
\end{align*}
For the variance estimators of their calibrated counterparts, we use analogous expressions as given in \eqref{equ: var_jasa-cal}.

\begin{remark}
    When the design matrix does not include a constant term, $\bbE\left[\hat{\mu}_{t,i,\adj,2}^{(-i)}\right] = 0$. The corresponding influence function is therefore $\phi_{t,i,\adj, 2} = \frac{\mathbbm{1}\{T_i=t\}}{\pi_t}\left(Y_i(t) - \tau_t\right) + \left(1 - \frac{\mathbbm{1}\{T_i=t\}}{\pi_t}\right){\mu}_{t,i,\adj,2}^{(-i)}$, which differs from the scenario where the design matrix  includes a constant term. In contrast, the influence function  $\phi_{t,i,\adj, 2}^{\dag}$ remains the same regardless of whether an intercept is included in the design matrix.
\end{remark}

\bmsection{More Simulation Experiments}
\vspace*{12pt}

\bmsubsection{The empirical performance of optimizing the anti-derivative of the jackknife score}
\label{app:MLE}
We consider an alternative approach by optimizing the anti-derivative of the jackknife score for binary outcomes. The optimization is formulated as follows:
\begin{align*}
    \hat{\bbeta}_{t}^{(-i)} = \argmin_{\bbeta_{t}^{(-i)} \in \mathcal{B}_{t}} f\left(\bbeta_{t}^{(-i)}\right),
\end{align*}
where $f\left(\bbeta_{t}^{(-i)}\right) = \frac{1}{n-1}\sum_{j\neq i}\frac{\mathbbm{1}\{T_j=t\}}{\hat{\pi}_{t}}Y_j\bZ_j^{\top}\bbeta_{t}^{(-i)} - \frac{1}{n}\sum_{l=1}^{n}\log\left(1+\exp\left(\bZ_l^\top \bbeta_{t}^{(-i)}\right)\right)$.
We denote the resulting \JASA{} and \JASACal{} estimator as $\hat{\tau}_{\JASA}^{\prime}$ and $\hat{\tau}_{\JASACal}^{\prime}$, respectively.

Using the simulation setup for binary outcomes described in Section~\ref{sec:sim binary}, we compare $\hat{\tau}_{\JASA}^{\prime},\hat{\tau}_{\JASACal}^{\prime}$ with $\hat{\tau}_{\JASA}, \hat{\tau}_{\JASACal}$ estimated by optimizing \eqref{opt:general beta} (\formbeta{}). 
Here, we  present  results only for the scenario in which $\mu_{0}(\bx_\star)=\bx_{\star}^{\top}\bbeta$ and $\mu_{1}(\bx_{\star}) - \mu_{0}(\bx_{\star}) \equiv 1$. As shown in Figure \ref{fig:binomial_jasa_anti_deri}, both ways perform similarly, but $\hat{\tau}_{\JASA}$ and $\hat{\tau}_{\JASACal}$ yield smaller SD ratios than $\hat{\tau}_{\JASA}^{\prime}$ and $\hat{\tau}_{\JASACal}^{\prime}$, respectively. 
Similar trends are observed in other scenarios, though we omit those results for brevity.

\begin{figure}[htbp]
    \centering
    \includegraphics[width=0.95\linewidth]{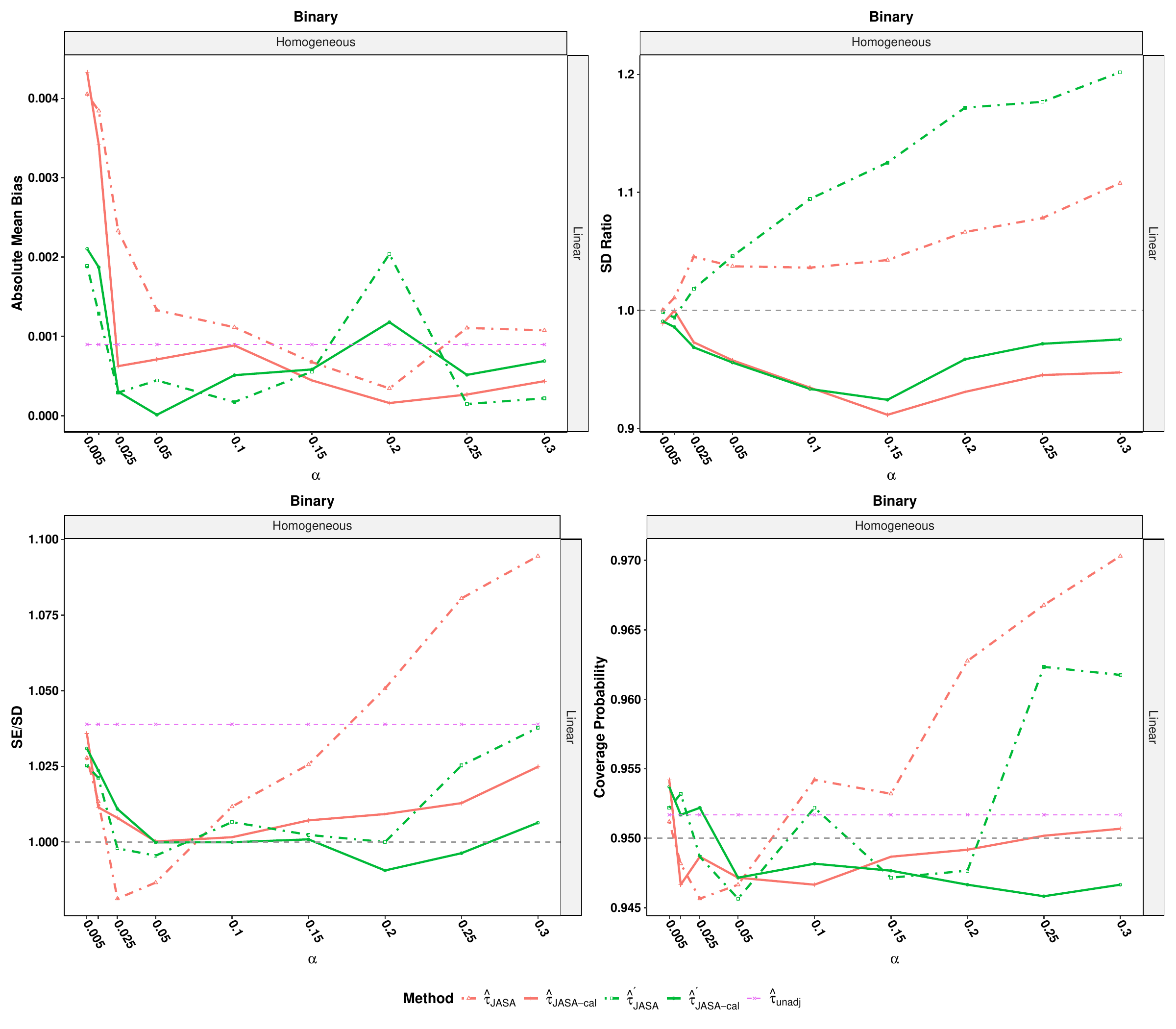}
    \caption{Performance comparison for \JASA{} and \JASACal{} by optimizing jackknife score ($\hat{\tau}_{\JASA}, \hat{\tau}_{\JASACal}$) or its anti-derivative ($\hat{\tau}_{\JASA}^{\prime},\hat{\tau}_{\JASACal}^{\prime}$) under binary outcomes.}
    \label{fig:binomial_jasa_anti_deri}
\end{figure}

\bmsubsection{Comparison of \JASA{} and \JASACal{} between \textbf{formulation (\texorpdfstring{$\bbeta$)}{}} and \textbf{formulation (\texorpdfstring{$\bm{\mu}$)}{}} under different outcome types}
In this section, we show the performance comparison for  \JASA{} and \JASACal{} by optimizing (\ref{opt:general beta}) (\formbeta{}) and (\ref{opt:general mu}) (\formmu{}) under three outcome types described in Section \ref{sec:sim}: continuous (Figure \ref{fig:gaussian_g0}), binary (Figure \ref{fig:binomial_g0}) and count (Figure \ref{fig:poisson_g0}). 

We also present the median computation time (minutes) of both approaches under binary (Figure \ref{fig:binomial_g0_time}) and count (Figure \ref{fig:poisson_g0_time}) outcome types. Computation time through optimizing \formbeta{} is much longer than that optimizing \formmu{}, particularly for binary outcomes when $\alpha > 0.1$.  
These computation-time estimates are based on parallelizing 2000 Monte Carlo replications, with leave-one-out optimization performed sequentially within each replication. This represents a conservative computational cost; in practice, users can parallelize the $n$  independent LOO subproblems with a single dataset, which has been implemented in our \texttt{HOIFCar} package.

\begin{figure}[!htbp]
    \centering    
    \includegraphics[width=0.95\linewidth, page=1]{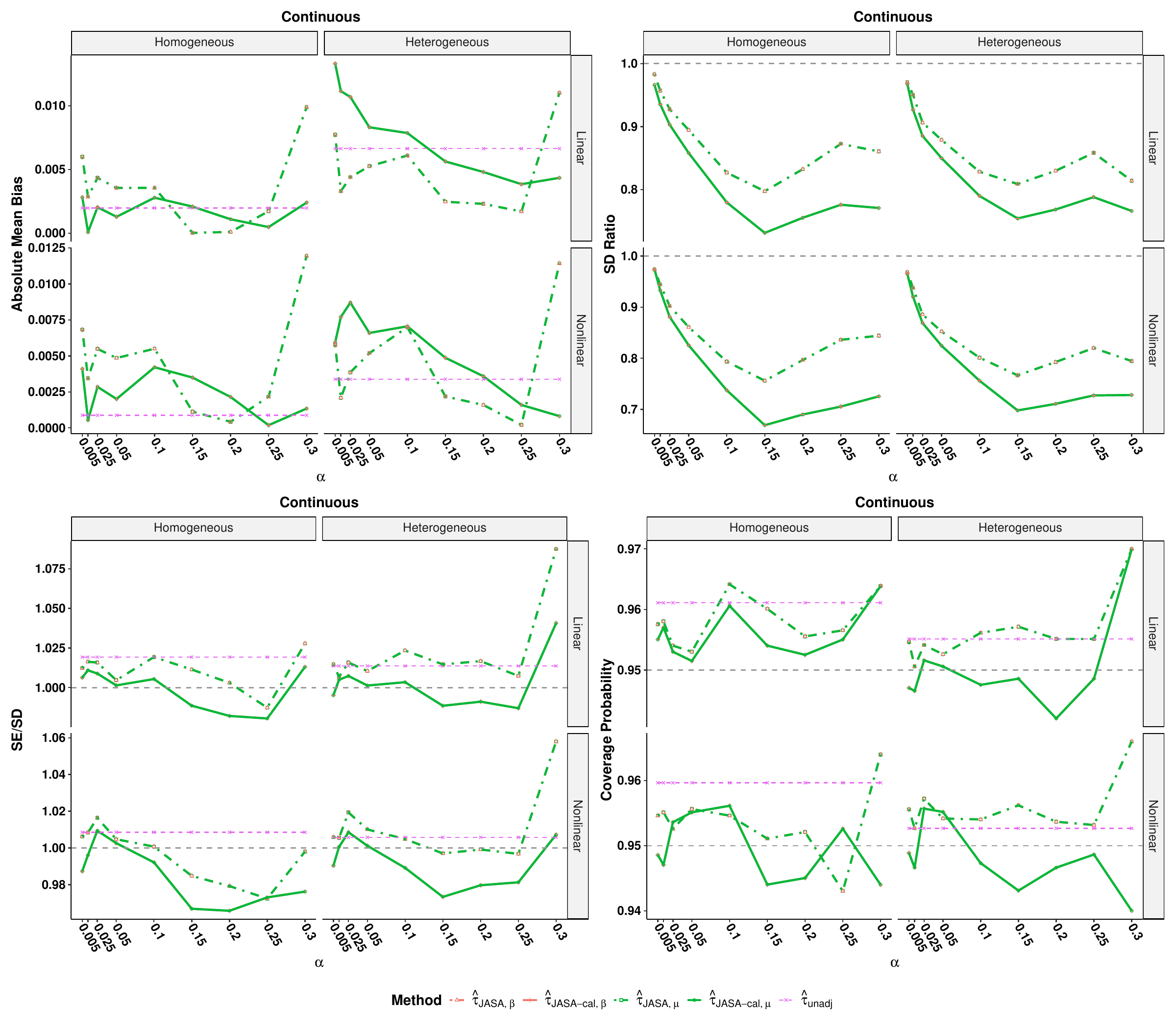}
    \caption{Performance comparison for \JASA{} and \JASACal{} by estimating $\beta$ or $\mu$ under continuous outcomes. Each panel  shows varied OR in the control group, CATE and the condition number $\alpha\coloneqq p/n$. Methods with calibration are shown in solid lines. }
    \label{fig:gaussian_g0}
\end{figure}

\begin{figure}[!htbp]
    \centering
    \includegraphics[width=0.95\linewidth, page=1]{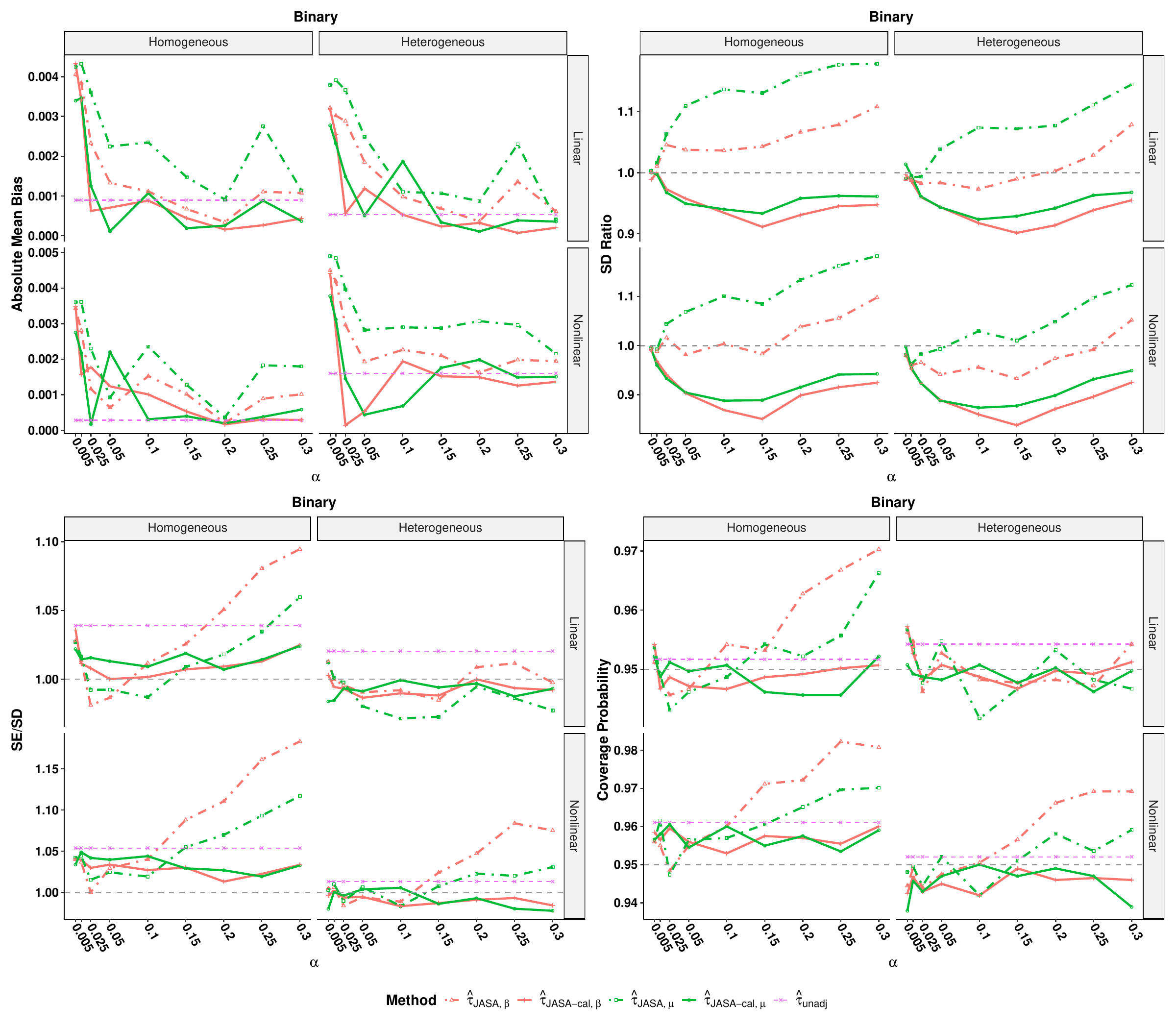}  
    \caption{Performance comparison for \JASA{} and \JASACal{} by estimating $\beta$ or $\mu$ under binary outcomes. Each panel  shows varied OR in the control group, CATE and the condition number $\alpha\coloneqq p/n$. Methods with calibration are shown in solid lines. }
    \label{fig:binomial_g0}
\end{figure}

\begin{figure}[!htbp]
    \centering   
    \includegraphics[width=0.95\linewidth, page=1]{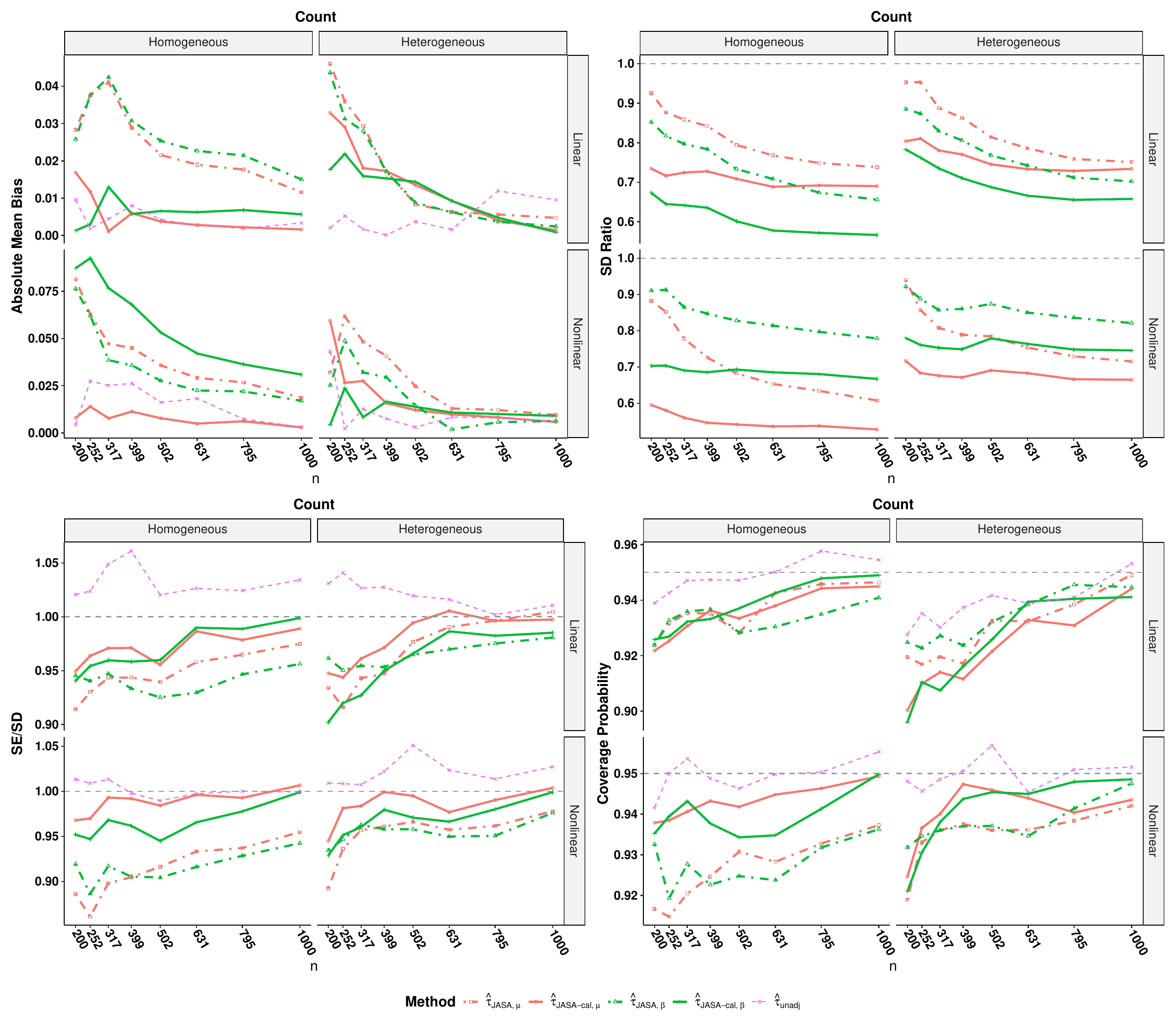}
    \caption{Performance comparison for \JASA{} and \JASACal{} by estimating $\beta$ or $\mu$ under count outcomes. Each panel  shows varied OR in the control group, CATE and the sample size $n$. Methods with calibration are shown in solid lines. }
    \label{fig:poisson_g0}
\end{figure}



\begin{figure}[!htbp]


    \centering
    \begin{subfigure}[b]{0.48\textwidth}
        \centering
        \includegraphics[width=\linewidth, page=2]{./Figures_SIM/Figure_E5_Time.pdf}
        \caption{Binary outcomes}
        \label{fig:binomial_g0_time}
    \end{subfigure}
    \hfill 
    \begin{subfigure}[b]{0.48\textwidth}
        \centering
        \includegraphics[width=\linewidth, page=2]{./Figures_SIM/Figure_E6_Time.pdf}
        \caption{Count outcomes}
        \label{fig:poisson_g0_time}
    \end{subfigure}
    \caption{Median computation time (minutes) for \JASACal{} by estimating $\beta$ or $\mu$ under different outcome types. (a) binary outcomes; (b) count outcomes. Panels show varied OR in the control group, CATE, and the condition number $\alpha\coloneqq p/n$ (for binary) or sample size $n$ (for count).}
    \label{fig:computation_time_comparison}
\end{figure}

\bmsubsection{Simulations for count outcomes}
\label{simu: count suppl infla}
We first consider one simulation scenario with count outcomes, which is similar to the setup for the continuous outcome described in the main text.
In each replication, we first generate the data matrix $\bbX_\star \in \bbR^{n\times p_\star}$ where $n=400, p_\star=60$. For  each row $\bmX_{\star,i}$, we first generate $\tilde{\bmX}_{\star,i} \sim t_3(0,\bSigma)$ with $\bSigma_{k,l}=0.1^{|k-l|},k,l \in [p_{\star}]$. The covariates are then truncated such that $\bmX_{\star,i}=\max\{\min\{\tilde{\bmX}_{\star,i},3\},-3\}$. 
The initial coefficients vector is defined as  $\bbeta = (-1)^j/j^{1/4},j\in[p_\star]$, and subsequently rescaled so that  $\|\bbeta\|_2 = 1$. 
The outcome model is nonlinear with $\mu_0(\bm{x}_\star)=\frac{1}{2}\mathrm{sign}\left(\bm{x}_\star^{\top} \bbeta\right)|\bm{x}_\star^{\top} \bbeta|^{\frac{1}{3}} + \cos(\bm{x}_\star^{\top} \bbeta) + \bm{x}_\star^{\top} \bbeta$ and the treatment effect is homogeneous with $\mu_1(\bm{x}_\star)=\mu_0(\bm{x}_\star) + 1$. We let the potential outcomes $Y_i(t)|\bX_{\star, i}=\bm{x}_\star \sim \mathrm{Poisson}(\exp(\tilde{\mu}_{t,i}))$, where $\tilde{\mu}_{t,i}=\min\{\mathbf{\mu}_{t,i}, 4\},t=0,1$. The observed outcome is $Y_i = T_iY_i(1) + (1-T_i)Y_i(0)$ with $\bbP(T_i=1)=1/3$ for $i\in [n]$.
The dimension of observed covariates $\bm{X}$ is set to $p\coloneqq \lceil n\cdot \alpha \rceil$, where $\alpha\in \{0.005,0.0125, 0.025, 0.05, 0.1, \mathbf{0.15}, 0.2, 0.25, 0.3\}$.  For $\alpha > 0.15$, additional noise variables are included, sampled independently from a $t(3)$ distribution and truncated to $[-3, 3]$.

We begin by comparing the performances of \JASA{} and \JASACal{} through estimation of  $\mu$ or $\beta$, as illustrated in Figure \ref{fig:appen_poisson_g0}. For very small values of $\alpha$, \JASACal{} fails to control bias when estimated using either $\mu$ or $\beta$, although the $\beta$-based estimator performs slightly better. As $\alpha$ increases, both estimators stabilize and their biases and other metrics become nearly identical.
In Figure \ref{fig:appen_poisson_anhecova}, we compare \JASA{}, \JASACal{} (estimated via $\mu$), and $\hat{\tau}_{\OLS}$. 
When $\alpha$ is small, $\hat{\tau}_{\OLS}$ maintains a small bias. However, it fails to control  bias as $\alpha$ continues to increase and the SD is underestimated.  
On the other hand, our proposed estimator $\hat{\tau}_{\JASACal}$ exhibits a higher bias than $\hat{\tau}_{\JASA}$ when $\alpha$ is very small, with bias inflation occurring after calibration. However, the problem diminishes as $\alpha$ increases. Moreover, the SD ratio of $\hat{\tau}_{\JASACal}$ decreases as $\alpha$ increases to 0.15, remaining substantially smaller than that of $\hat{\tau}_{\JASA}$. This further demonstrates the benefits of calibration when $\alpha$ is not very small.

\begin{figure}[htbp]
    \centering
    \includegraphics[width=0.95\linewidth]{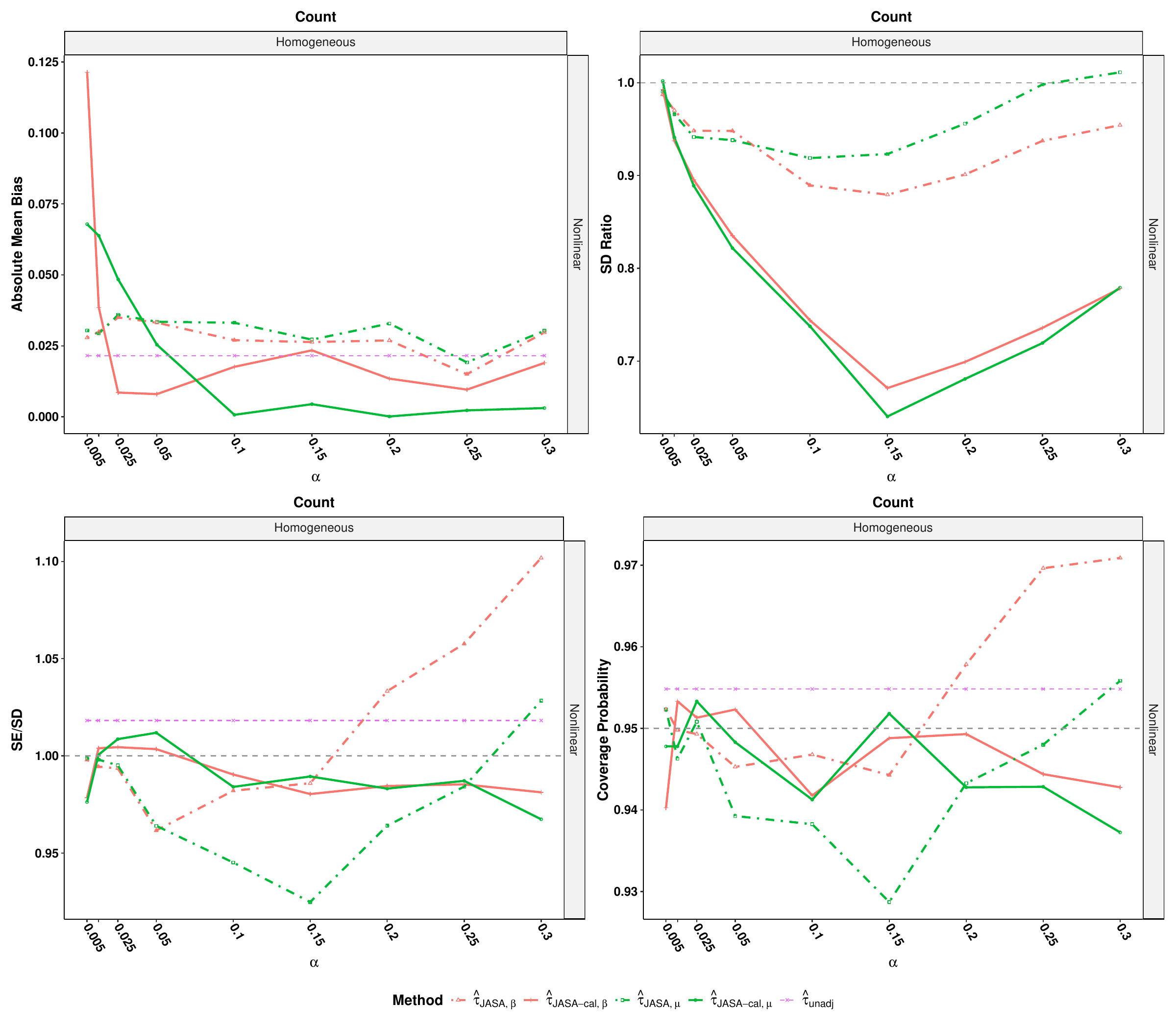}
    \caption{Performance comparison for \JASA{}, \JASACal{} by estimating $\bbeta$ and $\bm \mu$ under count outcomes. The condition number $\alpha\coloneqq p/n$. Methods with calibration are shown in solid lines.}
    \label{fig:appen_poisson_g0}
\end{figure}

\begin{figure}[htbp]
    \centering
    \includegraphics[width=0.95\linewidth]{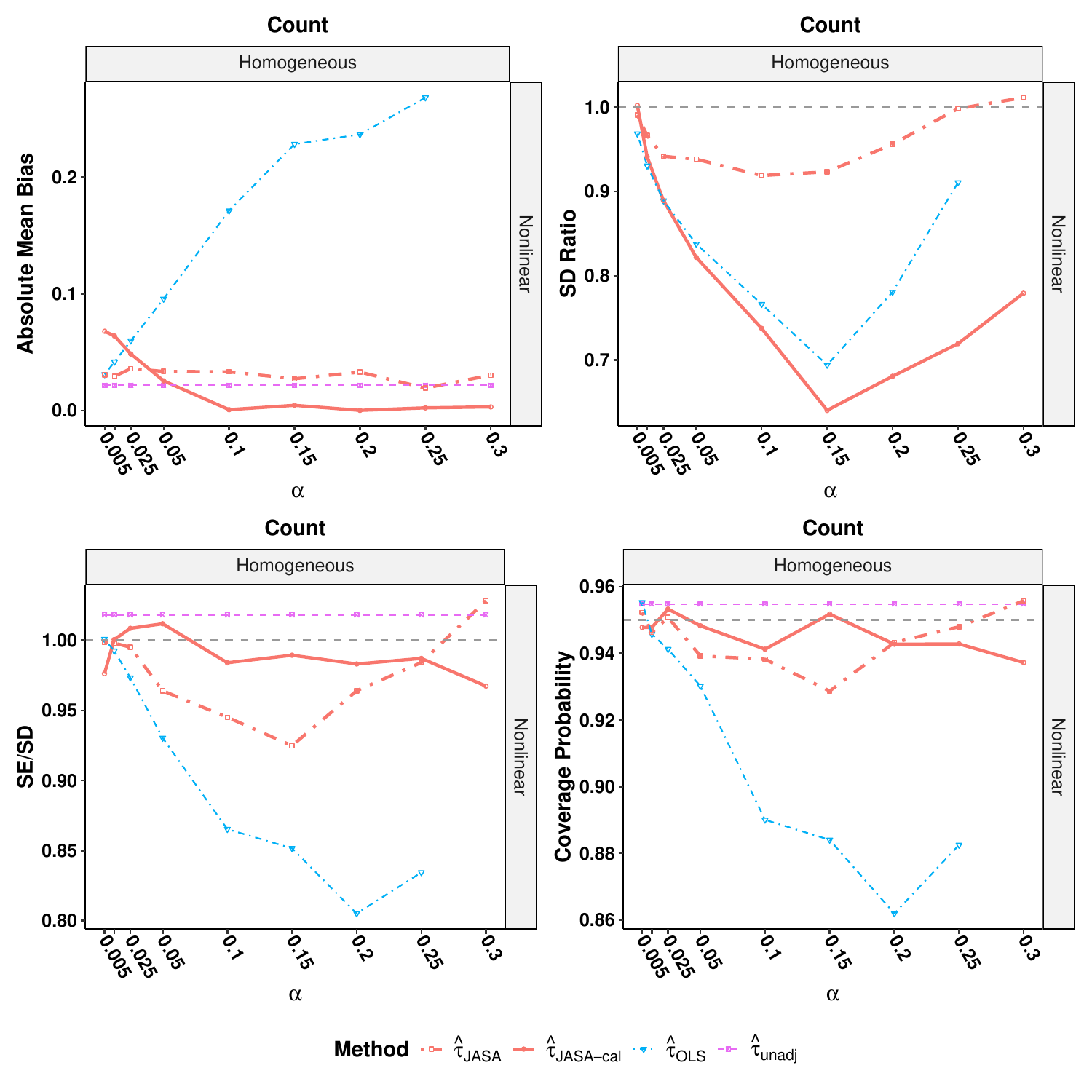}
    \caption{Performance comparison for  \JASA{}, \JASACal{}, $\hat{\tau}_{\OLS}$ and unadjusted estimator under count outcomes. The condition number $\alpha\coloneqq p/n$. Methods with calibration are shown in solid lines. Here \JASA{} and \JASACal{} are estimated by optimizing $\mu$.}
    \label{fig:appen_poisson_anhecova}
\end{figure}

Furthermore, we examine an alternative simulation setup by \citet{cohen2024no}. Specifically, we first generate a data matrix $\bX \in \bbR^{n\times 1}$ with each entry following a uniform distribution on $[-5, 5]$. Then we generate the potential outcomes as follows
\begin{align*}
    Y_i(0) \sim \text{Poisson}\left(72 - 0.45\exp(x_i)\right), \quad Y_i(1) \sim \text{Poisson}\left(\exp(x_i)\right).
\end{align*}
The treatment assignment follows a \textrm{Bernoulli} distribution with $\bbP(T_i=1)=3/5$ for $i\in [n]$. We vary the sample size $n \in \{300,500,\ldots,1500\}$ and perform $K=2000$ Monte Carlo replications. The true ATE $\tau$ is  approximated via sample averaging with an increased sample size of $n=10^7$. We compare the performance of \JASA{}, \JASACal{} (estimating through $\mu$ or $\beta$, estimating using $\hat{\pi}$ or $\pi$) with  LOO estimator $\hat{\tau}_{\loo}$, its calibrated version $\hat{\tau}_{\loocal}$ \cite{cohen2024no} and the alternative version $\hat{\tau}_{\loocalout}$ defined in Section \ref{sec:sim}. 

Results are shown in Figure \ref{fig:no harm vary n}. As expected, the calibrated estimators $\hat{\tau}_{\gbcal},\hat{\tau}_{\loocal}$ exhibit improved SD ratio below 1 after calibration, compared to their uncalibrated counterparts $\hat{\tau}_{\gob}, \hat{\tau}_{\loo}$. All of these estimators maintain negligible bias compared to $\hat{\tau}_{\unadj}$, even when $n$ is relatively small. However, $\hat{\tau}_{\JASA}$ displays higher bias than the others when  $\hat{\pi}$ is used for estimation, and its bias converges to zero at a  slower rate. In contrast, the bias of $\hat{\tau}_{\JASA}$ estimated with the true $\pi$ decreases to zero rapidly. Notably, for $\hat{\tau}_{\JASACal}$, not only does the bias converge to zero quickly, but the SD ratio is also smaller than 1 when $n \geq 500$, regardless of how the estimator is computed.

\begin{figure}[htbp] 
    \centering
    \includegraphics[width=0.95\linewidth]{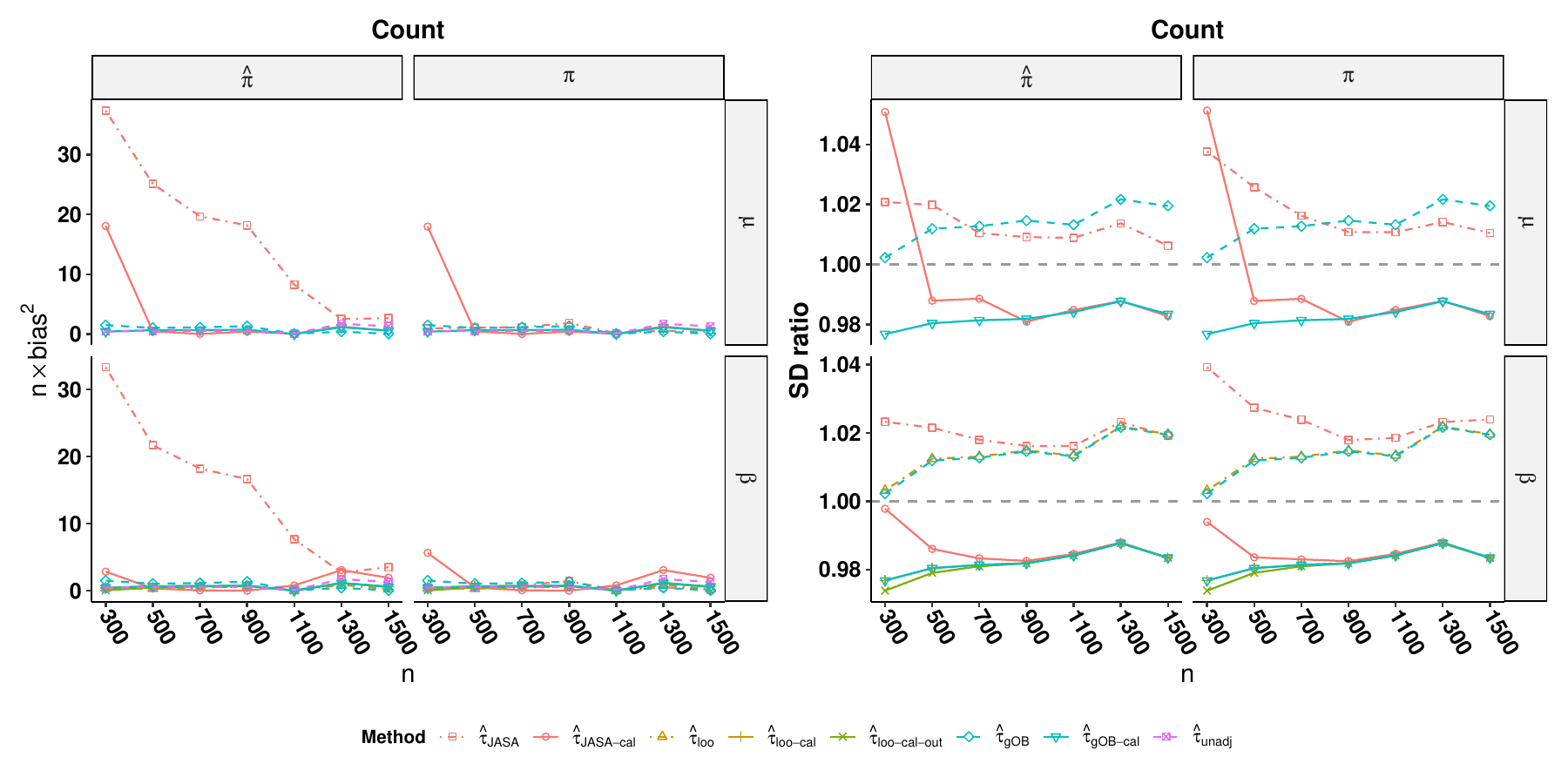}
    \caption{Bias and standard deviation (SD) ratio for the \JASA{} and \JASACal{} estimators, alongside other estimators in a count outcomes scenario. The panels labeled $\hat{\pi}$ and $\pi$ indicate whether the estimators were computed using the empirical assignment $\hat{\pi}$ or the true value $\pi$, respectively. Similarly, the panels labeled $\mu$ and $\beta$ denote estimation via the outcome model parameter $\mu$ or the calibration parameter $\beta$.}
    \label{fig:no harm vary n}
\end{figure}

\bmsubsection{The distribution of \texorpdfstring{$\hat{\tau}_{\JASA}$}{} and \texorpdfstring{$\hat{\tau}_{\JASACal}$}{}}
\label{app:qqplots}

To assess the asymptotic normality of $\hat{\tau}_{\JASA}$ and $\hat{\tau}_{\JASACal}$, we choose the simulation settings in the cases of continuous and binary outcomes and $\alpha = 0.15$, as described in Sections~\ref{sec:sim continuous} and~\ref{sec:sim binary}. We present QQ plots of these estimators in Figure~\ref{fig:qqplot}.
\begin{figure}[htbp]
    \centering
    \includegraphics[width=0.95\linewidth]{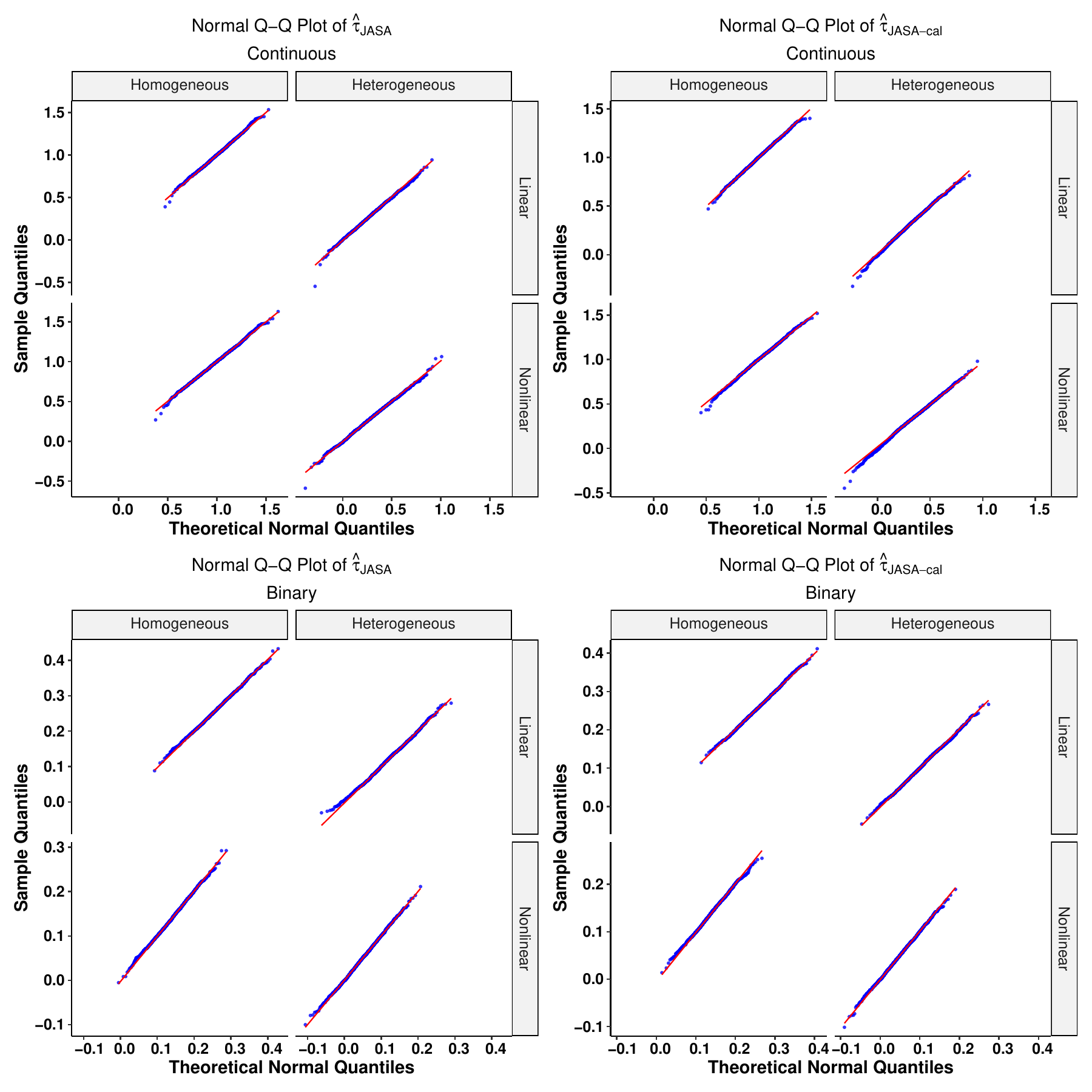}
    \caption{QQ plots of $\hat{\tau}_{\JASA}$ and $\hat{\tau}_{\JASACal}$ in the cases of continuous and binary outcomes.}
    \label{fig:qqplot}
\end{figure}




\end{document}